\documentclass[pra,onecolumn,unsortedaddress,superscriptaddress,notitlepage]{revtex4-1}

\usepackage[utf8]{inputenc}
\usepackage{amsthm}
\usepackage{amsmath,amssymb}
\usepackage{fullpage}
\usepackage{color}
\usepackage{tikz}

\usepackage{thmtools}
\usepackage{thm-restate}

\usepackage{hyperref} % has to come before cleveref, otherwise weird errors
\usepackage[capitalize]{cleveref} % after hyperref 

\declaretheorem[name=Theorem,numberwithin=section]{theorem}
\declaretheorem[name=Proposition,numberwithin=section]{proposition}

\declaretheorem[name=Definition,style=definition]{definitionRE}
\declaretheorem[name=Lemma,numberwithin=section]{lemma}
\declaretheorem[name=Definition,style=definition,numberwithin=section]{definition}
\declaretheorem[name=Corollary,numberwithin=section]{corollary}
\declaretheorem[name=Remark,style=remark,numberwithin=section]{remark}

\declaretheorem[name=Example,style=definition,numberwithin=section]{example}
\newcommand{\tr}{\operatorname{Tr}}

\newcommand{\bra}[1]{\langle #1 \vert}
\newcommand{\ket}[1]{\vert #1 \rangle}

\newcommand{\innerCG}[2]{ \left\langle #1\left| \right.  #2 \right\rangle }

\newcommand{\LL}[0]{ \mathcal{L}}
\newcommand{\RR}[0]{\mathcal{R} }
\newcommand{\CC}[0]{\Theta}
\newcommand{\II}[0]{\mathbb{I}}
\newcommand{\HN}[0]{\mathcal{H}^{\otimes N}}

\newcommand{\myinv}[1]{#1^{\scalebox{0.9}[1.0]{-}1}}

\tikzset{
	mpstensor/.pic = {
		\draw (-0.5,0)--(0.5,0);
		\draw (0,0)--(0,0.5);
		\draw[fill=white] (-0.25,-0.25) rectangle (0.25,0.25);
		\node at (0,0) {\tikzpictext};
	}
}

\tikzset{
	mpstensorMED/.pic = {
		\draw (-0.5,0)--(0.5,0);
		\draw (0,0)--(0,0.5);
		\draw[fill=white] (-0.3,-0.3) rectangle (0.3,0.3);
		\node at (0,0) {\tikzpictext};
	}
}

\tikzset{
	mpstensorBIG/.pic = {
		\draw (-0.65,0)--(0.65,0);
		\draw (0,0)--(0,0.65);
		\draw[fill=white] (-0.35,-0.35) rectangle (0.35,0.35);
		\node at (0,0) {\tikzpictext};
	}
}

\tikzset{
	horizontalmatrix/.pic = {
		\draw (-0.5,0)--(0.5,0);
		\draw[fill=white] (-0.25,-0.25) rectangle (0.25,0.25);
		\node at (0,0) {\tikzpictext};
	}
}
\tikzset{
	horizontalmatrixMED/.pic = {
		\draw (-0.5,0)--(0.5,0);
		\draw[fill=white] (-0.3,-0.3) rectangle (0.3,0.3);
		\node at (0,0) {\tikzpictext};
	}
}

\tikzset{
	horizontalmatrixBIG/.pic = {
		\draw (-0.65,0)--(0.65,0);
		\draw[fill=white] (-0.35,-0.35) rectangle (0.35,0.35);
		\node at (0,0) {\tikzpictext};
	}
}

\tikzset{
	horizontalmatrixWIDE/.pic = {
		\draw (-0.8,0)--(0.8,0);
		\draw[fill=white] (-0.5,-0.35) rectangle (0.5,0.35);
		\node at (0,0) {\tikzpictext};
	}
}
\tikzset{
	horizontalmatrixWIDE1/.pic = {
		\draw (-1,0)--(1,0);
		\draw[fill=white] (-0.7,-0.35) rectangle (0.7,0.35);
		\node at (0,0) {\tikzpictext};
	}
}

\tikzset{
	verticalmatrix/.pic = {
		\draw (0,-0.5)--(0,0.5);
		\draw[fill=white] (-0.25,-0.25) rectangle (0.25,0.25);
		\node at (0,0) {\tikzpictext};
	}
}

\tikzset{
	verticalmatrixMED/.pic = {
		\draw (0,-0.5)--(0,0.5);
		\draw[fill=white] (-0.3,-0.3) rectangle (0.3,0.3);
		\node at (0,0) {\tikzpictext};
	}
}

\tikzset{
	verticalmatrixBIG/.pic = {
		\draw (0,-0.65)--(0,0.65);
		\draw[fill=white] (-0.35,-0.35) rectangle (0.35,0.35);
		\node at (0,0) {\tikzpictext};
	}
}

\tikzset{
	verticalmatrixSuperWIDE/.pic = {
		\draw (0,-0.65)--(0,0.65);
		\draw[fill=white] (-0.6,-0.35) rectangle (0.6,0.35);
		\node at (0,0) {\tikzpictext};
	}
}

\tikzset{
	verticalmatrixWIDE/.pic = {
		\draw (0,-0.65)--(0,0.65);
		\draw[fill=white] (-0.5,-0.35) rectangle (0.5,0.35);
		\node at (0,0) {\tikzpictext};
	}
}

\tikzset{
	verticalmatrixSuperWIDE/.pic = {
		\draw (0,-0.65)--(0,0.65);
		\draw[fill=white] (-0.8,-0.35) rectangle (0.8,0.35);
		\node at (0,0) {\tikzpictext};
	}
}

\tikzset{
	pics/projector/.style 2 args = {code={
			\pic[pic text = #1] at (-0.75,0) {gauge};
			\pic[pic text = #2] at (0.75,0) {gauge};
			\pic[pic text = #1] at (0,0.75) {gaugevert};
			\pic[pic text = #2] at (0,-0.75) {gaugevert};
		}}
	}
	
	\tikzset{
		pepsoverlap/.pic = {
			\draw (-0.5,0.05)--(0.5,0.05);
			\draw (0.05,-0.5)--(0.05,0.5);
			\draw (-0.5,-0.05)--(0.5,-0.05);
			\draw (-0.05,-0.5)--(-0.05,0.5);
			\draw[fill=white] (-0.25,-0.25) rectangle (0.25,0.25);    
		}
	}
	
	\tikzset{
		every picture/.style = {
			baseline={([yshift=-.5ex]current bounding box.center)}, 
			scale=1.2,
			transform shape,
			font=\scriptsize
		}
	}

\begin{document}
  
%  
%  \begin{frontmatter}
%  	
  		\title{Classification of Matrix Product States with a Local (Gauge) Symmetry}
  		
  	\date{\today}
  	\author{Ilya Kull}
  	\author{Andras Molnar}
  	\author{Erez Zohar}
   	\author{J. Ignacio Cirac}
  	\address{Max-Planck-Institut f\"ur Quantenoptik, Hans-Kopfermann-Stra\ss e 1, 85748 Garching, Germany.}
  	
  	\begin{abstract}
  	Matrix Product States (MPS) are a particular type of one dimensional tensor network states, that have been  applied to the study of numerous quantum many body problems. One of their key features is the possibility to describe and encode symmetries on the level of a single building block (tensor), and hence they provide a natural playground for the study of symmetric systems. In particular, recent works have proposed to use MPS (and higher dimensional tensor networks) for the study of systems with local symmetry that appear in the context of gauge theories. In this work we classify MPS which exhibit local invariance under arbitrary gauge groups. We study the respective tensors and their structure, revealing known constructions that follow known gauging procedures, as well as different, other types of possible gauge invariant states.
  	\end{abstract}
  	
%  	\begin{keyword}
%  		Tensor network states \sep Matrix product states \sep Lattice gauge theories
%  	\end{keyword}
 
\maketitle

\section{Introduction}
Gauge theories play a  paramount role in modern physics.   Through  the gauge principle, the   theories describing the fundamental interactions in the standard model of particle physics are obtained  by lifting the global symmetries of the interaction-free matter theories to be  local symmetries, minimally coupled \cite{Peskin:257493} to a gauge field.  Moreover, they also emerge as effective low-energy descriptions in several condensed matter scenarios \cite{Fradkin2013}. 
Historically, the gauging procedure  was first conceived  as a transformation of a Lagrangian or Hamiltonian describing a physical  system; however, it can be performed on the level of quantum states as well, irrespective of dynamics associated to a specific theory.

In spite of their central role in the standard model, non-Abelian gauge theories still involve puzzles to be solved. Their complete understanding still poses a significant challenge due to non-perturbative phenomena (e.g.\ low energy QCD). Among the various approaches proposed to tackle the strongly coupled regime, a particularly  general and successful one is lattice gauge theory \cite{Wilson:1974sk}.  
 Monte Carlo sampling of Wilson's Euclidean lattice version of gauge theories    has so far been  the most successful method of numerical   simulation, 
 nevertheless, it suffers from its own drawbacks.
The  sign problem  \cite{Troyer:2004ge} prevents application to systems with large fermionic densities,  and the use of Euclidean time does not allow to  study real time evolution and  non-equilibrium phenomena in general scenarios.
In order to describe real-time evolution of such theories, one is forced to abandon the Monte Carlo approach, and search for other methods. In this context, the Hamiltonian formulation of Kogut and Susskind \cite{Kogut1975} has   been receiving  renewed interest,  with two recent approaches coming from the  quantum information and quantum optics community: quantum simulation, using optical, atomic or solid-state systems  \cite{Zohar:2015hwa,Wiese:2013uua},
and tensor network states.
 
The representation of quantum many-body states as tensor networks  is connected to  White's  density-matrix renormalization group \cite{White1992}, and  in the case of one dimensional spin lattices is known as matrix product states (MPS) \cite{Verstraete2008}. 
Among many useful properties of tensor networks, one which makes them well suited to the description of states with symmetries, is the ability to encode the symmetry on the level of a single  tensor (or a few) describing the state. In the case of  global symmetries, both for   MPS and for certain classes of PEPS in 2D (Projected Entangled Pair States - the generalization of MPS to higher dimensional lattices), the relation between the symmetry of the state and the properties of the tensor is well understood  \cite{Perez-Garcia2010}.
Tensor networks studies of lattice gauge theories have so far included  numerical works (e.g., mass spectra, thermal states, real time dynamics and string breaking, phase diagrams etc.\ for the Schwinger model and others) 
\cite{ Banuls2013, Shimizu2014, Buyens2014, Silvi2014, 2014arXiv1411.0020B,PhysRevLett.112.201601, Saito:2014bda,Kuhn2015, PhysRevD.92.034519,PhysRevD.93.094512, PhysRevX.6.011023, Buyens:2016hhu, Silvi:2016cas,PhysRevD.93.085012, Banuls:2016hhv, PhysRevD.94.085018,Saito:2015ryj, PhysRevLett.118.071601, 2017EPJWC13704001B,Banuls:2017ena}, 
furthermore, several theoretical formulations of classes of gauge invariant tensor network states have been proposed \cite{Tagliacozzo2014,Haegeman:2014maa,Zohar:2015eda,Zohar:2015jnb,Zohar:2016wcf}.
In all  of the latter the construction method follows the ones common to conventional gauge theory formulations: symmetric tensors are used to describe the matter degree of freedom, and later on a gauge field degree of freedom is added, or, alternatively - a pure gauge field theory  is considered. While the usefulness of tensor networks in  lattice gauge theories has certainly been demonstrated by the above mentioned works, so far there were few attempts (e.g.\ \cite{Buyens2014}) to generally classify tensor network states with local symmetry. 
  
In this paper,
starting   from the assumption of  a local symmetry, we find necessary and sufficient conditions to be satisfied by the tensors encoding a MPS. Similar work was done in \cite{Buyens2014} for MPS with local U(1) symmetry and  with open boundary conditions. We focus on  translation-invariant MPS, and deal with arbitrary finite or compact Lie groups. Clearly, one could come up with arbitrarily complicated  constructions of states with a local symmetry (e.g.\ by using many kinds of symmetric tensors). Our analysis is therefore limited to three physically meaningful settings corresponding to:  states describing {matter},  {pure gauge field} states  and states of both matter and gauge field. In our analysis the matter degrees of freedom are represented by  ``spins''; this could in principle be extended to fermionic systems, and in particular to Majorana fermions.

For states describing only matter we find that local symmetries can only be  trivial,   and show how to gauge such states by adding another degree of freedom.
When investigating pure gauge states we show that local symmetry in MPS requires a specific structure of the Hilbert space describing the gauge field degree of freedom.
In  Wilson's lattice gauge theories, in order to obtain minimal coupling in a continuum limit, the gauge field degree of freedom  is set as a group element in the same representation as the one acting on the matter  \cite{Wilson:1974sk}. In the Hamiltonian formulation, the corresponding Hilbert space is isomorphic to $L^2(G)$, equipped with the left and right regular representations \cite{Zohar:2014qma}, and  is referred to by  Kogut and Susskind  as   ``the rigid rotator'' (in the $SU(2)$ case) \cite{Kogut1975}. The structure that  we find for the gauge field Hilbert space  is more general and contains   the rotator-like space introduced by Kogut and Susskind as a particular case.

In the matter and gauge field setting  we show that, similar to the case of MPS with a global symmetry, the tensor describing the matter degree of freedom is  a (generalized) vector operator, and its structure is therefore determined by the Wigner-Eckart theorem; the gauge field tensor's structure is simpler: it is an intertwining map that translates the physical symmetry operators into a group action on the virtual (bond) spaces. This is a one dimensional version of the construction principle used in \cite{Zohar:2015jnb} - our work describes the sense in which this construction method is unique and the available structural and parametric freedom in choosing the tensors. However, the structure we derive allows for more general gauge invariant MPS, namely, ones that do not arise as a result of gauging a global symmetry or coupling matter to a pure gauge field. We construct examples of such states: while possessing a local symmetry when coupled to each other, the matter and gauge field degrees of freedom do not retain their  individual  symmetries when separated. Finally, we discuss mutual implications between the condition of {local symmetry of the pure gauge field} and the condition of {global symmetry of the matter}  when the two can be coupled to each other to produce a MPS with  local symmetry. 

The paper is organized as follows. In \cref{sec:Formalism} we introduce the basic notation and define the settings which will be investigated in subsequent sections. \cref{sec:RsultsOverview} presents a summary of our results. In \cref{sec:GlobSymm} we review the known classification of MPV with a global symmetry.  In   \Cref{sec:detProofs} we derive the proofs of the   stated results. 

\section{Formalism} \label{sec:Formalism}
 In this section we    introduce the MPS formalism and the notation used in this paper.
 We   present the different settings of states and symmetries that will be the focus of investigation in subsequent sections. 
 We    motivate the choices of those settings, and relate them to physical theories.
 This section covers all the definitions and the essential background needed in order for our results to be stated in \cref{sec:RsultsOverview}.

\subsection{Matrix product vectors}
	 We   consider matrix product vectors (MPV) rather than states (MPS). The   distinction is emphasized because MPV can refer to unnormalized MPS as well as to matrix product operators, to which our results can also be applied. Moreover, in the following we shall define symmetries  in terms of equalities between vectors and not states, i.e.\ we shall not  allow a phase difference. For a comprehensive introduction to MPS we refer the reader to \cite{Perez-Garcia2007, Verstraete2008,Cirac2017}. In the following we shall review the basic definitions, and quote essential results.
	
	Let $\mathcal{H}$ be a $d$-dimensional Hilbert space. 
	A matrix product vector (MPV) is a vector $\ket{\psi_A^N} \in \HN$ given by
	\begin{equation} \label{eq:MPVdef}
		\ket{\psi_A^N} = \sum_{\{i\}}\tr \left( A^{i_1}A^{i_2}\ldots A^{i_N} \right) \ket{i_1 i_2 \ldots i_N} \ ,
	\end{equation}
	where $\{A^i | i= 1,\ldots,d\}$ are $D\times D $ matrices and $\{\ket{i} | i= 1,\ldots,d\}$ is an orthonormal basis in $\mathcal{H}$. The dimension of the matrices - $D$ - is called the bond dimension of $A$.  We say that the tensor $A$, which consists of the matrices $A^i$, generates the MPV  $\ket{\psi_A^N} $; in fact, it generates a family of vectors: $\left\{ \ket{\psi_A^N } | N\in \mathbb{N} \right\}$. We refer to the entire family of vectors as the MPV generated by $A$. 
	
	A MPV of this form is translationally invariant (TI). It is possible to describe vectors that are not TI in a similar way, with a different tensor associated with each tensor copy of $\mathcal{H}$. Throughout this paper we consider only TI-MPV.
	
	In order to avoid cumbersome notation involving many indices, we will use the graphical notation commonly used in tensor networks. Each tensor is denoted by a rectangle with lines connected to it. Each line corresponds to an index of the tensor.
	For example, the tensor $A$ generating the MPV above is represented as:
	
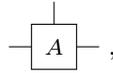
\begin{equation*}
		\begin{tikzpicture}[baseline=-1mm]
		\pic[pic text = $A$] at (0,0) {mpstensor};
		 \end{tikzpicture} \ ,
	\end{equation*}
	where the top line corresponds to the physical index: $i= 1,\ldots,d$, and the horizontal lines - to the (``virtual'' or ``bond'') matrix indices: $\alpha=1,\ldots,D$. Contraction of tensor indices is indicated by connecting the respective lines. If $M$ is a square matrix, i.e.\ a rank 2 tensor, then $\tr(M)$ is denoted by:
	
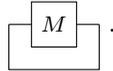
\begin{equation*}
		\begin{tikzpicture}[baseline=-1mm]
		\pic[pic text = $M$] at (0,0) {horizontalmatrix};
		\draw (-.5,0)--(-.5,-.5)--(0.5,-.5)--(.5,0);
		\end{tikzpicture} \ .
	\end{equation*}
	The coefficient corresponding to  the $\ket{i_1 i_2 \ldots i_N}$ basis element of the MPV $\ket{\psi^N_A}$ in \cref{eq:MPVdef} is denoted  by:
	
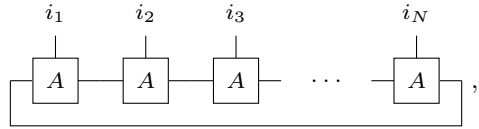
\begin{equation*}
		\begin{tikzpicture}[baseline=-1mm]
		\foreach \i in {0,1,2,4} \pic[pic text = $A$] at (\i,0) {mpstensor};
		\node at (0,.75) {$i_{1}$};
		\node at (1,.75) {$i_{2}$};
		\node at (2,.75) {$i_{3}$};
		\node at (4,.75) {$i_{N}$}; 
		\node at (3,0) {$\ldots$}; 
		\draw (-.5,0)--(-.5,-.5)--(4.5,-.5)--(4.5,0);
		\end{tikzpicture} \ ,
	\end{equation*}
	where we specified the values of the  physical indices.  We identify the  MPV of length $N$ generated by $A$ with the set of its coefficients  and denote the MPV as:
 	
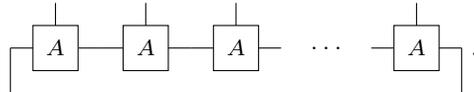
\begin{equation*}
		\begin{tikzpicture}[baseline=-1mm]
		\foreach \i in {0,1,2,4} \pic[pic text = $A$] at (\i,0) {mpstensor};
		\node at (3,0) {$\ldots$}; 
		\draw (-.5,0)--(-.5,-.5)--(4.5,-.5)--(4.5,0);
		\end{tikzpicture} \ .
	\end{equation*}
	
	\begin{definition} \label{def:blocking}
		Let $A$ be a tensor composed of matrices $\{A^i\}$.
		Blocking of  $b$ copies of   $A$  defines a new tensor denoted by $A_{\times b}$, which is composed of the matrices given by the $b$-fold products of  $A^i$, and are numbered by an  index $I:=(i_1,i_2,\ldots,i_b)$:
		\begin{equation*}
			\left\{ 
			(A_{\times b})^{I}  = A^{i_1}A^{i_2}\ldots A^{i_b}
			\left| \right. \ i_1,i_2,\ldots,i_b = 1,\ldots,d_A 
			\right\} \ .
		\end{equation*}
		 The new index $I$  corresponds to the basis $\left\{\ket{I}:= \ket{i_1}\otimes \ket{i_2}\otimes ,\ldots, \otimes \ket{i_b} \right\}$ of $ \mathcal{H}^{\otimes b} $. Graphically:
		 
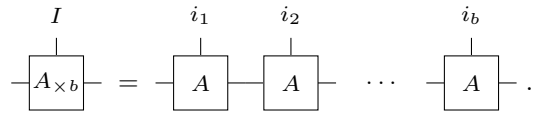
\begin{equation*}
			  \begin{tikzpicture}[baseline=-1mm]
			  \foreach \i in {0} \pic[pic text = $A_{\times b}$] at (\i,0) {mpstensorMED};
			  \node at (0,.75) {$I$};
			  \end{tikzpicture} \ = \
			  \begin{tikzpicture}[baseline=-1mm]
			  \foreach \i in {0,1,3} \pic[pic text = $A$] at (\i,0) {mpstensorMED};
			  \node at (0,.75) {$i_{1}$};
			  \node at (1,.75) {$i_{2}$};
			  \node at (3,.75) {$i_{b}$}; 
			  \node at (2,0) {$\ldots$}; 
			  \end{tikzpicture} \ .
		\end{equation*}
		The MPV of length $N$ generated by $A_{\times b}$ is  $\ket{\psi^N_{A_{\times b}}} \in \left( \mathcal{H}^{\otimes b} \right)^{\otimes N}$.
	\end{definition}

 	\begin{definition}[Injective tensor] \label{def:Injectivity}
 		A tensor $A$ consisting of $D\times D $ matrices $\{A^i\}_{i=1}^d$  is injective if 
 		\begin{equation*} 
	 		span \left\{ A^{i} \left|\right. i =1,\ldots, d    \right\} = \mathcal{M}_{D\times D} \ ,
 		\end{equation*}
 		where $\mathcal{M}_{D\times D}$ is the algebra of $D\times D $ matrices.
	 	\end{definition}
	 	
	\begin{definition}
		Let $A$ be a tensor consisting of matrices  $\{A^i\}_{i=1}^d$. The completely positive (CP) map associated with $A$ is defined by:
		\begin{equation*}
		E_A(\cdot)=\sum_{i=1}^D A^i \cdot {A^i}^\dagger \ , 
		\end{equation*}
		i.e., the matrices $\{A^i\}$ are the  Kraus operators of $E_A$  \cite{Wolf2012a}.
	\end{definition}
	 	
 	\begin{definition}[Normal tensor] \label{def:NormalTensor}
 		a tensor $A$, consisting of $D\times D $ matrices $\{A^i\}_{i=1}^d$,  is normal if there exists $L\in \mathbb{N}$ such  that:
 		\begin{equation*}  
	 		span\left\{ A^{i_1}A^{i_2}\ldots A^{i_L}\left|\right.{i_1,i_2,\ldots,i_L} = 1,\ldots, d  \right\} = \mathcal{M}_{D\times D} \ ,
 		\end{equation*}
 		where $\mathcal{M}_{D\times D}$ is the algebra of $D\times D $ matrices. 
 		That is, $A$ is normal if it becomes injective   after blocking   a sufficient number of its copies.   
 		In addition we require that  the  spectral radius of the CP map $E_A$ is equal to $1$.
 	\end{definition}
	 	
 	\begin{remark} 
 		If a tensor becomes injective after  blocking  $L_0$ copies, it is also injective when  blocking any number $L \geq L_0$ of copies. There is an upper bound on the minimal number of copies of a normal tensor needed to be blocked   in order for the blocked tensor  to  be injective, which depends only on its bond dimension \cite{Sanz2010}.
 	\end{remark}
	 	
 	\begin{proposition} \label{prop:NormalIFFPrimitive}
 		A tensor is normal (\Cref{def:NormalTensor}) iff the CP map   associated with it is primitive (irreducible and non-periodic). \textup{\cite{Wolf2012a}} 
 	\end{proposition}

 	\begin{definition}[Canonical form] \label{def:CF}
 		A tensor $A$ is in CF  if the matrices $A^i$ are block diagonal and have the following structure: 
 		\begin{equation} \label{eq:NTexpansion}
	 		A^i=\oplus_{k=1}^n \nu_k A^i_k  \ ,
 		\end{equation}
 		where $\{A_k\}$ are normal tensors and $\nu_k$ are constants. 
 	\end{definition}
	 	
 	\begin{definition}[Canonical form II] \label{def:CFII}
		 $A$ is in CFII if in addition to being in CF, for any $k$ appearing in \cref{eq:NTexpansion} the CP map $E_{A_k}$ is trace preserving, and has a positive full rank diagonal fixed point $\Lambda_k>0$.
 	\end{definition}

 	\begin{proposition} \label{prop:wlogCF}  
 		Let $\ket{\psi^N_A}$ be the MPV generated by a tensor $A$. 
 		If the CP map  $E_A$ has no periodic irreducible blocks, then 
	 		there exists  a tensor $\tilde{A}$ in CF (or CFII) such that:
	 	\begin{equation*}
		 	\ket{\psi^N_A}=\ket{\psi^N_{\tilde{A}}}\ ,  \forall N\in \mathbb{N} \ .
	 	\end{equation*}
	 	 If $E_A$ does have periodic blocks, then there exist a tensor $\tilde{A}$ in CF (of CFII) and  $b\in \mathbb{N}$ such that:
	 	\begin{equation*}
	 		\ket{\psi^{N}_{A_{\times b} }}=\ket{\psi^{N}_{\tilde{A}}} \ , \forall N\in \mathbb{N} \ , 
	  	\end{equation*}	
	  	where $A^{\times b}$ is the tensor obtained by blocking  $b$ copies of $A$ (\cref{def:blocking}). \textup{\cite{Cirac2017}}
 	\end{proposition}

	\begin{definition}[Basis of normal tensors] \label{def:BNT}
		Let $A$ be a tensor in CF.	A set of tensors $\{\hat{A}_j\}$ is said to be a basis of normal tensors (BNT) of $A$   if $\hat{A}_j$ are normal tensors, and for every $A_k$ appearing in $A$'s expansion (\cref{eq:NTexpansion}) there exists a unique $\hat{A}_{j }$, an invertible matrix $V$ and a phase $e^{i\phi}$ such that $A_k = e^{i\phi} V^{-1}\hat{A}_{j } V $.
	\end{definition}

	From now on whenever we consider a tensor $A$ in CF  we shall write it in terms of a BNT $\{A_j\}_{j=1}^m$:
	\begin{equation} \label{eq:BNTexpansion}
		A^i=\oplus_{j=1}^m\oplus_{q=1}^{r_j} \mu_{j,q} V_{j,q}^{-1} A^i_j V_{j,q} \ .
	\end{equation}
	The MPV of length $N$ generated by such a tensor  $A$ takes the form:
	\begin{equation*}
		\ket{\psi^N_A} = \sum_{j=1}^{m}\sum_{q=1}^{r_j}\left(\mu_{j,q}\right)^N\ket{\psi^N_{A_j}} \ . 
	\end{equation*}

\subsection{Representation theory}

	In this section we introduce  projective  representations. We review basic  facts  from representation theory, stated in the more general setting of projective representation, following \cite{HallQMath,Chuangxun}. Next, we describe how the general setting of  a MPV with a symmetry with respect to a  finite dimensional representation $\CC(g)$, can be simplified by writing the MPV in a form compatible with the decomposition of $\CC(g)$ into irreducible representations. Finally, we quote two theorems: Schur's lemma and the Wigner-Eckart theorem,  that will allow us to classify the tensors  generating symmetric MPVs.

\subsubsection{Projective representations}

	Let $\mathcal{H}$ be a finite dimensional Hilbert space. Denote by $\mathsf{U}(\mathcal{H})$ the group of unitary operators on $\mathcal{H}$. Throughout the paper, unless explicitly stated otherwise, $G$ will always refer to a finite group or a compact Lie group.

	\begin{definition}
		A function $\gamma: G\times G \rightarrow U(1)$ satisfying: 
		\begin{equation*}
			\begin{split}
			\gamma(g,h)\gamma(gh,f)=& \gamma(g,hf)\gamma(h,f) , \;\; \forall g,h,f\in G\\
			\gamma(g,e)=& \gamma(e,g) = 1,\;\; \forall g\in G \ , 
			\end{split}
		\end{equation*}
		where $e\in G$ is the trivial element, is called a multiplier of $G$. For compact Lie groups we require $\gamma$ to be continuous.
	\end{definition}
	
	\begin{definition} \label{def:projectiveRepresentation}
		A projective unitary representation of a group $G$ on $\mathcal{H}$  is a map $\CC:G \rightarrow \mathsf{U}(\mathcal{H})$ such that for all $g,h\in G$ $\CC(g)\CC(h)=\gamma(g,h)\CC(gh)$, where $\gamma$ is   a multiplier of $G$.	
	\end{definition}
	That is, projective unitary representations are unitary representations up to a phase factor.
	 Throughout this paper all representations will be assumed to be unitary and finite dimensional. From this point on, \textit{unitary representation}  shall be used to emphasize that it is  not projective. \textit{Projective representations} can refer to both, as unitary representations are a particular case of projective representations, namely, they are the  ones with the trivial multiplier.

	Two  projective representations $(\CC,\mathcal{H})$ and $(\CC^\prime,\mathcal{H}^\prime)$ with multipliers $\gamma$ and $\gamma^\prime$ are \textit{equivalent in the sense of  projective representations}  if there exist an isomorphism $\phi: \mathcal{H} \rightarrow \mathcal{H}^\prime$ and a function $\mu:G\rightarrow U(1)$ such that  $\CC^\prime(g)\phi = \mu(g)\phi \CC(g)$ for all $g\in G$. Their multipliers then satisfy:
	\begin{equation} \label{eq:cohomology}
		\gamma^\prime(g,h) = \gamma(g,h) \mu(g)\mu(h)\mu(gh)^{-1} \ .
	\end{equation} 
	\Cref{eq:cohomology} defines an equivalence relation on the group of  multipliers of $G$. The quotient  of the subgroup of multipliers of the form $\gamma(g,h)=\mu(g)\mu(h)\mu(gh)^{-1}$ in the group of all multipliers is the second cohomology group $H^2(G,U(1))$ of $G$ over $U(1)$ \cite{Chuangxun}. When two projective  representations  $\CC$ and $\CC^\prime$ have multipliers related  by  \cref{eq:cohomology}, for some function  $\mu:G\rightarrow U(1)$ we say they are in the same cohomology class.
	 
	 \begin{definition} \label{def:equivalenRepresentations}
	 	Two projective  representations $(\CC,\mathcal{H})$ and $(\CC^\prime,\mathcal{H}^\prime)$ with the same multiplier $\gamma$  are \textit{equivalent} if there exists an isomorphism $\phi: \mathcal{H} \rightarrow \mathcal{H}^\prime$ such that  $\CC^\prime(g)\phi = \phi \CC(g)$ for all $g\in G$. We denote $\CC^\prime(g)\cong \CC(g)$. 
	 \end{definition}

\subsubsection{Complete  reducibility}

	Fix a choice of representatives from the equivalence classes (\cref{def:equivalenRepresentations}) of irreducible  projective  representations of $G$ with multiplier $\gamma$; denote them by $D_\gamma^j:G\rightarrow \mathsf{U}(\mathcal{H}_j)$. Fixing a basis $\{\ket{i}\}$ in $\mathcal{H}_j$ for every $j$ defines the irreducible projective representation matrices: $D_\gamma^j(g) = \sum_{m,n} D_\gamma^j(g)_{m,n} \ket{m}\bra{n}$. These generalize the $SU(2)$  Wigner matrices to projective representations of arbitrary groups.
	
	Let $\mathcal{H}$  be a finite dimensional Hilbert space, and let $\CC:g\mapsto\CC(g)$ be a projective   representation 
	of  $G$  with multiplier $\gamma$.
	For finite and compact groups any finite dimensional projective representation is fully reducible and is equivalent to a direct sum of irreducible projective representations 
	$ \oplus_j D_\gamma^j(g)$ with the same multiplier, i.e.,  
	there exists a basis $\left\{\ket{j,m} \right\} $ of $\mathcal{H}$ such that: 
	\begin{equation} \label{eq:IrrepBasis}
		\CC(g) \ket{j,m}=\sum_n {D_\gamma^j}(g)_{n,m}\ket{j,n} \ .
	\end{equation} 
	We refer to such a basis as the irreducible representation basis of $\CC(g)$ (in general it is not unique, e.g., when an irreducible representation appears  multiple times \cite{Klimyk}; we shall assume a choice of such a basis). 
		
	When considering a  representation acting on a MPV, it is convenient to write the MPV in the irreducible representation basis. In the following we describe how this is achieved, and show that it does not interfere with CF properties of the tensor generating the MPV.
	
	\begin{remark} \label{rem:changeOfBasis}
		A change of basis of the {physical space}
		from $\{\ket{i}\}$ to the irreducible representation basis $\{\ket{j,m}\}$  (\cref{eq:IrrepBasis}), involves a transformation of the tensor generating the MPV:
		$A \mapsto \tilde{A}$, where $\tilde{A}$ consists of the matrices 
		$\{\tilde{A}^{j,m}=\sum_i \innerCG{j,m}{i} A^i \}$. This is easily seen by inserting an identity operator $\sum_{j,m} \ket{j,m}\bra{j,m}$ for every copy of $\mathcal{H}$ in the definition of $\ket{\psi^N_A}$  (\cref{eq:MPVdef}).
	\end{remark}
	
	\begin{proposition} \label{prop:KrusMix}
		Let $\{A^i\}_{i=1}^d$ be the Kraus operators defining a CP map $E_A$. For any   unitary $d\times d$ matrix $U$ the matrices $\{\sum_j U_{i,j}A^j\}_{i=1}^d$ define the same CP map. \textup{\cite{Wolf2012a}}
	\end{proposition}

	\begin{corollary} \label{prop:ChangeOfBasisDoesntRuinCF}
		Let A be a tensor in CF (CFII) composed  of the matrices $\{A^i\}$ corresponding to the basis $\{\ket{i}\}$ of $\mathcal{H}$. Then the tensor $\tilde{A}$, composed of the matrices 
		$\{\tilde{A}^{j,m}=\sum_i A^i \innerCG{j,m}{i}\}$ as in \cref{rem:changeOfBasis}, is also in CF (CFII).
	\end{corollary}
	\begin{proof}
		$ \tilde{A}$ has the same block structure as $A$ (\cref{eq:NTexpansion}):
		\begin{equation*}  
			\tilde{A}^{j,m}= \oplus_{k=1}^n \nu_k \tilde{A}^{j,m}_k =
			\oplus_{k=1}^n \nu_k \sum_i \innerCG{j,m}{i} {A}^i_k  \ .
		\end{equation*}
		According to	\cref{prop:NormalIFFPrimitive}, the normality and CFII properties of each block $\tilde{A}_k$ are defined by the CP map associated to it.   \Cref{prop:KrusMix} says this maps is not affected by the transformation $A_k \mapsto\tilde{A}_k$ because $\{\innerCG{j,m}{i}\}$ are the entries of a unitary matrix. Each block  $\tilde{A}_k$  is therefore a normal tensor (and in CFII).
	\end{proof}

\subsubsection{Intertwining relations}
	It was shown in \cite{Sanz2009,PollmannBergOshikawa} that   an injective tensor $A$ which generates a MPV with a global symmetry with respect to a representation $\CC_g$, satisfies:
	\begin{equation} \label{eq:AintertwinesExample}
		\begin{tikzpicture}[baseline=-1mm]
		\pic[pic text = $A$] at (0,0) {mpstensorBIG};
		\pic[pic text = $\CC(g)$] at (0,1) {verticalmatrixBIG};
		\end{tikzpicture} \ = \
		\begin{tikzpicture}[baseline=-1mm]
		\pic[pic text =  $\myinv{X (g)}$] at (0,0) {horizontalmatrixWIDE}; 
		\pic[pic text = $A$] at (1.15,0) {mpstensorBIG};
		\pic[pic text = $X(g)$] at (2.3,0) {horizontalmatrixWIDE};	   
		\end{tikzpicture} \ ,
	\end{equation} 
	i.e., for all $i= 1,\ldots,d$:	$ \sum_{i^\prime} \CC(g)_{ii^\prime} A^{i^\prime} = X(g)^{-1}A^iX(g)$,
	where $X(g)$ is a projective representation of $G$. While we will make the precise statement and derive this result later, we now point out that in  \cref{eq:AintertwinesExample}  the tensor  $A$ translates the action of $\CC(g)$ on the physical space into a group action on the virtual space. 
	
	In the following, we quote two theorems: Schur's lemma and the Wigner-Eckart theorem, which can be used to  classify tensors satisfying such intertwining relations.
	
	\begin{definition}[Intertwining map]
		Let $(\eta,V)$ and $(\pi,W)$ be projective representations  of a group $G$ with the same multiplier. A linear  map $T:V\rightarrow W$ is called an intertwining map if $\pi(g) T = T \eta(g) , \;\; \forall g\in G$. 
	\end{definition}

	\begin{lemma}[Schur's lemma] \label{thm:Schur}
		An intertwining map between irreducible projective  representations with the same multiplier is zero if they are inequivalent, and proportional to the identity if they are equal.\textup{ \cite{Chuangxun}}
	\end{lemma}
	
	The tensor product of two irreducible projective representations with multipliers $\gamma$ and $\gamma^\prime$ is a projective representation with multiplier $\gamma\gamma^\prime$ ($\gamma\gamma^\prime :(g,h)\mapsto\gamma(g,h)\gamma^\prime(g,h)$), and is generally  a reducible one. The unitary map that realizes the   decomposition of $D_\gamma^j(g)\otimes D_{\gamma^\prime}^l(g)$ into a direct sum of irreducible representations $\oplus_{J\in \mathfrak{J}} D_{\gamma\gamma^\prime}^J$ is the Clebsch-Gordan map  whose matrix elements are the Clebsch-Gordan coefficients $\innerCG{j,m;l,n}{J,M}$, which are determined by the choice of the representation matrices $D^j_\gamma$  (for a discussion of their uniqueness having fixed the representation matrices see \cite{Klimyk}).
	
	The following is a generalization of the  $SO(3)$ vector operators, well known in quantum mechanics \cite{HallQMath}.
	
	\begin{definition}[Vector operator] \label{def:VecOp}
		Let $(\eta,V),(\pi,W)$ and $(\kappa,\mathcal{H})$ be projective  representations of $G$ with $dim(\mathcal{H})= d$.
		A vector operator with respect to $(\kappa,{\pi},\eta)$ is a $d$-tuple of linear operators  $\vec{A}=(A^1,A^2,\ldots,A^d)$, $A^i:V\mapsto W$ which, for all  $ g\in G$ and all $\vec{v}\in \mathcal{H}$, satisfies:
		\begin{equation} \label{eq:VecOpDef}
			\left(\kappa(g) \vec{v}\right)\cdot \vec{A} = 
			\pi(g) \left( \vec{v}\cdot\vec{A}\right) \eta(g)^{-1}  \,
		\end{equation}
		where $\vec{v}\cdot \vec{A}:=\sum_i v^iA^i$.
	\end{definition}

	It was shown in \cite{Sanz2009} that  \cref{eq:AintertwinesExample} can be used to determine the tensor $A$ satisfying it, and that it consists of Clebsch-Gordan coefficients. We will derive the same result using a generalized version of the well known Wigner-Eckart theorem, using the fact  that \cref{eq:AintertwinesExample} resembles a vector operator relation for $A$ (\cref{def:VecOp}).
			
	\begin{theorem}[Wigner-Eckart] \label{thm:GenWigEck}
		Let $ D_\gamma^{J_0}(g),\,{D_{\gamma^\prime}^j(g)}$ and ${D_{\gamma^{\prime\prime}}^l(g)}$ be irreducible projective  representations.
		Let $\vec{A}$ be a vector operator with respect to  $(\kappa: =D_\gamma^{J_0},\,\pi:=D_{\gamma^\prime}^j,\,\eta:=D_{\gamma^{\prime\prime}}^l)$. If $ \gamma \gamma^{\prime\prime}\neq\gamma^\prime $, then $A=0$. Otherwise (if $ \gamma \gamma^{\prime\prime}= \gamma^\prime $), then
		$\{A^M | {M=1,\ldots, dim(J_0)}\}$  are of the form:
		\begin{equation} \label{eq:WignerEckartVectorOP}
			A^M = 
			\sum_{J\in \mathfrak{J}:{D^J}={D^{J_0}}} \alpha_J \sum_{m,n}\innerCG{j,m;\overline{l},n}{J,M} \ket{m}\bra{n} \ ,
		\end{equation}
		where $\mathfrak{J}$ is the set of irreducible projective representation  indices appearing in the decomposition of ${D_{\gamma^{\prime}}^j(g)}\otimes \overline{{D_{\gamma^{\prime\prime}}^l(g)}}$, $\innerCG{j,m;\overline{l},n}{J,M}$ are the Clebsch-Gordan coefficients of this decomposition,  $\overline{D_{\gamma^{\prime\prime}}^l(g)} $ is the complex conjugate representation to ${D_{\gamma^{\prime\prime}}^l(g)}$, $\{\ket{m}\}$ and $\{\ket{n}\}$ are the irreducible representation  bases: $\pi(g)\ket{m}=\sum_{m^\prime}{D_{\gamma^{\prime}}^j(g)}_{m^\prime,m}\ket{m^\prime}$, $\eta(g)\ket{n}=\sum_{n^\prime}{D_{\gamma^{\prime\prime}}^l(g)}_{n^\prime,n}\ket{n^\prime}$ and $\alpha_J$ are arbitrary constants.
	\end{theorem}
	
	For a proof of the theorem in the familiar $SO(3)$ setting, we refer the reader to \cite{HallQMath}; for a proof in the  the setting of projective representations  see \cite{Agrawala1980}.

	\begin{remark} \label{rem:WignerEckartRecipe}
		Apart from the freedom of  choosing the constants $\{\alpha_J\}$ in \cref{eq:WignerEckartVectorOP}, there is an additional freedom which comes from  the fact that the  the Clebsch-Gordan coefficients are not uniquely determined by the  irreducible representation matrices  \cite{Klimyk}.
	\end{remark}
	
	\begin{remark}
		The multiplier of the complex conjugate projective representation  $\overline{{D_\gamma^l(g)}}$  is $\gamma^{-1}$. We will always use \cref{thm:GenWigEck} with $\gamma\equiv 1$, then $A=0$ unless  $\gamma^\prime = \gamma^{\prime\prime}$.
	\end{remark}

	\begin{remark} \label{rem:uniqueCoice}
		We assume a choice of a unique representative in each equivalence class of irreducible projective representations of $G$, so any two are either inequivalent or are represented by the same matrices.
	\end{remark}
	\begin{remark}
		$A$ is zero if ${D_\gamma^{J_0}(g)}$ does not appear in the decomposition of  ${D_{\gamma^\prime}^j(g)}\otimes \overline{{D_{\gamma^{\prime\prime}}^l(g)}}$.
		 There is a $J$ summation in \cref{eq:WignerEckartVectorOP} because in general the same irreducible   representation could appear multiple times  in the decomposition of the tensor product of two irreducible   representations.
	\end{remark}

\subsection{Physical states and their symmetries} 
	Gauge theories involve the dynamics of two kinds of degrees of freedom: \textit{matter } and \textit{gauge field}. Given those two ingredients, one can consider three  types of states:
	states of only matter degrees of freedom,   states of only gauge field degrees of freedom and  states of both matter and gauge field. 
	These correspond to non-interacting theories,  pure gauge theories and interacting gauge theories respectively (where interactions are understood as those between matter and gauge degrees of freedom).
	
	When constructing a gauge theory one usually starts from an interaction-free theory of the matter  degree of freedom which is invariant with respect to a  group  of global transformations, i.e., the same group element acting in each point in space (or space-time). Adding an additional degree of freedom - the gauge field -  with its own transformation law with respect to  the group, allows to define local symmetry  operators which act  on both the  matter and the gauge field degrees of freedom. These operators commute with the transformed (gauged) Hamiltonian, and the subspace  of states which is invariant under all such operators is considered as the space of physical states.
	The generators of such local symmetry operators are the so-called Gauss law operators. They correspond to locally conserved quantities (charges), i.e., associated to each point in space (or space-time). 
	
	Conversely, one could start from a pure gauge field theory with a local symmetry and  couple a matter degree of freedom to it, once again  resulting in a system with local symmetry. Finally one could have matter and gauge field coupled in such a way that the combined state has a local symmetry but neither the mass state  nor the gauge field state have a symmetry on their own.

	We shall now describe the three types  of MPVs considered in this paper, corresponding to the above mentioned types of states,  and for each one of them define the symmetries which will be investigated in subsequent sections.

\subsubsection{Matter MPV} \label{sub:MatterMPV}
	Let $\mathcal{H}_A$  be a $d_A$ dimensional Hilbert space corresponding to a single  {degree of freedom} (``spin'').
	Consider  $N$ such   ``spins''  positioned on a one dimensional lattice, with periodic boundary conditions.
	A tensor $A$ consisting of square matrices $\{A^i \}_{i=1}^{d_A}$ generates a TI-MPV that describes a state  of the chain of matter ``spins''. 
	Let $\CC$ be a  unitary representation of $G$ on $\mathcal{H}_A$, $\CC:g\mapsto\CC(g)$. 
	  
	It is well known that in order to lift a global symmetry   to be a local one, an additional degree of freedom must    be introduced \cite{Peskin:257493}. When investigating the possibility of a local symmetry for a matter MPV, we  find this statement reaffirmed  (see \cref{thm:TrivTensorTrans}). We define the setting of the theorem in the following:
	\begin{definitionRE}[Local Symmetry for matter MPV] \label{def:OneDofSymm}
	 	A MPV $\ket{\psi^N_A} $ 
		 has a local symmetry with respect to $\CC(g)$ if for all $N\in \mathbb{N}$:
		\begin{equation*}
			\CC_{g_1}\otimes \CC_{g_2}\otimes \ldots \otimes\CC_{g_N}\ket{\psi_A^N}=\ket{\psi_A^N} , \;\; \forall g_1,g_2,\ldots, g_N\in G \ .
		\end{equation*}
	\end{definitionRE}
	 
	 Global symmetry in MPS have been studied extensively \cite{Sanz2009,Perez-Garcia2008PhRvL}. In order for this paper to be self contained, we   quote and then derive the main  result, which classifies the tenors $A$ that generate MPV with the following symmetry:
	\begin{definitionRE}[Global Symmetry for matter MPV] \label{def:OneDofGlobalSymm}
		A MPV $\ket{\psi^N_A}  $ 
		has a global symmetry with respect to $\CC(g)$ if for all $N\in \mathbb{N}$:
		\begin{equation*}
			\CC_{g}\otimes \CC_{g}\otimes \ldots \otimes\CC_{g}\ket{\psi_A^N}=\ket{\psi_A^N} , \;\; \forall g\in G \ .
		\end{equation*}
	\end{definitionRE}
	
	\begin{remark}
	  	The condition of a local symmetry (\cref{def:OneDofSymm}) is equivalent to invariance  under  any single-site group action (all $g_i=e$ except one). For TI-MPV it is therefore sufficient to consider only $g_1\neq e$.
	\end{remark}

	\subsubsection{Gauge field  MPV} \label{sub:GFMPV}
	Next we shall consider a case in which the local transformations act on two neighboring sites of a TI-MPV, which will be eventually seen as the pure gauge case.
	
	Let $\mathcal{H}_B$  be a $d_B$ dimensional Hilbert space corresponding to a single ``spin''.
	Consider  $N$ such spins positioned  on a one dimensional lattice, with periodic boundary conditions.
	A tensor $B$ consisting of square matrices $\{B^i \}_{i=1}^{d_B}$ generates a TI-MPV that describes a state  of the chain of gauge field ``spins''.
	
	\begin{definitionRE}[Local Symmetry for gauge field MPV] \label{def:BBSymm}
		Let $\RR,\LL$  be two projective   representations of $G$ on $\mathcal{H}_B$, $\RR:g\mapsto\RR(g)$, $\LL:g\mapsto\LL(g)$ with multipliers $\gamma$ and $\gamma^{-1}$, so that the tensor product $\RR(g)\otimes\LL(g)$ is a unitary representation.
		A MPV $\ket{\psi^N_{B}}  $ 
		has a local symmetry with respect to $\RR(g)\otimes\LL(g)$ if for all $N\in \mathbb{N}$ and for any two neighboring lattice sites $K$ and $K+1$:
		\begin{equation*}
			\RR^{[K]}_{g }\otimes \LL^{[K+1]}_{g } \ket{\psi^N_{B}}=\ket{\psi^N_B} , \;\; \forall g \in G \ .
		\end{equation*}
	\end{definitionRE}

\subsubsection{Matter and gauge field  MPV} \label{sub:MandGFSetting}
	 Let $\mathcal{H}_A$ and  $\mathcal{H}_B$ be as in \cref{sub:MatterMPV} and \cref{sub:GFMPV} respectively. 
	 Consider a lattice of length  $2N$ with  matter and gauge field spins alternating  among sites. 
	Tensors $A$ and $B$, consisting of   $D_1 \times D_2$ matrices $\{A^i\}_{i=1}^{d_A}$
	and $D_2 \times D_1$ matrices   $\{B^j\}_{j=1}^{d_B}$ respectively, generate a TI-MPV (in the sense of translating two sites)  that  describes a state of the chain of matter and gauge field ``spins''. The MPV, generated by a tensor we denote $AB$, takes the form:
	\begin{equation*} 
		\ket{\psi^N_{AB}} =\sum_{\{i\},\{j\}} \tr \left( A^{i_1}B^{j_1}A^{i_2}B^{j_2} \ldots A^{i_N}B^{j_N} \right) \ket{i_1 j_1 i_2 j_2 \ldots i_N j_N} \ .
	\end{equation*}
	 
 	In lattice gauge theories,   the matter degrees of freedom are located on the sites of a lattice  whereas the gauge field degrees of freedom - on the links connecting adjacent sites \cite{Wilson:1974sk}. In the one dimensional case, our setting differs from this structure only in notation, e.g., we could have chosen to call the even numbered sites ``links''.

	 Let $\CC(g)$ and $\RR(g), \,\LL(g)$ be as in \cref{sub:MatterMPV} and \cref{sub:GFMPV} respectively.
	 \begin{definitionRE}[Local Symmetry for both matter and gauge field MPV] \label{def:BABSymm}
	 	A MPV $\ket{\psi^N_{AB}}  $ 
	 	has a local symmetry with respect to  $\RR(g)\otimes\CC(g)\otimes\LL(g)$ if for all $N\in \mathbb{N}$ and for any three neighboring lattice sites numbered $2K,2K+1$ and $2K+2$ (corresponding to $\mathcal{H}_B\otimes \mathcal{H}_A \otimes \mathcal{H}_B$):
		\begin{equation*} 
			\RR(g)^{[2K]}\otimes\CC(g)^{[2K+1]} \otimes\LL(g)^{[2K+2]} \ket{\psi^N_{AB}}= \ket{\psi^N_{AB}}, \;\; \forall g\in G \ .  
		\end{equation*}
	 \end{definitionRE}

\subsection{Generators and Gauss' law} \label{sub:GenAndGaussLaw}
	In the previous section we defined the symmetries   in terms of representations of a group $G$. For matrix Lie groups it is often the case that one could describe  the same symmetry in terms of representations of the Lie algebra  $\mathfrak{g}$  of $G$. While the two  descriptions are mathematically equivalent, it is precisely the elements of the Lie algebra representation that correspond to  observables in physical theories. Such observables are conserved by the dynamics in a theory which respects the symmetry, and are therefore of great importance.
	
	To each scenario described above (\cref{sub:MatterMPV},
	\cref{sub:GFMPV} and \cref{sub:MandGFSetting})  correspond different such observables, and physical theories corresponding to the different settings  - matter, gauge field or matter and gauge field - observe different conservation laws. In the following we   describe the relation of those settings to   physical lattice gauge theories \cite{Zohar:2014qma}.
	
	When $G$ is a compact and connected Lie group, e.g.\ $U(1)$ or $SU(N)$, the exponential map $\exp:\mathfrak{g}\rightarrow G$   is surjective. Thus every group element can be written as  an exponential of an element in the Lie algebra $\mathfrak{g}$ \cite{HallLieGroup}. 
	Let $\RR(g),\, \LL(g)$ and $\CC(g)$ be   representations  on $\mathcal{H}_B$ and $\mathcal{H}_A$ respectively 
	(for $SU(N)$ we can always choose $\RR(g)$ and $\LL(g)$ to be unitary   representations keeping $\RR(g)\otimes\LL(g)$ unchanged \cite{HallQMath}), 
	and let $\ket{\psi^N_{AB}}$ be  as defined in \cref{sub:MandGFSetting}. We can  express the physical representations as exponentials of generators:
	\begin{equation*}
		\begin{split}
		\CC(g) = & \exp\left( i \sum_a  Q_a{\varphi_a(g)}\right) \\
		\RR(g) =& \exp\left( i \sum_a  R_a{\varphi_a(g)}\right) \\
		\LL(g) =& \exp\left( i \sum_a  L_a{\varphi_a(g)}\right) \ ,
		\end{split}
	\end{equation*}
	where $\{\varphi_a(g)\}_{a=1}^{dim(\mathfrak{g})}$ are real parameters and  $\{R_a\}_{a=1}^{dim(\mathfrak{g})} , \, \{L_a\}_{a=1}^{dim(\mathfrak{g})}$ and $\{Q_a\}_{a=1}^{dim(\mathfrak{g})}$ are  Hermitian operators on $\mathcal{H}_B$ and $\mathcal{H}_A$ respectively such that  $\{iR_a\}, \{iL_a\}$ and $\{iQ_a\}$  are bases of the respective Lie algebras. 
	In the Hamiltonian formulation of lattice gauge theories    \cite{Kogut1975,Zohar:2014qma} 
	$\{R_a\}$ and $\{L_a\}$ satisfy the Lie algebra relations:
	\begin{equation*}
		\begin{split}
		\left[ R_a , R_b \right] =&  i f_{abc}R_c \\
		\left[ L_a , L_b \right] =&  i f_{abc}R_c  \\
		\left[ R_a ,L_b \right] =&  0 \ , 
		\end{split}
	\end{equation*}
	where $f_{abc}$ are the  structure constants of the Lie algebra $\mathfrak{g}$. $\{Q_a\}$ satisfy the relations: 
	\begin{equation*}
		\left[ Q_a , Q_b \right] =  i f_{abc}Q_c \ .
	\end{equation*}
	The local symmetry transformations appearing in the matter and gauge field MPV scenario (\cref{def:BABSymm}):
	\begin{equation} \label{eq:invarianceBABforGauss}
		\RR^{[2K]}(g)\otimes\CC^{[2K+1]}(g)\otimes \LL^{[2K+2]}(g)\ket{\psi^N_{AB}} = \ket{\psi^N_{AB}} \ ,
	\end{equation}
	are generated by the operators:
	\begin{equation*}
		G_a^{[2K+1]}:=	\left(R_a^{[2K]} + Q_a^{[2K+1]} +L_a^{[2K+2]} \right)  \ .
	\end{equation*}	
	Differentiating \cref{eq:invarianceBABforGauss} with respect to any of the parameters $\varphi_a$ we obtain:
	\begin{equation} \label{eq:GaussLaw}  
		\left(R_a^{[2K]} + Q_a^{[2K+1]} +L_a^{[2K+2]} \right) \ket{\psi^N_{AB}} = G_a^{[2K+1]} \ket{\psi^N_{AB}}
		=0 \ .
	\end{equation} 
	This is the lattice version of Gauss' law. 
	In  physical theories,   states $\ket{\psi_A}$ have a global symmetry generated by $\{Q_a\}$ - the $SU(N)$ charge operators. In the $U(1)$ case there is one generator  $Q$ - the electric charge operator; furthermore,  
	for Abelian groups $L = -R$. In that case \cref{eq:GaussLaw}  says that  at each lattice site corresponding to matter, the charge is equal  to the difference between the values of $L$ on the right and on the left of it (the 1D lattice divergence of $L$). This becomes Gauss' law when taking a  continuum limit. $L$ is therefore identified as the electric field. Analogously, in the $SU(N)$ case $\{R_a\}$ and $\{L_a\}$ are identified with right and left electric fields respectively
	\cite{Zohar:2014qma}.

	The same kind of equation can be obtained for the case of a gauge field MPV with a local symmetry (\cref{def:BBSymm}):
	\begin{equation*}   
		\left(R_a^{[K]} +L_a^{[K+1]} \right) \ket{\psi^N_{ B}} =0 \ .
	\end{equation*}
	In the case of a global symmetry  for  a matter MPV, differentiating the symmetry relation (\cref{def:OneDofGlobalSymm}), we obtain  a global operator - the total charge:
	\begin{equation*}   
		\sum_K Q_a^{[K]}  \ket{\psi^N_{A}} =0 \ .
	\end{equation*}

\section{Results} \label{sec:RsultsOverview}
	We  summarize the results presented in this paper, first stating the main results of each of the cases presented above, and then turning to a more detailed and formal description. The detailed proofs will be given in the subsequent sections. 
	For each one of the settings introduced  in the previous section,  we shall first
	show that the symmetry condition implies a transformation relation satisfied by the tensor(s) generating the MPV. Second,  we shall show that those  transformation relations determine the structure of the tensor(s). For each setting we shall then discuss implications of the derived tensor structures.

\subsection{Matter MPV with local symmetry} \label{sub:Summary_1degree of freedom} 
	We show that    a MPV with one degree of freedom - the mass ``spins'' - can have a  local symmetry as in \cref{def:OneDofSymm}, only if it is the trivial one. This is consistent with the way gauge invariant states are usually constructed in lattice gauge theories, as well as with the construction of continuum gauge theories, where an additional degree of freedom is introduced. 
	The first observation is a general one, not restricted to  MPVs:
	
	\begin{restatable}{propositionRE}{propSinglets}  
		\label{prop:Singlets}%
		Let $\mathcal{H}$ be a finite dimensional Hilbert space and let $\CC:g\mapsto\CC(g)$ be a representation on $\mathcal{H}$. Let $\ket{\psi^N} \in \HN$ be a vector with  a local symmetry, i.e.\ 
		\begin{equation*} 
			\CC({g_1})\otimes \CC({g_2})\otimes \ldots \CC({g_N})\ket{\psi^N}=\ket{\psi^N} , \;\;\; \forall g_1,g_2,\ldots, g_N\in G \ .
		\end{equation*}
		Then  $\ket{\psi^N} \in {\mathcal{H}_0}^{\otimes N}$, where $\mathcal{H}_0\subset \mathcal{H}$ is the subspace on which $\CC(g)$ acts trivially. 
	\end{restatable}

 In the following we show that for MPVs a similar statement to \cref{prop:Singlets} can be made for the tensor generating the MPV. Let $\ket{\psi_A}$ and $\CC(g)$ be as in \cref{sub:MatterMPV}. According to \cref{prop:wlogCF}, given an arbitrary tensor $A$ generating  $\ket{\psi_A}$, one can obtain a tensor in CF  which generates the same state, (possibly after blocking $A$). We therefore assume $A$ to be in CF.

	\begin{restatable}{theoremRE}{TrivTensorTrans}
	\label{thm:TrivTensorTrans}%
		Let $A$ be a tensor  in CF generating  a MPV   with a local symmetry with respect to  a   representation $\CC(g)$ (\cref{def:OneDofSymm}). Then for all $g\in G$ the tensor  $A$ satisfies: 
		\begin{equation*} \label{eq:thetaA_eq_A}
			\begin{tikzpicture}[baseline=-1mm]
			\pic[pic text = $A$] at (0,0) {mpstensorMED};
			\pic[pic text = $\CC(g)$]  at (0,1) {verticalmatrixMED};
			\end{tikzpicture} \ = \ 
			\begin{tikzpicture}[baseline=-1mm]
			\pic[pic text = $A$] at (0,0) {mpstensorMED};
			\end{tikzpicture} \ ,
		\end{equation*} 
		i.e., for all $i= 1,\ldots,d_A$:	$ \sum_{i^\prime} \CC(g)_{ii^\prime} A^{i^\prime} = A^i$.
	\end{restatable} 
	According to \cref{rem:changeOfBasis}, the MPV generated by $A$ can be written in terms of a tensor $\tilde{A}$, composed of the matrices $\{\tilde{A}^{j,m}\}$, corresponding to the irreducible representation basis $\{\ket{j,m}\}$ on which $\CC(g)$ acts as 
	$\CC(g) \ket{j,m}=\sum_n {D^j}(g)_{n,m}\ket{j,n}$. According to \cref{prop:ChangeOfBasisDoesntRuinCF}, $\tilde{A}$ is also in CF. Applying \cref{thm:TrivTensorTrans} to $\tilde{A}$ leads to the following:

	\begin{restatable}{corollaryRE}{decomposeTheta}\label{cor:decomposeTheta}
		 The matrices   $\tilde{A}^{j,m} $ are non-zero only for $j$ such that ${D^j(g)}\equiv \II_{1 \times 1}$.
	\end{restatable}

\subsection{Gauge field MPV} \label{sub:pureGFSummary}
	
	We show that a  local symmetry for a gauge field MPV $\ket{\psi^N_B}$ generated by a tensor $B$ (in CFII) (as defined in \cref{sub:GFMPV}),  implies the following transformation relations for $B$:
	\begin{equation}   \label{eq:BLRTGroupransformationPRE}
		\begin{tikzpicture}[baseline=-1mm]
		\pic[pic text = $B$] at (0,0) {mpstensorBIG};
		\pic[pic text = $\RR(g)$] at (0,1) {verticalmatrixBIG};
		\end{tikzpicture} \  = \ 
		\begin{tikzpicture}[baseline=-1mm]
		\pic[pic text = $B$] at (0,0) {mpstensorBIG};
		\pic[pic text =  $X(g)$] at (1,0) {horizontalmatrixBIG};
		\end{tikzpicture} \;\;\;\;\;\;\;\;\; ;  \;\;\;\;\;\;\;\;\; 
		\begin{tikzpicture}[baseline=-1mm]
		\pic[pic text = $B $] at (0,0) {mpstensorBIG};
		\pic[pic text = $\LL(g)$] at (0,1) {verticalmatrixBIG};
		\end{tikzpicture} \   = \ 
		\begin{tikzpicture}[baseline=-1mm]
		\pic[pic text =  $\myinv{X(g)}$] at (0,0) {horizontalmatrixWIDE};
		\pic[pic text = $B$] at (1.15,0) {mpstensorBIG}; 
		\end{tikzpicture} \ ,
	\end{equation} 
	where $X(g)$ is a projective   representation with the same multiplier as that of $\RR(g)$.
	This transformation relation  allows to determine the structure of the physical Hilbert space of the gauge field degree of freedom. We find that the gauge field ``spins'' are composed of   right and  left parts:
	\begin{equation*}	 %\label{eq:PhysSpaceDecomp}
		\mathcal{H}_B = \bigoplus_k \mathcal{H}_{l_k} \otimes \mathcal{H}_{r_k} \ , 
	\end{equation*}
	where $\mathcal{H}_{r_k}$ are irreducible representation spaces of $G$. The physical representations $\RR(g)$ and $\LL(g)$  take the forms: $\RR(g) = \oplus_k(\II \otimes D^{r_k}_\gamma(g)) ,\,\LL(g) = \oplus_k( {D^{l_k}_{\myinv{\gamma}}(g)}\otimes  \II)$, and act on the right and left parts of $\mathcal{H}_B$ respectively. 
	
	The transformation relation \cref{eq:BLRTGroupransformationPRE} also determines the structure of the tensor $B$.  
	Decompose $X(g)$ into its constituent irreducible representations  and project  \cref{eq:BLRTGroupransformationPRE} to the corresponding irreducible subspaces (virtual and physical). 
	The obtained  blocks of $B$ intertwine irreducible representations, and   their structure is therefore determined by Schur's lemma (\cref{thm:Schur}). When the irreducible representations in \cref{eq:BLRTGroupransformationPRE} match, the corresponding elementary block of $B$ is proportional to the tensor composed of the matrices:
	\begin{equation*}
		%	\left\{ 
		B^{m,n}=\ket{m}\bra{n}  \ ,
		%	\left|\right.  
		%m=1,\ldots,dim(l), n=1,\ldots,dim(r) 
		%  \right\} 
	\end{equation*}
	 so that $B$, when represented   in  graphical notation,  takes the form:
	\begin{equation*} 
		%\begin{split}
		\begin{tikzpicture}[baseline=-1mm]
		\draw (-.5,0)--(-.15,0)--(-.15,.75);
		\draw (.5,0)--(.15,0)--(.15,.75);
		\pic[pic text = $B$] at (0,0) {horizontalmatrixWIDE};
		\end{tikzpicture} \  \propto \ 
		\begin{tikzpicture}[baseline=-1mm]
		\draw (-.5,0)--(-.15,0)--(-.15,.75);
		\draw (.5,0)--(.15,0)--(.15,.75);
		\end{tikzpicture} \ .
	\end{equation*}
	Otherwise, if the irreducible representations do not match, that block of $B$ is zero.

	The tensor $B$ is composed out of such elementary building blocks   multiplied  by constants - free parameters. Finally, we show that for any B generating a gauge field MPV with a local symmetry, one can always find a tensor A, describing a matter degree of freedom, such that the   matter and gauge field MPV generated by A and B   has a local symmetry.
		
	We shall now describe these results in detail, and state the relevant theorems.
	  
	\bigskip
	  
	Let $\ket{\psi_B}$ be a MPV generated by a tensor $B$ and let $\RR(g),\,\LL(g)$ be projective   representations as defined in \cref{sub:GFMPV}. 
	As in the case of a matter MPV above, according to \cref{prop:wlogCF} we can assume $B$ is in CFII and write it in terms of its BNT: %(\cref{rem:wlogBNT}):
	\begin{equation} \label{eq:BsBlockStructure}
		B^i=\oplus_{j=1}^n \oplus_{q=1}^{r_j} \mu_{j,q} B^i_j  \ ,
	\end{equation}
	where $\{B_j\}$ are  normal tensors in CFII forming a  BNT of $B$ (\cref{def:BNT}) and $\mu_{j,q}$ are constants.

	\begin{restatable}[Gauge field MPV with a local symmetry]{theoremRE}{OnlyFieldGaugeGroup}	
		\label{thm:OnlyFieldGaugeGroup}
		A tensor $B$ in CFII which generates a MPV that has a local symmetry with respect to    $\RR(g)\otimes\LL(g)$ where $\RR(g)$ and $\LL(g)$ are projective    representations with inverse multipliers  (\cref{def:BBSymm}), transforms under the representation matrices as:
		\begin{equation}   \label{eq:BLRTGroupransformation}
			\begin{tikzpicture}[baseline=-1mm]
			\pic[pic text = $B$] at (0,0) {mpstensorBIG};
			\pic[pic text = $\RR(g)$] at (0,1) {verticalmatrixBIG};
			\end{tikzpicture} \  = \ 
			\begin{tikzpicture}[baseline=-1mm]
			\pic[pic text = $B$] at (0,0) {mpstensorBIG};
			\pic[pic text =  $X(g)$] at (1,0) {horizontalmatrixBIG};
			\end{tikzpicture} \;\;\;\;\;\;\;\;\; ;  \;\;\;\;\;\;\;\;\; 
			\begin{tikzpicture}[baseline=-1mm]
			\pic[pic text = $B $] at (0,0) {mpstensorBIG};
			\pic[pic text = $\LL(g)$] at (0,1) {verticalmatrixBIG};
			\end{tikzpicture} \   = \ 
			\begin{tikzpicture}[baseline=-1mm]
			\pic[pic text =  $\myinv{X(g)}$] at (0,0) {horizontalmatrixWIDE};
			\pic[pic text = $B$] at (1.15,0) {mpstensorBIG}; 
			\end{tikzpicture} \ ,
		\end{equation} 
		where   $X(g)$ is a projective  representation of $G$ with the same multiplier as $\RR(g)$ and with the same block structure as $B$ (\cref{eq:BsBlockStructure}):
		\begin{equation} \label{eq:XStruct1}
			X(g) = \oplus_{j=1}^m \oplus_{q=1}^{r_j} {X_j}(g) \ .
		\end{equation}
	\end{restatable}
	
	When considering matter and gauge field MPVs in the next section, we will show that in that setting, a more general relation than \cref{eq:BLRTGroupransformation} is satisfied by the tensor $B$. Namely:
	\begin{equation} \label{eq:BRightGroupTransPre}
		\begin{tikzpicture}[baseline=-1mm]
		\pic[pic text = $B$] at (0,0) {mpstensorBIG};
		\pic[pic text = $\RR(g)$] at (0,1) {verticalmatrixBIG};
		\end{tikzpicture} \  = \ 
		\begin{tikzpicture}[baseline=-1mm]
		\pic[pic text = $B$] at (0,0) {mpstensorBIG};
		\pic[pic text =  $X(g)$] at (1,0) {horizontalmatrixBIG};
		\end{tikzpicture} \;\;\;\;\;\;\;\;\; ;  \;\;\;\;\;\;\;\;\; 
		\begin{tikzpicture}[baseline=-1mm]
		\pic[pic text = $B $] at (0,0) {mpstensorBIG};
		\pic[pic text = $\LL(g)$] at (0,1) {verticalmatrixBIG};
		\end{tikzpicture} \   = \ 
		\begin{tikzpicture}[baseline=-1mm]
		\pic[pic text =  $\myinv{Y{(g)}}$] at (0,0) {horizontalmatrixWIDE};
		\pic[pic text = $B$] at (1.15,0) {mpstensorBIG}; 
		\end{tikzpicture} \ ,
	\end{equation} 
	where $X(g)$ and $Y(g)$ are different projective representations (in the case when $B$ is composed of non-square matrices they are of different dimensions).
	We shall now present results which follow from the more general  relation (\cref{eq:BRightGroupTransPre}), as they will be relevant also in the next section. Then we will apply them to the case at hand - \cref{eq:BLRTGroupransformation} (i.e., when $X(g) = Y(g)$ and $B$ is composed out of square matrices).

	\Cref{eq:BRightGroupTransPre} allows us to determine the structure of the Hilbert space of the gauge field degree of freedom. The fact that the action of $\RR(g)$ is translated to a matrix multiplication from the right, and that of $\LL(g)$ - to multiplication from the left implies that their actions on the ``spin'' representing the gauge field are independent, consequently the ``spin'' must be composed of right and left parts:
	 
	\begin{restatable}[Structure of $\mathcal{H}_B$]{propositionRE}{HBstructure}  \label{prop:HBsturct}
		Given a tensor $B$, projective representations $\RR(g)$, $\LL(g)$ with inverse multipliers $\gamma$ and $\gamma^{-1}$ (as defined in \cref{sub:GFMPV}) and matrices $X(g)$ and $Y(g)$ which satisfy \cref{eq:BRightGroupTransPre},  the Hilbert space  $\mathcal{H}_B$ can be restricted to a representation space of $G\times G$ and thus decomposes into a direct sum of tensor products of irreducible representation spaces of $G$:
		\begin{equation*}	
			\mathcal{H}_B = \bigoplus_{k=1}^M \mathcal{H}_{l_k} \otimes \mathcal{H}_{r_k} \ , 
		\end{equation*}
		where $r_k$ and $l_k$ are  irreducible representation labels.
	\end{restatable}

	The structure of $\mathcal{H}_B$ described in \cite{Zohar:2015jnb} is a particular case of this Hilbert space. There:
	\begin{equation} \label{eq:KSH}
		\mathcal{H}_B = \bigoplus_{k=1}^M \mathcal{H}_{\overline{r_k}} \otimes \mathcal{H}_{r_k} \ , 
	\end{equation}
	where $\overline{r_k}$ indicates the complex conjugate representation to $r_k$.
	\Cref{eq:KSH} is a truncated version of the K-S Hilbert space, which allows to regain the whole space if M is increased such that all the irreducible representations are included. Each $k$ sector in \cref{eq:KSH}:
	$\mathcal{H}_{\overline{r_k}} \otimes \mathcal{H}_{r_k}$ is isomorphic to the function space spanned by 
	\begin{equation*}
		\left\{ D^{r_k}_{m,n}:g\mapsto D^{r_k}_{m,n}(g) \left|\right. m,n = 1,\ldots, dim(r_k) \right\} \subset L^2(G) \ ,
	\end{equation*}
	 with $\RR(g)$ and $\LL(g)$ equivalent to  the right and left translations \cite{FollandHarmAna}.

	\begin{remark} \label{rem:LLRRareTensorProducts}
		The group transformations $\RR(g)$ and $\LL(g)$ are equivalent, according to \cref{prop:HBsturct}, to 
		$\oplus_k  (\II \otimes {D_\gamma^{r_k}(g)}) $ 
		and 
		$\oplus_k ({D_{\myinv{\gamma}}^{l_k}(g)} \otimes \II)$
		respectively, where ${D_\gamma^{j}(g)}$ are irreducible projective representations. Changing the basis of the physical Hilbert space (as in \cref{rem:changeOfBasis}) to 
		$\{\ket{l_k,m}\otimes\ket{r_k,n}\}$
		in which the representations take this block diagonal form,
		involves transforming  $B$ into $\tilde{B}$ given by the matrices:
		$\tilde{B}^{ k,m ,n} = \sum_i B^i \innerCG{l_k,m;r_k,n}{i}$.
		According to \cref{prop:ChangeOfBasisDoesntRuinCF} $\tilde{B}$ is also in CFII. 
		\Cref{eq:BRightGroupTransPre} holds for the new tensor under the action of the transformed operators:
		$ \tilde{\RR}(g)= \oplus_k (\II \otimes {D_\gamma^{r_k}(g)}) $ and 
		$\tilde{\LL}(g) =  \oplus_k ({D_{\myinv{\gamma}}^{l_k}(g)} \otimes \II)$.
		We shall always assume  $B$, $\LL(g)$ and $\RR(g) $ are in these forms. 
	\end{remark}
	\begin{remark}  \label{rem:BdoubleIndex}
		The simplest case of \cref{eq:BRightGroupTransPre} one could consider is when   
		$ \RR(g) = \II \otimes {D^{r}(g)}$ 
		and 
		$ \LL(g) ={D^{l}(g)} \otimes \II$,
		for irreducible projective representations ${D_{{\gamma}}^{r}(g)}$ and ${D_{\myinv{\gamma}}^{l}(g)}$. To these corresponds the basis $\{ \ket{m}\otimes\ket{n} \left| \right. m=1,\ldots,dim(l) ,n=1,\ldots,dim(r)  \}$, 
		and  the matrices composing the tensor $B$ are numbered by  two indices: 
		\begin{equation*}
			B^{m,n} = \sum_{\alpha,\beta} B^{m,n}_{\alpha,\beta} \ket{\alpha}\bra{\beta} \ .
		\end{equation*}
		$B$  transforms under  $\RR(g)$ and $\LL(g)$ in the following manner:
		\begin{equation*}
			\begin{split}
			\RR(g): B^{m,n}  &  \mapsto \sum_{n^\prime}{D_{{\gamma}}^r(g)}_{n,n^\prime}  B^{m,n^\prime} 
			= B^{m,n} X(g) \\
			\LL(g): B^{m,n}  & \mapsto \sum_{m^\prime} {D_{\myinv{\gamma}}^l(g)}_{m,m^\prime} B^{m^\prime,n} 
			= Y(g)^{-1} B^{m,n} \ .
			\end{split}
		\end{equation*}
	\end{remark}

	We have seen in \cref{rem:changeOfBasis} how to change the basis of the physical Hilbert space in order to bring the physical representations to block diagonal form. We would like to do the same for the virtual  projective representation $X(g)$ appearing in \cref{eq:BLRTGroupransformation}. This can be  achieved by a different transformation of the tensor $B$ described in the following:	
		
	\begin{remark} \label{rem:Xfreedom}
		Given $ B,\RR(g),\,\LL(g)$ and $X(g)$  that satisfy \cref{eq:BLRTGroupransformation}, redefine $B$: 
		\begin{equation*}
			B^{k;m,n}\mapsto \tilde{B}^{k;m,n}=V^{-1}B^{k;m,n} V \ ,
		\end{equation*} 
		with any invertible matrix $V$. The new tensor $\tilde{B}$ generates the same MPV and transform as in \cref{eq:BLRTGroupransformation} with $X(g)$   replaced by  $\tilde{X}(g)=V^{-1}X(g)V $.
	\end{remark}
		
	\begin{remark} \label{rem:XshangeOfBasisOnlyAfterCF}
		Note that the transformation described in \cref{rem:Xfreedom} may ruin the CF property of $B$, as $V$ does not in general preserve $B$'s block structure (\cref{eq:BsBlockStructure}). We shall therefore take care to use this freedom of choosing the basis of $X(g)$ only when we no longer intend to use the  CF property.
	\end{remark}
		
	\Cref{rem:Xfreedom} allows us to assume without loss of generality $X(g)$ takes the form $\oplus_aX^a(g)$, 
	where $X^a(g)$ are irreducible projective representations. Next we project \cref{eq:BLRTGroupransformation} to the $k$ sector of the physical Hilbert space (\cref{rem:LLRRareTensorProducts}) and to the $(a,b)$ block in the virtual space, since the representations are block diagonal they commute with the projection operators for every group element $g\in G$. We therefore obtain:
	\begin{equation} \label{eq:ElementaryBblockDef}
		\begin{tikzpicture}[baseline=-1mm]
		\pic[pic text = $B^k_{a,b}$] at (0,0) {mpstensorBIG};
		\pic[pic text = $\II\otimes D_\gamma^{r_k}(g)$] at (0,1) {verticalmatrixSuperWIDE};
		\end{tikzpicture} \  = \ 
		\begin{tikzpicture}[baseline=-1mm]
		\pic[pic text = $B^k_{a,b}$] at (0,0) {mpstensorBIG};
		\pic[pic text =  $X^b(g)$] at (1.3,0) {horizontalmatrixWIDE};
		\end{tikzpicture} \;\;\;\;\;\;\;\;\; ;  \;\;\;\;\;\;\;\;\; 
		\begin{tikzpicture}[baseline=-1mm]
		\pic[pic text = $B^k_{a,b} $] at (0,0) {mpstensorBIG};
		\pic[pic text = $D_{\myinv{\gamma}}^{l_k}(g)\otimes \II$] at (0,1) {verticalmatrixSuperWIDE};
		\end{tikzpicture} \   = \ 
		\begin{tikzpicture}[baseline=-1mm]
		\pic[pic text =  $ \myinv{X ^a{(g)}}$] at (0,0) {horizontalmatrixWIDE};
		\pic[pic text = $B^k_{a,b}$] at (1.3,0) {mpstensorBIG}; 
		\end{tikzpicture} \ ,
	\end{equation} 
	 where $B^k_{a,b}$ is the tensor that  consists of the $(a,b)$ blocks of the matrices   $B^{k;m,n}$. 
 
	 The reduction procedure described above   motivates the following definition of an elementary $B$ block. Next we shall show that the irreducible representations appearing in  \cref{eq:ElementaryBblockDef} determine such blocks up to a constant.
	\begin{restatable}{definition}{ElementaryBblock}  \label{def:ElementaryBblock}
		An elementary block of the tensor $B$ is one which satisfies \cref{eq:BRightGroupTransPre},  where $ \RR(g) = \II \otimes {D_\gamma^{r}(g)}$,		
		$ \LL(g) ={D_{\myinv{\gamma}}^{l}(g)} \otimes \II$ 
		and  $X(g)$, $Y(g)$, $D_\gamma^{r}(g)$ and  $D_{\myinv{\gamma}}^{l}(g)$ are  irreducible projective representations  (both $X(g)$ and $Y(g)$ have multiplier $\gamma$).
	\end{restatable}
 
	 \begin{restatable} [Structure of an elementary $B$ block] {propositionRE}{ElemBBlock}
	 	\label{prop:ElemBBlock}
	 	Let  $B$ be an elementary $B$ block (\cref{def:ElementaryBblock}). If $X(g)={D_\gamma^r(g)}$ and $ \overline{Y(g)}={{D_{\myinv{\gamma}}^l(g)}}$, then
	 	$B$  is proportional to the tensor composed of the matrices
	 	\begin{equation*}
		 	B^{m,n}=\ket{m}\bra{n}  \ ,
		 	m=1,\ldots,dim(l), n=1,\ldots,dim(r) 
		 	\ . 
	 	\end{equation*}
	 	Otherwise $B=0$.
	 \end{restatable}

	We have thus classified all tensors $B$ that satisfy \cref{eq:BLRTGroupransformation}. There is however more information to be extracted from \cref{thm:OnlyFieldGaugeGroup}. According to   \cref{prop:ElemBBlock},  when  projected   to sectors corresponding to  inequivalent representations, the tensor $B$ is zero. This result, combined with the assumption that $B$ is in CF imposes relations between the irreducible  representations that comprise $\RR(g),\,\LL(g)$ and $X(g)$:

	\begin{restatable}{propositionRE}{BstructGeneral}
		\label{prop:BstructGeneral}
		Let $B,\RR(g),\LL(g)$ and $X(g)$ be as in \cref{thm:OnlyFieldGaugeGroup}. Let ${X_j}(g) =\oplus_a X_j^a(g)$ be a block  of $X(g)$ appearing in \cref{eq:XStruct1}, consisting of irreducible projective representations $X_j^a(g)$. 
		Let  
		$\RR(g) = \oplus_k (\II \otimes {D_\gamma^{r_k}(g)})$ 
		and 
		$\LL(g) = \oplus_k ( {D_{\myinv{\gamma}}^{l_k}(g)} \otimes \II)$, 
		where $D_\gamma^{r_k}$ and  $D_{\myinv{\gamma}}^{l_k}$ are irreducible projective representations. Then the following hold:
		\begin{enumerate}
			\item
				For all $k$  either there exist  $  a$ and $b$ such that 
				$X_j^b(g) = {D_\gamma^{r_k}(g)} $ and  
				$\overline{ X_j^a(g)} = {{D_{\myinv{\gamma}}^{l_k}(g)}}$, or the projection of the corresponding tensor $B_j$ (a BNT element of $B$) to the sector $k$ of the physical space is zero.
			\item
				$ \forall a$ $\exists k$ such that $\overline{X_j^a(g)} = {{D_{\myinv{\gamma}}^{l_k}(g)}}$.
			\item
				$\forall a$  $\exists k$ such that $X_j^a(g) = {{D_\gamma^{r_k}(g)}}$.
		\end{enumerate} 
	\end{restatable}

	The elementary block of $B$ described in \cref{prop:ElemBBlock} is the same as the one used in \cite{Zohar:2015jnb}. Note that even in lattices of higher dimensionality  each  gauge field degree of freedom still connects two lattice sites. There:
	\begin{equation}\label{eq:BErez}
		B^{j;m,n} = \beta_j \ket{j,m}\bra{j,n} \ ,
	\end{equation}
	where $\beta_j$ are arbitrary constants.
	The overall structure of the $B$ tensor  derived above admits more general structures than \cref{eq:BErez}; these  structures are recovered if for example, all blocks $X_j(g)$ appearing in $X(g)$ (\cref{eq:XStruct1}) are irreducible representations. In this case (since in \cref{prop:BstructGeneral} the index $a$ can assume only one value),  for all $k$  ${D_{\myinv{\gamma}}^{l_k}(g)} = \overline{D_\gamma^{r_k}(g)}$ and $\mathcal{H}_B$ takes the K-S form, as in \cref{eq:KSH}.

	In the following two propositions we consider adding a matter degree of freedom to a  gauge field MPV with a local symmetry. We show that it is always possible to find a tensor $A$  and a unitary representation $\CC(g)$ (non-trivial ones) that couple to  it:
	\begin{restatable} {propositionRE}{NonTrivAcoupling} \label{prop:invGauge1}
		Let $B$ be in CFII and let $\ket{\psi^N_B}$ have a local symmetry with respect to  $\RR(g)\otimes\LL(g)$ (as in \cref{thm:OnlyFieldGaugeGroup}). It is always possible to find a tensor $A$ and a representation $\CC(g)$ such that the corresponding matter and gauge field MPV $\ket{\psi^N_{AB}}$ has a local symmetry with respect to $\RR(g) \otimes\CC(g) \otimes \LL(g)$  (\cref{def:BABSymm}).
		 In addition, the corresponding   matter MPV - $\ket{\psi^N_A}$ - has a global symmetry with respect to $\CC(g)$.
	\end{restatable}
	For a  restricted class of   $B$ tensors, \textit{any}   $A$ and  $\CC(g)$ that couple to it (satisfy \cref{def:BABSymm}) will have a global symmetry:

	\begin{restatable}{propositionRE}{InverseGaugeNES} \label{prop:InverseGaugeNES}
		Let $B,\, \RR(g)$ and $\LL(g)$ be as in \cref{thm:OnlyFieldGaugeGroup} and in addition let  $ span\{ B^{k;m,n} \left|\right. k,m,n \} $ contain the identity matrix  (e.g.\ \cref{eq:BErez}). Let $A$ and $\CC(g)$ be such that  the MPV generated by $AB$ has a local symmetry with respect to $\RR(g)\otimes\CC(g)\otimes\LL(g)$ (\cref{def:BABSymm}). Then $\ket{\psi_A^N}$ has a global symmetry with respect to  $\CC(g)$. If in addition $A$ is in CF with the same block structure as $B$ (\cref{eq:BsBlockStructure}), then $A$ transforms as:			
		\begin{equation*} 
			\begin{tikzpicture}[baseline=-1mm]
			\pic[pic text = $A$] at (0,0) {mpstensorBIG};
			\pic[pic text = $\CC(g)$] at (0,1) {verticalmatrixBIG};
			\end{tikzpicture} \ = \ 
			\begin{tikzpicture}[baseline=-1mm]
			\pic[pic text = $\myinv{{X}(g)}$] at (0,0) {horizontalmatrixWIDE};
			\pic[pic text = $A$] at (1.15,0) {mpstensorBIG};
			\pic[pic text = $X(g)$] at (2.3,0) {horizontalmatrixWIDE};
			\end{tikzpicture} \ ,
		\end{equation*} 
		with the same $X(g)$ from \cref{thm:OnlyFieldGaugeGroup}.
	\end{restatable}

	The MPVs described above may be combined in a way that allows coupling matter and gauge fields such that each of them could be invariant on its own, as in the conventional well known scenarios of gauge theories. However, as we shall demonstrate in the next section, this is not the most general setting of a local symmetry involving these two building blocks.

\subsection{Matter and gauge field MPV} \label{sub:Summary_MatterAndGaugeField}
 	
	We show that a local symmetry for a combined matter and  gauge field MPV $\ket{\psi^N_{AB}}$ (defined in \cref{sub:MandGFSetting}) generated by tensors $A$ and $B$ (in an appropriate form), 
	 implies the following transformation relations for $A$ and $B$:
	\begin{equation*}  
		\begin{tikzpicture}[baseline=-1mm]
		\pic[pic text = $B$] at (0,0) {mpstensorBIG};
		\pic[pic text = $\RR(g)$] at (0,1) {verticalmatrixBIG};
		\end{tikzpicture} \  = \ 
		\begin{tikzpicture}[baseline=-1mm]
		\pic[pic text = $B$] at (0,0) {mpstensorBIG};
		\pic[pic text =  $X(g)$] at (1,0) {horizontalmatrixBIG};
		\end{tikzpicture} \;\;\;\;\;\;\;\;\; ;  \;\;\;\;\;\;\;\;\; 
		\begin{tikzpicture}[baseline=-1mm]
		\pic[pic text = $B $] at (0,0) {mpstensorBIG};
		\pic[pic text = $\LL(g)$] at (0,1) {verticalmatrixBIG};
		\end{tikzpicture} \   = \ 
		\begin{tikzpicture}[baseline=-1mm]
		\pic[pic text =  $\myinv{Y{(g)}}$] at (0,0) {horizontalmatrixWIDE};
		\pic[pic text = $B$] at (1.15,0) {mpstensorBIG};
		\end{tikzpicture}
	\end{equation*}

	\begin{equation}  \label{eq:AGroupTransXXX} 
		\begin{tikzpicture}[baseline=-1mm]
		\pic[pic text = $A$] at (0,0) {mpstensorBIG};
		\pic[pic text = $\CC(g)$] at (0,1) {verticalmatrixBIG};
		\end{tikzpicture} \ = \ 
		\begin{tikzpicture}[baseline=-1mm]
		\pic[pic text =  $\myinv{X{(g)}}$] at (0,0) {horizontalmatrixWIDE};
		\pic[pic text = $A$] at (1.15,0) {mpstensorBIG};
		\pic[pic text =  $Y(g)$] at (2.3,0) {horizontalmatrixWIDE}; \ ,
		\end{tikzpicture}
	\end{equation} 
	where  $X(g)$ and $Y(g)$ are projective representations from the same cohomology class. 
	As described in the previous section, the relation for $B$ allows to infer the structure of the Hilbert space $\mathcal{H}_B$ associated with the gauge field degree of freedom. As before, $\mathcal{H}_B$ splits into right and left parts. The structure of the tensor $B$ can be derived in the same way as in the previous section. 	Each elementary block of the tensor $A$, obtained by projecting \cref{eq:AGroupTransXXX} to irreducible representation spaces, satisfies a vector operator relation, and is therefore determined by the Wigner-Eckart theorem (\cref{thm:GenWigEck}). 
	
	In the general case, the structure described in this section allows for ``unconventional'' gauge symmetries where a local symmetry exists for the matter and gauge field MPV but none of the constituents has a symmetry on its own, i.e., the gauge field MPV does not have a local symmetry and the matter MPV does not have a global one. We construct an explicit example of such a case (see \cref{ex:locSymNotGlob}).
		 	 	
 	Finally we use the known results about global symmetries in MPV  \cite{Sanz2009} to find a class of matter MPVs with a global symmetry that can be gauged by adding a gauge field degree of freedom. We shall now state the above results in detail.
	
 \bigskip

	Let $\ket{\psi^N_{AB}}$ be a MPV generated by tensors $A$ and $B$ and let $\RR(g),\,\CC(g)$ and $\LL(g)$ be  as defined in \cref{sub:MandGFSetting}. 	 

 	\begin{restatable}[Matter and gauge field MPV with a local symmetry]{theoremRE}{VirtualGroupTrans} \label{thm:VirtualGroupTrans}
	 	Let both $BA$ and $AB$ be normal tensors in CFII and let
	 	 $\CC(g)$ and   $\RR(g),\LL(g)$
	 	be unitary and projective   representations (with inverse multipliers)  of a group $G$ respectively. Let $\ket{\psi^N_{AB}}$ be a  MPV with a local symmetry with respect to $\RR(g)\otimes\CC(g)\otimes\LL(g)$ (\cref{def:BABSymm}).
	 	 Then there exist  projective   representations $X(g)$ and $Y(g)$ on $\mathbb{C}^{D_1}$ and $\mathbb{C}^{D_2}$ respectively, such that $X(g)$ has the same multiplier as $\RR(g)$, and $Y(g)$ - the inverse multiplier to that of $\LL(g)$. The tensors  $A$ and $B$ transform as follows:
	 	\begin{equation} \label{eq:BRightGroupTrans}
		 	\begin{tikzpicture}[baseline=-1mm]
		 	\pic[pic text = $B$] at (0,0) {mpstensorBIG};
		 	\pic[pic text = $\RR(g)$] at (0,1) {verticalmatrixBIG};
		 	\end{tikzpicture} \  = \ 
		 	\begin{tikzpicture}[baseline=-1mm]
		 	\pic[pic text = $B$] at (0,0) {mpstensorBIG};
		 	\pic[pic text =  $X(g)$] at (1,0) {horizontalmatrixBIG};
		 	\end{tikzpicture} \;\;\;\;\;\;\;\;\; ;  \;\;\;\;\;\;\;\;\; 
		 	\begin{tikzpicture}[baseline=-1mm]
		 	\pic[pic text = $B $] at (0,0) {mpstensorBIG};
		 	\pic[pic text = $\LL(g)$] at (0,1) {verticalmatrixBIG};
		 	\end{tikzpicture} \   = \ 
		 	\begin{tikzpicture}[baseline=-1mm]
		 	\pic[pic text =  $\myinv{Y{(g)}}$] at (0,0) {horizontalmatrixWIDE};
		 	\pic[pic text = $B$] at (1.15,0) {mpstensorBIG};
		 	\end{tikzpicture}
	 	\end{equation}

	 	\begin{equation} \label{eq:AGroupTrans} 
		 	\begin{tikzpicture}[baseline=-1mm]
		 	\pic[pic text = $A$] at (0,0) {mpstensorBIG};
		 	\pic[pic text = $\CC(g)$] at (0,1) {verticalmatrixBIG};
		 	\end{tikzpicture} \ = \ 
		 	\begin{tikzpicture}[baseline=-1mm]
		 	\pic[pic text =  $\myinv{X{(g)}}$] at (0,0) {horizontalmatrixWIDE};
		 	\pic[pic text = $A$] at (1.15,0) {mpstensorBIG};
		 	\pic[pic text =  $Y(g)$] at (2.3,0) {horizontalmatrixWIDE};
		 	\end{tikzpicture}
	 	\end{equation} 
	\end{restatable}	 

	 In the following  proposition we show that given arbitrary tensors $A$ and $B$, generating a MPV $ \ket{\psi^N_{AB}} $, it is possible to describe the same MPV as a linear combination of MPVs that satisfy the normality condition in \cref{thm:VirtualGroupTrans}:
	  	  
	\begin{restatable}{propositionRE}{RedToNormal}\label{prop:RedToNormal}
	  	Let $\ket{\psi^N_{AB}}$ be a MPV  generated by arbitrary tensors $A$ and $B$. Then there exist tensors $\{A_\chi\}$ and $\{B_\chi\}$, and there exists  $b\in \mathbb{N} $ such that for all $\chi$ both $A_\chi B_\chi$ and $B_\chi A_\chi$ are normal tensors and  $\forall N\in \mathbb{N}$ $\ket{\psi^{ N}_{AB_{\times b}}}= \sum_\chi \mu_\chi^N \ket{\psi^N_{A_\chi B_\chi}}$, where $\mu_\chi$ are constants and $AB_{\times b}$ is the tensor obtained by blocking $b$ copies of the tensor $AB$.	  	
	\end{restatable}

 	Next we show that if  $\ket{\psi^N_{AB}} = \sum_\chi \mu^N_\chi \ket{\psi^N_{A_\chi B_\chi}}$  has a local symmetry with respect to   $\RR(g)\otimes\CC(g) \otimes\LL(g)$, then every normal component $\ket{\psi^N_{A_\chi B_\chi}}$ must have the same symmetry. We can then apply \cref{thm:VirtualGroupTrans} to each of the components.
	 	
 	\begin{restatable}{propositionRE}{EveryNormCompInv} \label{prop:EveryNormCompInv}
 		Let $\ket{\psi^N_{AB}} = \sum_\chi \mu^N_\chi\ket{\psi^N_{A_\chi B_\chi}}$ where   both $A_\chi B_\chi$ and $B_\chi A_\chi$ are normal tensors. Let $O$ be a local operator acting on a fixed number of adjacent sites. If $\forall N$ $O$ leaves the MPV invariant:
 		\begin{equation*}
	 		O\otimes \II|_{rest} \ket{\psi^N_{AB}} = \ket{\psi^N_{AB}} \ , 
 		\end{equation*} 
 		then $O$ leaves every component invariant: 
 		\begin{equation*}
	 		O\otimes \II|_{rest} \ket{\psi^N_{A_\chi B_\chi}} = \ket{\psi^N_{A_\chi B_\chi}}\, \forall \chi \ . 
 		\end{equation*} 
 	\end{restatable}
	 	
	 Having derived  \cref{eq:BRightGroupTrans}, \cref{prop:HBsturct} can be applied to determine the structure of the Hilbert space $\mathcal{{H}}_B$. 
	 As in the case of a gauge field MPV discussed in the previous section, we are free to assume  $X(g)$ and $Y(g)$ are block diagonal in irreducible representations: 
 
	 \begin{remark} \label{rem:XYfreedom}
	 	In \cref{thm:VirtualGroupTrans} we are free to choose similarity transformations for $X(g)$ and $Y(g)$ independently. Given $A,B,\RR(g),\CC(g),\LL(g),X(g)$ and $Y(g)$ that satisfy \cref{eq:BRightGroupTrans}  and \cref{eq:AGroupTrans}  we can redefine $A$ and $B$: 
	 	\begin{equation*}
		 	A^{j,m}\mapsto \tilde{A}^{j,m}=U^{-1} A^{j,m} V \ , 
		 	\;\;\; B^{k;m,n}\mapsto \tilde{B}^{k;m,n}=V^{-1}B^{k;m,n} U \ ,
	 	\end{equation*} 
	 	with any invertible matrices  $U$ and $V$ of fitting dimensions. The new tensors generate the same MPV $\ket{\psi^N_{AB}}$ and transform as in \cref{thm:VirtualGroupTrans} with $X(g)$ and $Y(g)$ replaced by  $\tilde{X}(g)=U^{-1}X(g)U $ and $\tilde{Y}(g)=V^{-1}Y(g)V$.
	 \end{remark}

	 \begin{restatable}[Elementary $A$ block]{definition}{ElementaryAblock} \label{def:ElementaryAblock}
	 	An elementary block of the tensor $A$ is one which satisfies \cref{eq:AGroupTrans},  where   $ \CC(g),\,X(g)$ and $Y(g)$  are all    irreducible projective representations.
	 \end{restatable}
	 By bringing all of the representations appearing in \cref{eq:BRightGroupTrans}  and \cref{eq:AGroupTrans} to block diagonal  form (using \cref{rem:changeOfBasis} on the physical representations and \cref{rem:XYfreedom} on the virtual ones), and projecting \cref{eq:BRightGroupTrans}  and \cref{eq:AGroupTrans} to irreducible sectors of the physical and virtual Hilbert spaces (as explained in \cref{sub:pureGFSummary}), we may reduce \cref{eq:BRightGroupTrans}  and \cref{eq:AGroupTrans} to the cases of elementary blocks of $B$ and of $A$ respectively. 
	 
	 We have seen in \cref{sub:pureGFSummary} that \cref{eq:BRightGroupTrans} determines the tensor $B$ given $\RR(g),\,\LL(g),\,X(g)$ and $Y(g)$ (\cref{prop:BstructGeneral}).
	 We now show that   \cref{eq:AGroupTrans} determines the tensor $A$ given $\CC(g),\,X(g)$ and $Y(g)$.

	 \begin{restatable} {propositionRE}{WigEckClassA}
	 	\label{lem:WigEckClassA}
	   Let $A$ be an elementary block (\cref{def:ElementaryAblock}), with ${\CC(g) = D^{J_0}(g)},\,X(g) ={{D_\gamma^j(g)}}$ and 
	   $ Y(g)={D_{\myinv{\gamma}}^l(g)}$. Then $A$  is built out of Clebsch-Gordan coefficients and has the form:
 		\begin{equation*}
	 		{A^M} = \sum_{J\in \mathfrak{J}:{D^J}={D^{J_0}}} {\alpha_J} \sum_{m,n}\innerCG{J,M}{\overline{j},m;l,n} \ket{m}\bra{n} \ ,
 		\end{equation*}	 
			where $\mathfrak{J}$ is the set of irreducible representation  indices appearing in the decomposition of $\overline{D_\gamma^j(g)}\otimes {D_{\myinv{\gamma}}^l(g)}$ into irreducible representations, $\innerCG{\overline{j},m:{l},n}{J,M}$ are the Clebsch-Gordan coefficients of the decomposition,
			$\overline{{D_\gamma^j(g)}}$ is the complex conjugate representation to ${D_\gamma^j(g)}$ and $\alpha_J$ are arbitrary constants. 	
	 \end{restatable}
	 \Cref{lem:WigEckClassA} was shown in \cite{Sanz2009} in the context of MPS with a global symmetry.

	  The relation between the irreducible projective  representations appearing in $\RR(g)$ ($\LL(g)$) and $X(g)$ ($Y(g)$) is characterized by the following:
	   
	  \begin{restatable}{propositionRE}{RXLYconnectionNormal} 		\label{prop:RXLYconnectionNormal}
			Let $AB$ and $BA$ be normal tensors and let $B$ satisfy \cref{eq:BRightGroupTrans}   with $\RR(g) = \oplus_k (\II \otimes {D_\gamma^{r_k}(g)})$, $\LL(g) = \oplus_k ( {D_{\myinv{\gamma}}^{l_k}(g)} \otimes \II)$, $Y(g) = \oplus_a Y^a(g)$ and $X(g)=\oplus_b X^b(g)$, where $D_\gamma^{r_k}$, $D_{\myinv{\gamma}}^{l_k}$, $Y^a$ and $X^b$ are irreducible projective representations, 
			then
			\begin{enumerate}
				\item
					For all $ k$  either there exist  $a$ and $b$  such that 
					$X^b(g) = {D_\gamma^{r_k}(g)} $ and  
					$ \overline{Y^a(g)} = {{D_{\myinv{\gamma}}^{l_k}(g)}}$ or the projection of the tensor $B$ to the sector $k$ of the physical space is zero (and it can be discarded).
				\item
					$ \forall a$ $\exists k$ such that $\overline{Y^a(g)} = {{D_{\myinv{\gamma}}^{l_k}(g)}}$.
				\item
					$\forall b$  $\exists k$ such that $X^b(g) = {{D_\gamma^{r_k}(g)}}$.
			\end{enumerate} 
	  \end{restatable}

	  By constructing tensors $A$ and $B$ that transform as in \cref{thm:VirtualGroupTrans} with $X(g)\neq Y(g)$ we show the existence of matter and gauge field MPVs  which have a local symmetry but for which the corresponding matter MPV does not have a global symmetry, nor does the gauge field MPV have a local one:

	 \begin{restatable}{propositionRE}{LocalSymmWOGlobal}
	 	 \label{ex:locSymNotGlob}
	 	There exist tensors $A$ and $B$ such that $\ket{\psi_{AB}}$ has a local symmetry with respect to $\RR(g)\otimes\CC(g)\otimes\LL(g)$, but $\ket{\psi_A}$ does not have a global symmetry with respect to $\CC(g)$. In addition   $\RR(g)\otimes\LL(g) \ket{\psi_{B}} \neq \ket{\psi_{B}} $.
	 \end{restatable}
	
	We review known results about MPV with global symmetry \cite{Sanz2009}. Let $A$ be a tensor in CFII:
   	\begin{equation} \label{eq:AsBlockStructure}
	   A^i=\oplus_{j=1}^n \oplus_{q=1}^{r_j} \mu_{j,q} A^i_j  \ ,
   	\end{equation}
   	where $\{A_j\}$ are  normal tensors in CFII forming a BNT of $A$ (\cref{def:BNT}) and $\mu_{j,q}$ are constants.

	\begin{restatable}{theoremRE}{GlobalSymmForMPVinCF} \label{thm:GlobalSymmVirtTrans}
	   	A tensor $A$ in CFII which generates a MPV with a  global symmetry with respect to  a representation $\CC(g)$ of a connected Lie group $G$, transforms under the representation matrix as:
	   	\begin{equation} \label{eq:AGlobalSymm}
		   	\begin{tikzpicture}[baseline=-1mm]
		   	\pic[pic text = $A$] at (0,0) {mpstensorBIG};
		   	\pic[pic text = $\CC(g)$] at (0,1) {verticalmatrixBIG};
		   	\end{tikzpicture} \ = \
		   	\begin{tikzpicture}[baseline=-1mm]
		   	\pic[pic text =  $\myinv{X (g)}$] at (0,0) {horizontalmatrixWIDE}; 
		   	\pic[pic text = $A$] at (1.15,0) {mpstensorBIG};
		   	\pic[pic text = $X(g)$] at (2.3,0) {horizontalmatrixWIDE};	   
		   	\end{tikzpicture} \ ,
	   	\end{equation} 
	   	where  $X(g)$ has the same block structure as $A$:
	   	\begin{equation} \label{eq:XStruct2}
		   	X(g) = \oplus_{j=1}^m \oplus_{q=1}^{r_j} {X_j}(g) \ ,
	   	\end{equation}
	   	 and where each block $X_j(g)$ is a projective   representation,  in the general case, for different $j$ values $X_j(g)$ belong to different cohomology classes.
	\end{restatable}

	In the case when all $X_j(g)$ obtained in \cref{thm:GlobalSymmVirtTrans} are from  the same cohomology class, we can find a gauge field  tensor $B$ and projective  representations $\RR(g)$ and $\LL(g)$  that gauge the symmetry:

	\begin{restatable}{propositionRE}{GaugeGlobSymm} \label{thm:GaugeGlobSymm}
		Let $A$ be a tensor in CFII generating a MPV with a global symmetry i.e., satisfying \cref{thm:GlobalSymmVirtTrans}. Let $X(g)$ (in \cref{eq:AGlobalSymm}) be a projective representation (i.e.\ all $X_j(g)$ in \cref{eq:XStruct2} are in the same cohomology class). Then there exist a tensor $B$ and projective representations $\RR(g)$ and $\LL(g)$ with inverse multipliers such that both local symmetries:  \cref{def:BABSymm} for $\ket{\psi^N_{AB}}$ and \cref{def:BBSymm} for $\ket{\psi^N_B}$ are satisfied.
	\end{restatable}

\section{MPV with a global symmetry} \label{sec:GlobSymm}
	In the next section we shall  present the derivation of the  previously described results. Before that, however, we review MPVs basics not covered in the formalism section, needed for the  derivation of of the classification of MPVs with a global symmetry, originally shown in \cite{Sanz2009}. In order for the paper to be self contained, we derive the result   from the fundamental theorem of MPV (see \cref{thm:fundMPV}), following \cite{Cirac2017} and references therein.

	\begin{proposition} \label{prop:InjInv}
		\cref{def:Injectivity} is equivalent to the existence of a one-sided inverse  tensor $A^{-1}$ which satisfies:
		\begin{equation*}
			\begin{tikzpicture}[baseline=-1mm]
			\pic[pic text = $A$] at (0,0) {mpstensor};
			\pic [rotate=180] at (0,1) {mpstensor};
			\node at (0,1) {$\myinv{A}$};
			\end{tikzpicture} \ = \ 
			\begin{tikzpicture}[baseline=-1mm]
			\draw (0,1)--(0.25,1)--(0.25,0)--(0,0);
			\draw (1,1)--(0.75,1)--(0.75,0)--(1,0);
			\end{tikzpicture} \ ,
		\end{equation*}
		that is:
		\begin{equation*}
			\sum_i A^i_{\alpha\beta} {(A^{-1})}^i_{\alpha\beta} = \delta_{\alpha,\alpha^\prime}\delta_{\beta,\beta^\prime} 
		\end{equation*} 
	\end{proposition}

	\begin{definition} [Span of matrix products] \label{def:SpanOfProducts}
		For a tensor $A$ with bond dimension $D$ let  $S_L \subseteq \mathcal{M}_{D\times D} $ be the space spanned by all possible $L$-fold matrix products:
		\begin{equation*}  
			\mathcal{S}_L:= span\left\{ A^{i_1}A^{i_2}\ldots A^{i_L}\left|\right.{i_1,i_2,\ldots,i_L}=1,\ldots, d  \right\} \ .
		\end{equation*}
	\end{definition}

	\begin{definition} \label{def:GammaMap}
		Let 	$\Gamma^L_A : \mathcal{M}_{D\times D} \rightarrow \mathcal{H}^{\otimes L}$ be defined by:
		\begin{equation*} \label{eq:GammaDef}
			\Gamma^L_A (X) = \sum \tr\left(XA^{i_1}A^{i_2}\ldots A^{i_L}\right)\ket{i_1i_2\ldots i_L} \ .
		\end{equation*}
	\end{definition}
	For a normal tensor, according to \cref{def:NormalTensor}, for $L$ large enough, $S_L =\mathcal{M}_{D\times D}$. For tensors in CF the following holds:  

	\begin{proposition} [Span property of BNT] \label{prop:SpanMatAlg}
		Let $A$ be in CF with each block being a unique element of its BNT, i.e.\ there is no $q$ summation in \cref{eq:BNTexpansion}.
		Then for $L$ large enough, $\mathcal{S}_L$		is the entire matrix algebra  
		$\mathcal{M}:=\oplus_{j=1}^m\mathcal{M}_{D_j\times D_j} $
		where $\mathcal{M}_{D_j\times D_j}$ is the algebra of $D_j\times D_j$ matrices and $D_j$ is the bond dimension of $A_j$. \textup{\cite{Perez-Garcia2007}}
	\end{proposition}

	\begin{proposition} \label{prop:SpanAndInjectivity}
		Let $A$ be a tensor consisting of  block diagonal matrices: $A^i \in \mathcal{M}:= \oplus^m_j \mathcal{M}_{D_j\times D_j}$, and let $\mathcal{S}_L$ and $\Gamma^L_A$ be as in \cref{def:SpanOfProducts} and \cref{def:GammaMap} respectively. Then $\mathcal{S}_L=\mathcal{M}$ iff $\Gamma^L_A|_{\mathcal{M}}$ is injective.
	\end{proposition}

	\begin{proof}
		Assume injectivity of $\Gamma^L_A|_{\mathcal{M}}$, then any element $X\in \mathcal{S}^\perp \cap \mathcal{M}$   satisfies $\Gamma^L_A(X^\dagger)=0$ because the coefficients of the the vector $\Gamma^L_A(X^\dagger)$ are inner products of $X$ with elements in $\mathcal{S}$.  This implies $X=0$. If $\mathcal{S}=\mathcal{M}$, then for every non zero $X\in \mathcal{M}$, $X^\dagger$ has  a non vanishing inner product with at least one element $A^{i_1}A^{i_2}\ldots A^{i_L}$, and therefore $\Gamma^L_A(X)$ is non zero.
	\end{proof}

	\begin{proposition} \label{cor:SpanForCF}
		For a tensor $A$ in CF as in \cref{eq:BNTexpansion}, for $L$ large enough
		the space $\mathcal{S}_L$  (\cref{def:SpanOfProducts}) has the form:
		\begin{equation} \label{eq:SLCF}
			\mathcal{S}_L=\left\{ \oplus_{j=1}^m\oplus_{q=1}^{r_j}\mu^L_{j,q}V_{j,q}^{-1}M_jV_{j,q} \left| \right. M_j \in \mathcal{M}_{D_j\times D_j} \right\}
		\end{equation}
	\end{proposition}

	\begin{proof}
		Consider a tensor $\tilde{A}$ which consists of the BNT of $A$ without multiplicities (as in \cref{prop:SpanMatAlg}). An element  $S=\oplus_{j=1}^m\oplus_{q=1}^{r_j}\mu^L_{j,q}V_{j,q}^{-1}M_jV_{j,q}$  in $\mathcal{S}_L$ is obtained by taking the same linear combination of the matrix products  $A^{i_1}A^{i_2}\ldots A^{i_L}$ as the one which generates $\tilde{S}= \oplus_{j=1}^m M_j$  from  the matrix products $\tilde{A}^{i_1}\tilde{A}^{i_2}\ldots \tilde{A}^{i_L}$.
	\end{proof}

	\begin{proposition} \label{prop:OrthogBNT}
		Let $\{A_j\}_{j=1}^m$ be a BNT of $A$, and let each $A_j$ appear in $A$ with no multiplicities, i.e.\ $A^i=\oplus_{j=1}^m \nu_j A_j^i$. For $L$ large enough the image of the algebra of block diagonal matrices 
		$\mathcal{M} :=\oplus_{j=1}^m \mathcal{M}_{D_j\times D_j}$,
		where $D_j$ is the bond dimension of $A_j$,
		under the map $\Gamma_A^L$ is a direct sum: 
		\begin{equation*}
			\Gamma_A^L \left( \mathcal{M} \right) := 
			\left\{ \Gamma_A^L(X) \left|\right. X\in  \mathcal{M}  
			\right\} = 
			\bigoplus_{j=1}^m  \Gamma_{A_j}^L (\mathcal{M}_{D_j\times D_j}) \ .
		\end{equation*} 
		In particular 
		$
		\sum_{j=1}^m \Gamma_{A_j}^L(X_j) = 0$ implies $ X_j=0 \;\;\forall j = 1,\ldots,m$. 
		\textup{\cite{Perez-Garcia2007}}
	\end{proposition}
	
	\cref{prop:OrthogBNT} allows us to prove the  following lemma:
	\begin{lemma}  \label{lem:BAAAA_CFequ}
		Let $A$ be a tensor in CF with BNT $\{A_j\}$,  and let  $S$ and $T$ be tensors with the exact same block structure as $A$:
		\begin{equation*} 
			A^i=\oplus_{j=1}^m\oplus_{q=1}^{r_j} \mu_{jq} V_{j,q}^{-1} A^i_j V_{j,q}  
		\end{equation*}
		\begin{equation*}
			S^i=\oplus_{j=1}^m\oplus_{q=1}^{r_j} \mu_{jq} V_{j,q}^{-1} S^i_j V_{j,q}  
		\end{equation*}
		\begin{equation*}
			T^i=\oplus_{j=1}^m\oplus_{q=1}^{r_j} \mu_{jq} V_{j,q}^{-1} T^i_j V_{j,q}  \ .
		\end{equation*}
		If the following equality holds for any length $N$:
		\begin{equation} \label{eq:STSTST}
			\sum_{\{i\}} \tr\left(S^{i_1}A^{i_2}\ldots A^{i_N} \right) \ket{i_1,i_2,\ldots,i_N} = 
			\sum_{\{i\}} \tr\left(T^{i_1} A^{i_2}\ldots A^{i_N} \right) \ket{i_1,i_2,\ldots,i_N} \ ,
		\end{equation}
		which in tensor notation reads:
		\begin{equation*}
			\begin{tikzpicture}[baseline=-1mm]
			\foreach \i in {1,2,4} \pic[pic text = $A$] at (\i,0) {mpstensor};
			\foreach \i in {0} \pic[pic text = $S$] at (\i,0) {mpstensor};
			\node at (3,0) {$\ldots$}; 
			\draw (-.5,0)--(-.5,-.5)--(4.5,-.5)--(4.5,0);
			\end{tikzpicture} \ = \ 
			\begin{tikzpicture}[baseline=-1mm]
			\foreach \i in {1,2,4} \pic[pic text = $A$] at (\i,0) {mpstensor};
			\foreach \i in {0} \pic[pic text = $T$] at (\i,0) {mpstensor};
			\node at (3,0) {$\ldots$}; 
			\draw (-.5,0)--(-.5,-.5)--(4.5,-.5)--(4.5,0);
			\end{tikzpicture}  \ ,
		\end{equation*} 
		then $S=T$.
	\end{lemma}
	\begin{proof}
		Plugging in the  block structure of the tensors  into \cref{eq:STSTST} we obtain:
		\begin{equation*}
			\begin{split}
			0 =	& \sum_{\{i\}} \tr\left(
			\oplus_{j=1}^m\oplus_{q=1}^{r_j}\mu^{N}_{j,q} 
			\left[
			T_{j}^{i_1} - S_{j}^{i_1}
			\right]
			A_j^{i_2}\ldots A_j^{i_N}  
			\right) 
			\ket{i_1,i_2,\ldots,i_N}
			\\ 
			= &   
			\sum_{j=1}^m  \sum_{q=1}^{r_j} \mu^{N}_{j,q} 
			\sum_{\{i\}}
			\tr\left(
			\left[
			T_{j }^{i_1} - S_{j }^{i_1}
			\right]
			A_j^{i_2}\ldots A_j^{i_N}  
			\right) 
			\ket{i_1,i_2,\ldots,i_N} \ .
			\end{split} 
		\end{equation*}
		Plugging in the definition of the map $\Gamma_A$ (\cref{def:GammaMap}) 
		\begin{equation*}
			\sum_{j=1}^m \sum_{i_1}  \Gamma_{A_j}^{N-1} 
			\left(
			\sum_{q=1}^{r_j} \mu^{N}_{j,q}
			\left[ 	T_{j}^{i_1} - S_{j}^{i_1} \right] \right) \otimes \ket{i_1} =0 \ .
		\end{equation*}
		According to  \cref{prop:OrthogBNT}, for $N$ large enough ($\geq L_0$)  we have for all $i_1$ and all $j$  
		\begin{equation*}
			\sum_{q=1}^{r_j} \mu^{N}_{j,q}
			\left[ 	T_{j}^{i_1} -S_{j}^{i_1} \right]   =0 \ .
		\end{equation*}
		For all $j$, since $\{\mu_{j,q}\}_{q=1}^{r_j}$ are non-zero, there exists an
		$N\geq L_0$ such that 
		$\sum_{q=1}^{r_j} \mu^{N}_{j,q}\neq0$. Therefore for all $j$ we have:
		\begin{equation} \label{eq:BAAAA_CFequ}
			T_{j}^{i} = S_{j}^{i} \ .
		\end{equation} 	
	\end{proof}

	We review the fundamental theorem of MPV \cite{Cirac2017} and apply it to the case of a MPV with a global symmetry.
	
	\begin{proposition}  \label{prop:similarBNTs} \textup{\cite{Cirac2017}}
		Let $A$ and $B$ be tensors in CF (\cref{eq:BNTexpansion}) with BNT $\{A_{j}\}_{j=1}^{g_a}$ and $\{B_{k}\}_{k=1}^{g_b}$ respectively. If for all $N$ the tensors $A$ and $B$ generate MPVs proportional to each other, then   $g_a = g_b$ and for every $j$ there is a unique $k(j)$, a unitary matrix  $X_j$ and a phase $e^{i\phi_j}$ such that: 
		\begin{equation*}
			A^i_j=e^{i\phi_j}X_j^{-1}B_{k(j)}^iX_j \ . 
		\end{equation*} 
	\end{proposition}

	\begin{remark} \label{rem:XupToConst}
		 Note that $X_j$ are   determined up to a phase.
	\end{remark}
	
	\Cref{prop:similarBNTs} was proved in \cite{Cirac2017} and was used to prove the following:
	
	\begin{theorem} [The Fundamental Theorem of MPV] \label{thm:fundMPV}
		 Let two tensors $A$ and $B$ in CF (CFII) generate the same MPV for all $N$. Then they have the same block structure, and there exists an invertible (unitary) matrix $X$:
		\begin{equation} \label{eq:XStruct3}
			X = \oplus_{j=1}^m \oplus_{q=1}^{r_j} X_j \ ,
		\end{equation}
		which is block diagonal, with the same block structure as $A$, and a permutation matrix $\Pi$ between the blocks, such that:
		\begin{equation*}
			\begin{tikzpicture}[baseline=-1mm]
			\pic[pic text = $A$] at (0,0) {mpstensor};
			\end{tikzpicture} \ = \
			\begin{tikzpicture}[baseline=-1mm]
			\pic[pic text = $\myinv{X}$] at (0,0) {horizontalmatrix};
			\pic[pic text = $\myinv{\Pi}$] at (1,0) {horizontalmatrix};
			\pic[pic text = $B$] at (2,0) {mpstensor};
			\pic[pic text = $\Pi$] at (3,0) {horizontalmatrix};
			\pic[pic text = $X$] at (4,0) {horizontalmatrix};	   
			\end{tikzpicture} \ .
		\end{equation*} 
	\end{theorem}
	 
	We now apply the fundamental theorem of MPV to the case when a MPV generated by a tensor $A$ in CFII is invariant under the action of the same unitary operator on every site:

	\begin{corollary} \label{thm:GlobalSymmUnitary}
		Let $A$  be a tensor in CFII (\cref{eq:BNTexpansion}) generating a MPV with a global invariance under a unitary $\CC$:
		\begin{equation*}
			\CC^{\otimes N}\ket{\psi^N_A} = \ket{\psi^N_A} \ ,
		\end{equation*}
		then $A$ transforms under the unitary matrix as:
		\begin{equation*}
			\begin{tikzpicture}[baseline=-1mm]
			\pic[pic text = $A$] at (0,0) {mpstensor};
			\pic[pic text = $\CC$] at (0,1) {verticalmatrix};
			\end{tikzpicture} \ = \
			\begin{tikzpicture}[baseline=-1mm]
			\pic[pic text = $\myinv{X}$] at (0,0) {horizontalmatrix};
			\pic[pic text = $\myinv{\Pi}$] at (1,0) {horizontalmatrix};
			\pic[pic text = $A$] at (2,0) {mpstensor};
			\pic[pic text = $\Pi$] at (3,0) {horizontalmatrix};
			\pic[pic text = $X$] at (4,0) {horizontalmatrix};	   
			\end{tikzpicture} \ ,
		\end{equation*} 
		where $X$ is a unitary matrix with the same block structure as $A$, and is unitary in each block (\cref{eq:XStruct3}), and $\Pi$ is a permutation between the $j$ blocks of $A$ (it does not permute the $q$ blocks). 
	\end{corollary}
	
	\begin{proof}
		The tensor $\tilde{A}$ consisting of the matrices  $\tilde{A}^i:=\sum_{i^\prime}\CC_{i,i^\prime} A^{i^\prime}$ generates $\CC^{\otimes N}\ket{\psi^N_A}$. Before finishing the proof, we shall now prove the following lemma: 
		\begin{lemma} \label{lem:BNTThetaA}
			Let $\{A_j\}$ be the BNT of $A$, then the tensors $\{\tilde{A}_j\}$ composed of  the matrices 
			$\tilde{A}^i_j = \sum_{i^\prime}\CC_{i,i^\prime} A_j^{i^\prime}$ form a BNT of $\tilde{A}$, and $\tilde{A}$ is in CFII.
		\end{lemma}
		 \begin{proof}[Proof: \cref{lem:BNTThetaA}]
		 	 $\tilde{A}_j$ are normal tensors and in CFII because a unitary mixture of the Kraus operators gives the same CP map (\cref{prop:KrusMix}), and they are a basis because $\{A_j\}$ is.
		 \end{proof}
		We can now apply the fundamental theorem of MPV to $A$ and $\tilde{A}$. In this case, however, because the coefficients $\mu_{j,q}$ in \cref{eq:BNTexpansion} are the same for $A$ and $\tilde{A}$, $\Pi$  permutes only between $j$ blocks.	
	\end{proof}
	
	Next we apply the above to  a MPV with a  global symmetry as in \cref{def:OneDofGlobalSymm}:
	\begin{equation*}
		\CC(g)^{\otimes N}\ket{\psi^N_A} = \ket{\psi^N_A} \ .
	\end{equation*}

\GlobalSymmForMPVinCF*
	\begin{proof}
		According to \cref{thm:GlobalSymmUnitary}, for every $g\in G$ we have:
		\begin{equation} \label{eq:preGroupProp}
			\sum_{i^\prime}\CC(g)_{i,i^\prime} A^{i^\prime} 
			= X(g)^{-1} \Pi(g)^{-1} A^i \Pi(g) X(g) \ .
		\end{equation}
		Consider the action of the group element $gh\in G$ in two ways using \cref{eq:preGroupProp}:
		\begin{equation*}
			\begin{split}
			X(gh)^{-1} \Pi(gh)^{-1} A^i \Pi(gh) X(gh) = & \sum_{i^\prime}\CC(gh)_{i,i^\prime} A^{i^\prime} \\
			=&\sum_{i^\prime,k}\CC(g)_{i,k} \CC(h)_{k,i^\prime} A^{i^\prime} \\
			=&\sum_{k}\CC(g)_{i,k}  X(h)^{-1} \Pi(h)^{-1} A^k \Pi(h) X(h) \\
			=& X(h)^{-1} \Pi(h)^{-1} X(g)^{-1} \Pi(g)^{-1} A^i \Pi(g) X(g) \Pi(h) X(h) \ .
			\end{split}
		\end{equation*}
		Taking the $L$-fold product  of the LHS and RHS for different indices $i_1,i_2,\ldots,i_L$ we obtain:
		\begin{equation} \label{eq:XPiAAAPiX}
			X(gh)^{-1} \Pi(gh)^{-1} \left(A^{i_1}A^{i_2}\ldots A^{i_L}\right) \Pi(gh) X(gh) 
			= X(h)^{-1} \Pi(h)^{-1} X(g)^{-1} \Pi(g)^{-1} \left(A^{i_1}A^{i_2}\ldots A^{i_L}\right) \Pi(g) X(g) \Pi(h) X(h) \ . 
		\end{equation}
		We shall now prove the following lemma, and then continue with the proof.
		\begin{lemma} \label{lem:PiIsRep}
			$\Pi(g)$ is a representation of $G$ and is therefore the trivial one.
		\end{lemma}
		\begin{proof}[Proof: \cref{lem:PiIsRep}]
			According to  \cref{cor:SpanForCF}, by taking  appropriate linear combinations of \cref{eq:XPiAAAPiX}  we can  obtain:
			\begin{equation} \label{eq:PermutationGroupPropForBlocks}
				X(gh)^{-1} \Pi(gh)^{-1} 
				\left(\Delta[{j}]\right) 
				\Pi(gh) X(gh) 
				= X(h)^{-1} \Pi(h)^{-1} X(g)^{-1} \Pi(g)^{-1}
				\left(\Delta[{j}]\right) 
				\Pi(g) X(g) \Pi(h) X(h) \ ,
			\end{equation}
			where
			$\Delta[{j_0 }]$ is a matrix consisting of multiples of $\II$ in the  $j_0$ block and zero in all the rest: 
			$\Delta[{j_0 }]:= \oplus_{j=1}^m\oplus_{q=1}^{r_j}\mu^L_{j,q}\delta_{j,j_0} \II_{D_j\times D_j}$. This is achieved by setting $M_j= \delta_{j,j_0}\II$ in \cref{eq:SLCF}.
			Denote by $g(j)$ the image of the block $j$ under the permutation $\Pi(g)$, then 
			$\Pi(g)^{-1} \Delta[{j}] \Pi(g) = \Delta[{g^{-1}(j)}] $. 
			Plugging this into \cref{eq:PermutationGroupPropForBlocks} we get:
			\begin{equation*}
				\begin{split}
				LHS=& X(gh)^{-1}\left(\Delta[{{(gh)}^{-1}(j)}]\right) X(gh)\\ 
				= &\Delta[{{(gh)}^{-1}(j)}] = \\
				RHS=& X(h)^{-1} \Pi(h)^{-1} X(g)^{-1}\left( \Delta[{g^{-1}(j)}]\right)X(g) \Pi(h) X(h)\\
				=&	 X(h)^{-1} \Pi(h)^{-1} \left(\Delta[{g^{-1}(j)}]\right) \Pi(h) X(h) \\
				=& \Delta[{h^{-1}(g^{-1}(j))}] \ ,
				\end{split}
			\end{equation*}
			where in each step the $X$s commute with the $\Delta$s because they have the same block structure and the $\Delta$s are proportional to $\II$ in each block. We conclude that $  {{(gh)}^{-1}(j)} $ and ${h^{-1}(g^{-1}(j))}$ are the same block number and therefore $\Pi(g)$ is a group homomorphism.  It remains to show that  $\Pi(g)$ depends on $g$ smoothly. From \cref{eq:preGroupProp} we obtain:
			\begin{equation} \label{eq:PiIsARep}
				X(g)^{-1} \Pi(g)^{-1} A^{i_1} A^{i_2} \ldots A^{i_L} \Pi(g)  X(g) =   
				\sum_{\{i^\prime\}}
				\left( \CC(g)_{i_1,i_1^\prime}A^{i_1^\prime} \right)
				\left(\CC(g)_{i_2,i_2^\prime} A^{i_2^\prime} \right)
				\ldots 
				\left( \CC(g)_{i_L,i_L^\prime} A^{i_L^\prime} \right) \ .
			\end{equation}
			As above, we can take a linear combination of the $A$s to get a $\Delta[j]$ between the permutations in the LHS. Knowing how the permutation acts on each  $\Delta[j]$ determines $\Pi(g)$ completely. The $X$s on the LHS commute with all $\Delta[j]$ as before. The RHS will then be a linear combination of $\{\CC(g)A\}$, and will thus depend on $g$ smoothly. 	Since we assumed $G$ is a connected Lie group we conclude that $\Pi(g)\equiv \II$.
		\end{proof}
		
		We now repeat the step leading to \cref{eq:PermutationGroupPropForBlocks} but this time with an arbitrary matrix $M$ in the $j$ block: $\Delta^M_{j_0}:= \oplus_{j=1}^m\oplus_{q=1}^{r_j}\delta_{j,j_0}\mu^L_{j,q}M$. 
		Equation \cref{eq:PermutationGroupPropForBlocks} becomes:
		\begin{equation*} 
			X(gh)^{-1} 
			\left(\Delta^M_{j}\right) X(gh)
			= X(h)^{-1}  X(g)^{-1}
			\left(\Delta^M_{j}\right)  X(g)X(h) \ .
		\end{equation*}
		This means that for any $j$ block we have:
		\begin{equation*} 
			\oplus_{q=1}^{r_j} \mu^L_{j,q} {X_j}(gh)^{-1} 
			MX_{j}(gh)
			=\oplus_{q=1}^{r_j} \mu^L_{j,q} X_{j}(h)^{-1}  X_{j}(g)^{-1}
			M X_{j}(g)X_{j}(h) \ . 
		\end{equation*}
		We see that $X_{j}(g)X_{j}(h)(X_{j})^{-1}(gh)$ commutes with every matrix $M$ and is therefore proportional to the identity.   $X_{j}(g)$ is therefore a projective  representation.	
	\end{proof}
	
	\begin{remark}
		Note that different blocks of $X(g)$ can belong to different equivalence classes of projective representations. We could construct such an example by taking the direct sum of two normal tensors $A$ and $\tilde{A}$, which transform under a given representation $\CC(g)$ with $X(g)$ and $\tilde{X}(g)$, projective representations from different cohomology classes. $X(g)\oplus \tilde{X}(g)$  is then not   a  projective representation because $X(gh)\oplus \tilde{X}(gh)$  differs from $X(g)X(h)\oplus \tilde{X}(g)\tilde{X}(h)$ by a diagonal matrix and not a scalar one.
	\end{remark}

\section{Derivation and Proofs of the Results} \label{sec:detProofs}
	In this section we prove the theorems stated in \cref{sec:RsultsOverview}.

\subsection{Matter MPV with local symmetry} \label{sub:Onedegree of freedom}

	\propSinglets*  
	
	\begin{proof}
		 Write $\ket{\psi^N}$ in the irreducible representation basis which satisfies:
		\begin{equation*}
		\CC(g) \ket{j,m} = \sum_n {D^j(g)}_{n,m} \ket{j,n}  \ , 
		\end{equation*}
		where ${D^j(g)}$ are irreducible representation matrices.
		\begin{equation*}
		\ket{\psi^N}=\sum c_{j_1,m_1,\ldots,j_N,m_N}\ket{j_1,m_1,\ldots,j_N,m_N} \  .
		\end{equation*}
		The local symmetry condition implies:
		\begin{equation*}
		\sum_{n_1} {D^{j_1}(g)}_{m_1,n_1} c_{j_1,n_1,\ldots,j_N,m_N} = c_{j_1,m_1,\ldots,j_N,m_N} \ ,
		\end{equation*}
		which means that the vector of coefficients $\overrightarrow{c}_{j_1,(\cdot),\ldots,j_N,m_N}$ is either zero or an invariant subspace of ${D^{j_1}(g)}$, in which case ${D^{j_1}(g)}$ is the trivial representation. This implies that the coefficients $c_{j_1,m_1,\ldots,j_N,m_N}$ are zero whenever any one of the $j_k$s corresponds to a non trivial representation.
	\end{proof}

	\TrivTensorTrans*

	\begin{proof}
		We apply \cref{lem:BAAAA_CFequ} with $S^i:= \sum_{i^\prime} \CC(g)_{ii^\prime}A^{i^\prime}$ and $T^i:= A^i$.		
	\end{proof}

	\begin{remark} \label{rem:NotOnlyReps}
		We have never used any properties of $\CC(g)$ as a representation. The same proof is valid for any operator $\CC$.
	\end{remark}

	According to \cref{rem:changeOfBasis}, the MPV generated by $A$ can be written in terms of a tensor $\tilde{A}$, composed of the matrices $\{\tilde{A}^{j,m}\}$, corresponding to the irreducible representation basis $\{\ket{j,m}\}$ on which $\CC(g)$ acts as 
	$\CC(g) \ket{j,m}=\sum_n {D^j}(g)_{n,m}\ket{j,n}$. According to \cref{prop:ChangeOfBasisDoesntRuinCF}, $\tilde{A}$ is also in CF. Applying \cref{thm:TrivTensorTrans} to $\tilde{A}$ leads to the following:
	
	\decomposeTheta*
	\begin{proof}
		From  \cref{eq:thetaA_eq_A} we deduce that each vector of matrix elements of $A$:
		$\vec{A}_{\alpha,\beta}^{j} = \left( A_{\alpha,\beta}^{j,1},A_{\alpha,\beta}^{j,2},\ldots, A_{\alpha,\beta}^{j,dim(j)}\right)^T$  
		is invariant under  ${D^j(g)}$ for all $g\in G$. This implies that   either 	$\vec{A}_{\alpha,\beta}^{j}$ is zero  or that ${D^j(g)}$ is the one dimensional trivial representation.
	\end{proof}

\subsection{Pure gauge field MPV}
	In order to prove \cref{thm:OnlyFieldGaugeGroup} we shall proceed as in \cref{sec:GlobSymm}: we shall first  prove a lemma which describes the case when $\RR$ and $\LL$ are just unitary operators, and later use that to prove the case when they are representations.

	\begin{lemma} \label{thm:InvFieldTensor}
		Let $B$ be a tensor in CFII:
		\begin{equation*} 
			B^i=\oplus_{j=1}^m\oplus_{q=1}^{r_j}\mu_{j,q}B^i_j \ ,
		\end{equation*}
		and let $\RR$ and $\LL$ be two unitary operators such that for all $K$
		\begin{equation*}
			\RR^{[K]}\LL^{[K+1]}\ket{\psi^N_{B}}= \ket{\psi^N_{B}} \ .
		\end{equation*}
		Then  $B$ transforms under the unitary matrices as follows:
		\begin{equation}   \label{eq:BLRTransformation}
			\begin{tikzpicture}[baseline=-1mm]
			\pic[pic text = $B$] at (0,0) {mpstensor};
			\pic[pic text = $\RR$] at (0,1) {verticalmatrix};
			\end{tikzpicture} \ = \ 
			\begin{tikzpicture}[baseline=-1mm]
			\pic[pic text = $B$] at (0,0) {mpstensor};
			\pic[pic text = $X$] at (1,0) {horizontalmatrix};
			\end{tikzpicture} \;\;\;\;\; ; \;\;\;\;\;
			\begin{tikzpicture}[baseline=-1mm]
			\pic[pic text = $B $] at (0,0) {mpstensor};
			\pic[pic text = $\LL$] at (0,1) {verticalmatrix};
			\end{tikzpicture} \ = \ 
			\begin{tikzpicture}[baseline=-1mm]
			\pic[pic text = $\myinv{{X}}$] at (0,0) {horizontalmatrix};
			\pic[pic text = $B$] at (1,0) {mpstensor}; 
			\end{tikzpicture} \ , 
		\end{equation} 
		where $X$ is a unitary matrix with the same block structure as $B^i$, as in \cref{eq:XStruct3}.
	\end{lemma}

	\begin{proof}
		Applying  \cref{thm:TrivTensorTrans} (recall \cref{rem:NotOnlyReps}) to the tensor $BB$ and the unitary $\RR\otimes\LL$ ($BB$ is in CF if $B$ is in CF), we obtain:
		\begin{equation}  \label{eq:RLB}
			\begin{tikzpicture}[baseline=-1mm]
			\pic[pic text = $B $] at (0,0) {mpstensor};
			\pic[pic text = $B $] at (1,0) {mpstensor};
			\pic[pic text = $\RR$] at (0,1) {verticalmatrix};
			\pic[pic text = $\LL$] at (1,1) {verticalmatrix};
			\end{tikzpicture} \ = \ 
			\begin{tikzpicture}[baseline=-1mm]
			\pic[pic text = $B$] at (0,0) {mpstensor};
			\pic[pic text = $B$] at (1,0) {mpstensor};
			\end{tikzpicture} \ .
		\end{equation}
		Applying the pair of operators to every site on the chain (for even $N$) we conclude that the MPV is invariant under the global application of the operators in reversed order: $(\LL\otimes\RR )^{\otimes N}\ket{\psi^{2N}_{B}} = \ket{\psi^2N_{B}}$. 
		Using \cref{thm:GlobalSymmUnitary} we obtain:
		\begin{equation} \label{eq:LRB}
			\begin{tikzpicture}[baseline=-1mm]
			\pic[pic text = $B$] at (0,0) {mpstensor};
			\pic[pic text = $B$] at (1,0) {mpstensor};
			\pic[pic text = $\LL$] at (0,1) {verticalmatrix};
			\pic[pic text = $\RR$] at (1,1) {verticalmatrix};
			\end{tikzpicture} \ = \
			\begin{tikzpicture}[baseline=-1mm]
			\pic[pic text = $\myinv{X}$] at (0,0) {horizontalmatrix};
			\pic[pic text = $\myinv{\Pi}$] at (1,0) {horizontalmatrix};
			\pic[pic text = $B$] at (2,0) {mpstensor};
			\pic[pic text = $B$] at (3,0) {mpstensor};
			\pic[pic text = $\Pi$] at (4,0) {horizontalmatrix};
			\pic[pic text = $X$] at (5,0) {horizontalmatrix};	   
			\end{tikzpicture} \ ,
		\end{equation} 
		where  $X$ is  unitary  and  $\Pi$  is a permutation, as in  \cref{thm:GlobalSymmUnitary}. Next consider the following tensor:
		\begin{equation*}
			\begin{tikzpicture}[baseline=-1mm]
			\foreach \i in {0,1,2,3,5} \pic[pic text = $B$] at (\i,0) {mpstensor};
			\foreach \i in {0,2} \pic[pic text = $\LL$]  at (\i,1) {verticalmatrix};
			\foreach \i in {1,3,5} \pic[pic text = $\RR$]  at (\i,1) {verticalmatrix};	
			\node at (4,0) {$\ldots$};
			\end{tikzpicture} \ .
		\end{equation*} 
		According to \cref{eq:RLB} this tensor is equal to the LHS of the following,  and according to \cref{eq:LRB} - to the RHS:
		\begin{equation} \label{eq:XPiBBBPiX}
			\begin{split}
			LHS= &
			\begin{tikzpicture}[baseline=-1mm]
			\foreach \i in {2,3,4,5,7} \pic[pic text = $B$] at (\i,0) {mpstensor};
			\foreach \i in {2} \pic[pic text = $\LL$]  at (\i,1) {verticalmatrix};
			\foreach \i in {7} \pic[pic text = $\RR$]  at (\i,1) {verticalmatrix};	
			\node at (6,0) {$\ldots$};
			\node at (-0.35,0) {$ $};
			\end{tikzpicture} = 
			\\
			\\
			RHS =&
			\begin{tikzpicture}[baseline=-1mm]
			\foreach \i in {2,3,4,5,7} \pic[pic text = $B$] at (\i,0) {mpstensor};
			\pic[pic text = $\myinv{X}$] at (0,0) {horizontalmatrix};
			\pic[pic text = $\myinv{\Pi}$] at (1,0) {horizontalmatrix};
			\pic[pic text = $\Pi$] at (8,0) {horizontalmatrix};
			\pic[pic text = $X$] at (9,0) {horizontalmatrix};
			\node at (6,0) {$\ldots$};
			\end{tikzpicture} \ .
			\end{split}	
		\end{equation} 
		Using the same argument as in equation \cref{eq:PermutationGroupPropForBlocks}, we show that the permutation must act trivially: use \cref{cor:SpanForCF} on the string of consecutive $B$s, excluding the extreme right and left ones, to obtain multiples of  $\II$ in a single $j$ block and zeros elsewhere. Note that $\RR$ and $\LL$ do not change the block structure of the tensors they act on. Now compare  the RHS with the LHS block-wise, if $\Pi$ acts non trivially on a block $j$, then we get that  $B_jB_j$ is zero, which is a contradiction to $B_j$ being normal.
		Next, having eliminated the possibility of a permutation, project \cref{eq:XPiBBBPiX} to any $(j,q)$ block to obtain:
		\begin{equation*}  
			\begin{split}
			&
			\begin{tikzpicture}[baseline=-1mm]
			\foreach \i in {1,2,3,4,6} \pic[pic text = $B_{j}$] at (\i,0) {mpstensor};
			\foreach \i in {1} \pic[pic text = $\LL$]  at (\i,1) {verticalmatrix};
			\foreach \i in {6} \pic[pic text = $\RR$]  at (\i,1) {verticalmatrix};	
			\node at (5,0) {$\ldots$};
				\node at (-.35,0) {$ $};
			\end{tikzpicture} = 
			\\
			\\
			=&
			\begin{tikzpicture}[baseline=-1mm]
			\foreach \i in {1,2,3,4,6} \pic[pic text = $B_{j}$] at (\i,0) {mpstensor};
			\pic[pic text = $\myinv{X}_j$] at (0,0) {horizontalmatrix};
			\pic[pic text = $X_j$] at (7,0) {horizontalmatrix};
			\node at (5,0) {$\ldots$};
			\end{tikzpicture} \ ,
			\end{split}	
		\end{equation*} 
		where $B_{j}$ is a normal tensor by assumption. We can now apply the inverse on the string of $B$s in the middle  ($BB$ is normal if $B$ is normal) to  obtain:
		\begin{equation*}  
			\begin{tikzpicture}[baseline=-1mm]
			\pic[pic text = $B_j$] at (0,0) {mpstensor};
			\pic[pic text = $\LL$] at (0,1) {verticalmatrix};
			\end{tikzpicture} 
			\otimes 
			\begin{tikzpicture}[baseline=-1mm]
			\pic[pic text = $B_j$] at (0,0) {mpstensor};
			\pic[pic text = $\RR$] at (0,1) {verticalmatrix};
			\end{tikzpicture} 
			\ = \ 
			\begin{tikzpicture}[baseline=-1mm]
			\pic[pic text = $\myinv{{X}}_j$] at (1,0) {horizontalmatrix};
			\pic[pic text = $B_j$] at (2,0) {mpstensor}; 
			%\node at (0,0) {$\left(x_j^{-1}\right)$};
			\end{tikzpicture} 		
			\otimes
			\begin{tikzpicture}[baseline=-1mm]
			\pic[pic text = $B_j$] at (1,0) {mpstensor};
			\pic[pic text = $X_j$] at (2,0) {horizontalmatrix};
			%\node at (0,0) {$\left(x_j\right)$};
			\end{tikzpicture} 
			\ . 
		\end{equation*} 
		According to \cref{rem:XupToConst}, the matrices $X_j$ are determined up to a constant. We now  choose a representative from the projective unitary class of $X_j$. 
		The above implies that for any such choice there is a constant $x_j$ such that:
		\begin{equation*}   
			\begin{tikzpicture}[baseline=-1mm]
			\pic[pic text = $B_j$] at (0,0) {mpstensor};
			\pic[pic text = $\RR$] at (0,1) {verticalmatrix};
			\end{tikzpicture} \ = \ 
			\begin{tikzpicture}[baseline=-1mm]
			\pic[pic text = $B_j$] at (1,0) {mpstensor};
			\pic[pic text = $X_j$] at (2,0) {horizontalmatrix};
			\node at (0,0) {$\left(x_j\right)$};
			\end{tikzpicture} 
			\;\;\;\;\; ; \;\;\;\;\;
			\begin{tikzpicture}[baseline=-1mm]
			\pic[pic text = $B_j$] at (0,0) {mpstensor};
			\pic[pic text = $\LL$] at (0,1) {verticalmatrix};
			\end{tikzpicture} \ = \ 
			\begin{tikzpicture}[baseline=-1mm]
			\pic[pic text = $\myinv{{X}}_j$] at (1,0) {horizontalmatrix};
			\pic[pic text = $B_j$] at (2,0) {mpstensor}; 
			\node at (0,0) {$\left(x_j^{-1}\right)$};
			\end{tikzpicture} \ . 
		\end{equation*} 
		Therefore the desired $X$ is  
		$X = \oplus_{j=1}^m \oplus_{q=1}^{r_j} x_jX_j$.
	\end{proof}

	\OnlyFieldGaugeGroup*

	\begin{proof}
		As we have seen in the proof of \cref{thm:InvFieldTensor},  \cref{eq:BLRTransformation} holds for each block of $B$, so for every group element $g\in G$ we have:
		\begin{equation} \label{eq:BgroupTransPerBlock}
			\begin{tikzpicture}[baseline=-1mm]
			\pic[pic text = $B_{j}$] at (0,0) {mpstensorBIG};
			\pic[pic text = $\RR(g)$] at (0,1) {verticalmatrixBIG};
			\end{tikzpicture} \ = \ 
			\begin{tikzpicture}[baseline=-1mm]
			\pic[pic text = $B_{j}$] at (0,0) {mpstensorBIG};
			\pic[pic text = ${X_j}(g)$] at (1,0) {horizontalmatrixBIG};
			\end{tikzpicture} \;\;\;\;\; ; \;\;\;\;\;
			\begin{tikzpicture}[baseline=-1mm]
			\pic[pic text = $B_{j} $] at (0,0) {mpstensorBIG};
			\pic[pic text = $\LL(g)$] at (0,1) {verticalmatrixBIG};
			\end{tikzpicture} \ = \ 
			\begin{tikzpicture}[baseline=-1mm]
			\pic[pic text = $\myinv{{X_j}(g)}$] at (0,0) {horizontalmatrixWIDE};
			\pic[pic text = $B_{j}$] at (1.15,0) {mpstensorBIG};
			\end{tikzpicture} \ . 
		\end{equation} 
		
		We write the action of the group element $\RR(gh)$ on $B$ in two ways:
		\begin{equation*} 
			\begin{split}
			\begin{tikzpicture} [baseline=-1mm]
			\pic[pic text = $B_j$] at (0,0) {mpstensorBIG};
			\pic[pic text = $X_j{(gh)}$] at (1.15,0) {horizontalmatrixWIDE};
				\node at (-1.2,0) {$\gamma(g,h)\times$};
			\end{tikzpicture} \ = \
			\begin{tikzpicture}[baseline=-1mm]
			\pic[pic text = $B_j$] at (0,0) {mpstensorBIG};
			\pic[pic text = $\RR{(gh)}$] at (0,1) {verticalmatrixWIDE};
				\node at (-1.2,0) {$\gamma(g,h)\times$};
			\end{tikzpicture} \ = \ 
			\begin{tikzpicture}[baseline=-1mm]
			\pic[pic text = $B_j$] at (0,0) {mpstensorBIG};
			\pic[pic text = $\RR{(g)}$] at (0,2) {verticalmatrixBIG};
			\pic[pic text = $\RR{(h)}$] at (0,1) {verticalmatrixBIG};
			\end{tikzpicture} \ = \ \\ \\
			=
			\begin{tikzpicture}[baseline=-1mm]
			\pic[pic text = $B_j$] at (0,0) {mpstensorBIG};
			\pic[pic text = $\RR{(g)}$] at (0,1) {verticalmatrixBIG};
			\pic[pic text = $X_j{(h)}$] at (1,0) {horizontalmatrixBIG};
			\end{tikzpicture} \ = \
			\begin{tikzpicture}[baseline=-1mm]
			\pic[pic text = $B_j$] at (0,0) {mpstensorBIG};
			\pic[pic text = $X_j{(g)}$] at (1,0) {horizontalmatrixBIG};
			\pic[pic text = $X_j{(h)}$] at (2,0) {horizontalmatrixBIG}; 
			\end{tikzpicture}  \ .
			\end{split}
		\end{equation*} 
		Now by contracting with the tensor $B_jB_j\ldots B_j$ from the left, and taking the appropriate linear combination which results in the identity matrix ($B_j$ is normal), we obtain $\gamma(g,h)X_j(gh)= X_j(g)X_j(h)$. 
		This means that for all $j$ $X_j(g)$ is a projective representation with the same multiplier as $\RR(g)$ ($\gamma$). Therefore $X(g)$ is a projective representation.
	\end{proof}

	\HBstructure*

	\begin{proof}
		Even though $\ket{\psi_{B}}$ is defined in terms of the basis $\{\ket{j}\}$ in $\mathcal{H}_B$,  it is sufficient to consider only vectors of the form:
		\begin{equation*}
			\ket{\phi_{\alpha,\beta}} = \sum_i \bra{\alpha} B^i \ket{\beta} \ket{i} \ \in \mathcal{H}_B \ .
		\end{equation*}
		Let $\mathcal{H}:= span\{\ket{\phi_{\alpha,\beta}} \}_{\alpha,\beta}$.
		The group transformations $\LL(g)$ and $\RR(g)$ preserve $\mathcal{H}$:
		\begin{equation*}
			\begin{split}
			\RR(g) \ket{\phi_{\alpha,\beta}} = 
			&\sum_i \bra{\alpha}B^i X(g) \ket{\beta}\ket{i} = 
			\sum_{i,\gamma} \bra{\alpha}B^i \ket{\gamma}\bra{\gamma} X(g) \ket{\beta}\ket{i} = 
			\sum_{\gamma} \bra{\gamma} X(g) \ket{\beta} \ket{\phi_{\alpha,\gamma}} \\
			\LL(g) \ket{\phi_{\alpha,\beta}} = 
			&\sum_i \bra{\alpha}Y(g)^{-1}B^i  \ket{\beta}\ket{i} = 
			\sum_{i,\gamma} \bra{\alpha}Y(g)^{-1} \ket{\gamma}\bra{\gamma}B^i  \ket{\beta}\ket{i} = 
			\sum_{\gamma} \bra{\alpha}Y(g)^{-1} \ket{\gamma}\ket{\phi_{\gamma,\beta}} \ ,
			\end{split}
		\end{equation*}
		where \cref{eq:BRightGroupTransPre}  was used.
		Performing a Schmidt decomposition of $\ket{\psi_{AB}}$ (or $\ket{\psi_{B}}$, the argument is the same) with respect to any partition where one gauge field Hilbert space  is split off from the rest of the system:
		\begin{equation*}
			\begin{split}
			\ket{\psi_{AB}} & = 
			\sum_{\{i\},\{j\},\alpha,\beta} 
			\left(
			\bra{\alpha} B^{j_1} \ket{\beta}\ \bra{\beta} A^{i_2}B^{j_2}\ldots A^{i_N}B^{j_N}A^{i_1} \ket{\alpha} 
			\right)
			\;
			\ket{i_1} \otimes  	 \ket{j_1} \otimes \ket{i_2\ldots i_N j_N} \\
			& = \sum_{\alpha,\beta} \ket{\phi_{\alpha,\beta}}_{[2]} \ket{\psi_{\beta,\alpha}}_{[3,\ldots,2N,1]} \ ,
			\end{split}
		\end{equation*}
		we see that only vectors from ${\mathcal{H}}$ appear.
		Therefore it is sufficient to restrict ourselves to $\mathcal{H}_B = \mathcal{H}$.
		Next we show that $\mathcal{H}$ has a representation space structure. 	 
		 \Cref{eq:BRightGroupTransPre} implies that $\RR(g)$ and $\LL(h)$ commute on $\mathcal{H}$:
		\begin{equation*}
			\LL(g)\RR(h) \ket{\phi_{\alpha,\beta}} = \sum_i \bra{\alpha} Y(g)^{-1} B^i X(h) \ket{\beta} \ket{i} = \RR(h)\LL(g) \ket{\phi_{\alpha,\beta}} \ . 
		\end{equation*}
		Thus $\mathcal{H}$ forms a projective representation space of $G\times G$ with the projective representation map 
		$(g,h) \mapsto  \LL(g)\RR(h)$
		 with multiplier
		  $\gamma^{-1} \times \gamma$ of $G\times G$ defined by 
		  $\gamma^{-1} \times \gamma: ((g,h),(g^\prime,h^\prime)) \mapsto \gamma^{-1} (g,g^\prime)\gamma(h,h^\prime)$  \cite{Chuangxun}:
		\begin{equation*}
			\LL(g)\RR(h)  \LL({g^{\prime}})\RR({h^{\prime}})|_{\mathcal{H}}  = 
			\LL(g)\LL({g^{\prime}})\RR(h)\RR({h^{\prime}})|_{\mathcal{H}} =
			\gamma^{-1}(g,g^\prime)\gamma(h,h^\prime) \LL({gg^{\prime}})\RR({hh^{\prime}})|_{\mathcal{H}} 
			\ ,
		\end{equation*}
		where we used the fact that $	\LL(g)$ and  $\RR(h)$ commute  and   preserve $\mathcal{H}$; . 
		For finite or compact groups $\mathcal{H}$ decomposes into a direct sum of irreducible projective representations of $G\times G$ with multiplier $\gamma^{-1}\times \gamma$, each one of which is equivalent to a projective representation of the form $(g,h)\mapsto D^l_{\myinv{\gamma}}(g) \otimes D_\gamma^r(h)$  \cite{Chuangxun},
		which proves the proposition.
	\end{proof}

	Recall the definition of an elementary $B$ block:
	\ElementaryBblock*
	
	\ElemBBlock*

	\begin{proof}
		Write $B$ as a map $B:\mathbb{C}^{D_2} \rightarrow \mathbb{C}^{D_1}\otimes\mathcal{H}_B $:
		\begin{equation*}
			B= \sum_{m,n} B^{m,n}\otimes \ket{m}\ket{n}=\sum_{m,n,\alpha,\beta} B_{\alpha,\beta}^{m,n} \ket{\alpha}\bra{\beta}\otimes \ket{m}\ket{n}
		\end{equation*}	
		By hypothesis $B$ satisfies (\cref{eq:BRightGroupTransPre}):
		\begin{equation*}
			\left[ \II \otimes \left(\RR(g)\LL(h)\right) \right]B =   
			\left[
			\II \otimes \left({D_{\myinv{\gamma}}^l(h)}\otimes {D_\gamma^r(g)}\right)
			\right]
			B  =
			\left[ {Y(h)}^{-1}  \otimes \II \right] B  \left[X(g)  \otimes \II\right] \ .
		\end{equation*} 
		Write the above equality explicitly (repeated indices are summed over):
		\begin{equation*}
			\begin{split}
			LHS = &\sum 
			B^{m,n}_{\alpha,\beta} \ket{\alpha}\bra{\beta} 
			\otimes {D_{\myinv{\gamma}}^l(h)}\ket{m} {D_\gamma^r(g)} \ket{n} =\\
			&\sum 
			B^{m,n}_{\alpha,\beta} \ket{\alpha}\bra{\beta} 
			\otimes {D_{\myinv{\gamma}}^l(h)}_{m^\prime,m}\ket{m^\prime} {D_\gamma^r(g)}_{n^\prime,n} \ket{n^\prime} =\\
			RHS= &\sum
			B_{\alpha,\beta}^{m,n}{Y(h)}^{-1} \ket{\alpha}\bra{\beta} X(g) 
			\otimes \ket{m}  \ket{n} = \\
			&\sum
			B_{\alpha,\beta}^{m,n}\overline{{Y(h)}_{\alpha,\alpha^\prime}} \ket{\alpha^\prime}\bra{\beta^\prime} {X(g)}_{\beta,\beta^\prime} 
			\otimes \ket{m}  \ket{n} \ .
			\end{split}
		\end{equation*} 
		Projecting both LHS and RHS to $\ket{\hat{\alpha }}\bra{\hat{\beta} }\otimes\ket{\hat{m} }\ket{\hat{n}}$ we obtain
		\begin{equation*}
			\sum_{m,n} {D_{\myinv{\gamma}}^l(h)}_{\hat{m} ,m} {D_\gamma^r(g)}_{\hat{n},n}
			B^{m,n}_{\hat{\alpha},\hat{\beta}} 
			=
			\sum_{\alpha,\beta}
			B_{\alpha,\beta}^{\hat{m},\hat{n}}\overline{{Y(h)}_{\alpha,\hat{\alpha} }}  {X(g)}_{\beta,\hat{\beta}} \ .
		\end{equation*} 	
		The LHS is a multiplication from the left (summing the indices $m,n$) of the matrix $\mathbf{B}$, with entries $\mathbf{B}_{(m,n),(\alpha,\beta)}:=B^{m,n}_{\alpha,\beta}$,
		with the matrix ${D_{\myinv{\gamma}}^l(h)} \otimes {D_\gamma^r(g)}$, which is an irreducible projective representation of $G\times G$. The RHS is a multiplication of $\mathbf{B}$ from the right (summing the indices $\alpha,\beta$) with the  matrix $\overline{Y(h)}\otimes X(g)$, which is also an irreducible  projective representation of $G\times G$ (with the same multiplier). By Schur's lemma (\cref{thm:Schur}) $\mathbf{B}\propto \II$ 
		(i.e.\ $B^{m,n}_{\alpha,\beta} \propto \delta_{\alpha,m} \delta_{\beta,n}$)
		if  ${D_{\myinv{\gamma}}^l(h)}  \otimes {D_\gamma^r(g)} = \overline{Y(h)}\otimes X(g)$, and zero otherwise.
	\end{proof}

	\BstructGeneral*

	\begin{proof}
		Recall the structure of the tensor $B$ and the projective representation $X(g)$:		
		\begin{equation*}  
			\begin{split}
			B^{k;m,n} =& \oplus_{j=1}^m \oplus_{q=1}^{r_j}\mu_{j,q} B_j^{k;m,n} \\
			X(g) =& \oplus_{j=1}^m \oplus_{q=1}^{r_j} {X_j}(g) \ ,
			\end{split}
		\end{equation*}
		where $\{B_j\}$ are normal tensors.
		Project \cref{eq:BLRTGroupransformation} to a block $j,q$ of the virtual space to obtain:
		\begin{equation*}    
			\begin{tikzpicture}[baseline=-1mm]
			\pic[pic text = $B_j$] at (0,0) {mpstensorBIG};
			\pic[pic text = $\RR(g)$] at (0,1) {verticalmatrixBIG};
			\end{tikzpicture} \  = \ 
			\begin{tikzpicture}[baseline=-1mm]
			\pic[pic text = $B_j$] at (0,0) {mpstensorBIG};
			\pic[pic text =  $X_j(g)$] at (1,0) {horizontalmatrixBIG};
			\end{tikzpicture} \;\;\;\;\;\;\;\;\; ;  \;\;\;\;\;\;\;\;\; 
			\begin{tikzpicture}[baseline=-1mm]
			\pic[pic text = $B_j $] at (0,0) {mpstensorBIG};
			\pic[pic text = $\LL(g)$] at (0,1) {verticalmatrixBIG};
			\end{tikzpicture} \   = \ 
			\begin{tikzpicture}[baseline=-1mm]
			\pic[pic text =  $\myinv{X_j(g)}$] at (0,0) {horizontalmatrixWIDE};
			\pic[pic text = $B_j$] at (1.15,0) {mpstensorBIG}; 
			\end{tikzpicture} \ .
		\end{equation*} 
		Let ${X_j}(g) =\oplus_a X_j^a(g)$ be a block  of $X(g)$.  We shall prove each item in the statement:
		\begin{enumerate}
			\item
				Let $B^k_j$ be the projection of the tensor $B_j$ to the $k$ sector of the physical Hilbert space. If for a certain $k$ there exist no   $a$ and $b$ such that $X_j^b(g) = {D_\gamma^{r_k}(g)} $ and  
				$\overline{ X_j^a(g)} = {{D_{\myinv{\gamma}}^{l_k}(g)}}$, then according to \cref{prop:ElemBBlock}, for all $a,b$ the $a,b$ block of $B^k_j$, consisting of the matrices $B_{j,a,b}^{k,m,n}$, is zero. This means $B_j^k$ is zero.
			\item \label{item:Y}
				If there is a $Y^a(g)$ for which there is no  appropriate $k$ then according to \cref{prop:ElemBBlock}, $B_j^{k,m,n}$ all have a zero row which is a contradiction to the normality of $B_j$.  
			\item 
				As in \cref{item:Y}, $B_j^{k,m,n}$ now would have a zero column, which contradicts the normality of $B_j$.
		\end{enumerate} 
	\end{proof}
 
	 The proof of \cref{prop:invGauge1} will be presented in the next section after we derive the structure of the symmetric matter tensor $A$.

	\InverseGaugeNES*
	\begin{proof}
		We use the local symmetry condition around every $A$:
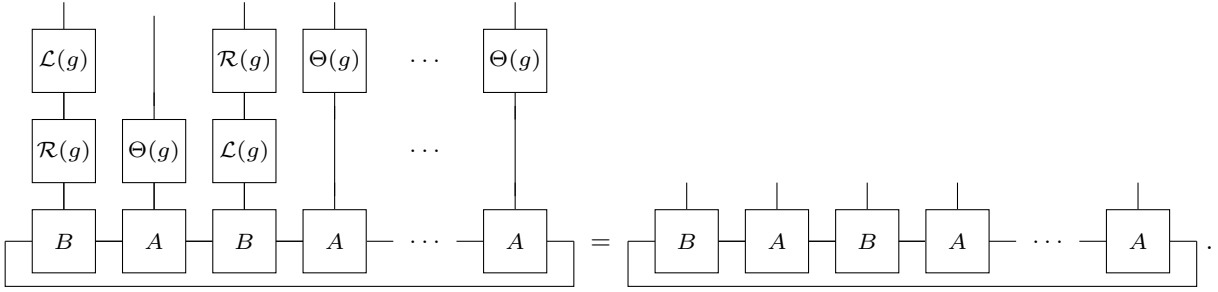
		\begin{equation*}
			\begin{tikzpicture}[baseline=-1mm]
			\foreach \i in {0,2} \pic[pic text = $B$] at (\i,0) {mpstensorBIG};
			\foreach \i in {1,3,5} \pic[pic text = $A$] at (\i,0) {mpstensorBIG};
			\foreach \i in {0,1,2} \pic at (\i,1) {verticalmatrixBIG};
			\foreach \i in {2,3,0,5} \pic at (\i,2) {verticalmatrixBIG};
			\node at (0,1) {$\RR(g)$};
			\node at (1,1) {$\CC(g)$};
			\node at (2,1) {$\LL(g)$};
			\node at (2,2) {$\RR(g)$};
			\node at (3,2) {$\CC(g)$};
			\node at (0,2) {$\LL(g)$};
			\node at (5,2) {$\CC(g)$};
			\node at (4,0) {$\ldots$};\node at (4,1) {$\ldots$};
			\node at (4,2) {$\ldots$};
			\draw (1,1.5)--(1,2.5);
			\draw (5,1.5)--(5,.5);
			\draw (3,1.5)--(3,.5);
			\draw (-.65,0)--(-.65,-.5)--(5.65,-.5)--(5.65,0);
			\end{tikzpicture} \ = \ 
			\begin{tikzpicture}[baseline=-1mm]
			\foreach \i in {0,2 } \pic[pic text = $B$] at (\i,0) {mpstensorBIG};
			\foreach \i in {1,3,5} \pic[pic text = $A$] at (\i,0) {mpstensorBIG};
			\node at (4,0) {$\ldots$};	
			\draw (-.65,0)--(-.65,-.5)--(5.65,-.5)--(5.65,0);
			\end{tikzpicture} \ .
		\end{equation*} 
		According to the transformation laws for $B$, the LHS of the above equals:
		
		\begin{equation*}
			= \ 
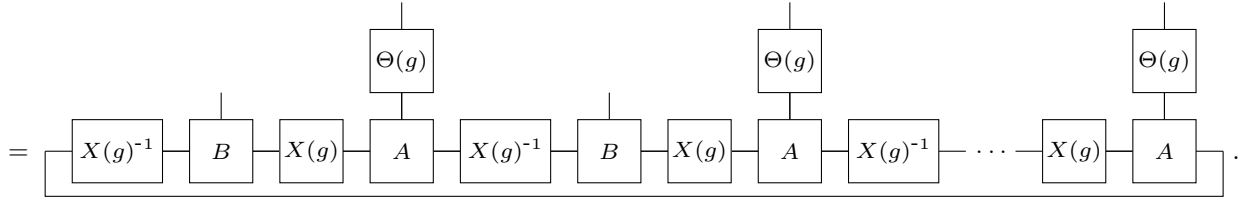
\begin{tikzpicture}[baseline=-1mm]
			\foreach \i in {1.15,5.45 } \pic[pic text = $B$] at (\i,0) {mpstensorBIG};
			\foreach \i in {3.15,7.45,11.6} \pic[pic text = $A$] at (\i,0) {mpstensorBIG};
			\foreach \i in {3.15,7.45,11.6} \pic[pic text = $\CC(g)$] at (\i,1) {verticalmatrixBIG};
			\foreach \i in {2.15,6.45,10.6} \pic[pic text = $X(g)$]  at (\i,0) {horizontalmatrixBIG};
			\foreach \i in {0,4.3,8.6} \pic[pic text = $\myinv{X(g)}$]  at (\i,0) {horizontalmatrixWIDE};
			\node at (9.68,0) {$\ldots$};	
			\draw (-.8,0)--(-.8,-.5)--(12.25,-.5)--(12.25,0);
			\end{tikzpicture} \ .
		\end{equation*} 
		We can now use the assumption $\II\in span\{B^{k;m,n}\}$ to eliminate the $B$s from the equation, the $X$s  then cancel out and we obtain the desired global symmetry:
		
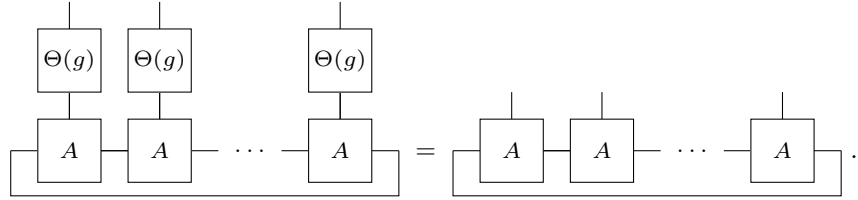
\begin{equation*}
			\begin{tikzpicture}[baseline=-1mm]
			\foreach \i in {0,1,3} \pic[pic text = $A$] at (\i,0) {mpstensorBIG};
			\foreach \i in {0,1,3} \pic[pic text = $\CC(g)$] at (\i,1) {verticalmatrixBIG};
			\node at (2,0) {$\ldots$};	
			\draw (-.65,0)--(-.65,-.5)--(3.65,-.5)--(3.65,0);
			\end{tikzpicture} \ = \ 
			\begin{tikzpicture}[baseline=-1mm]
			\foreach \i in {0,1,3} \pic[pic text = $A$] at (\i,0) {mpstensorBIG};
			\node at (2,0) {$\ldots$};	
			\draw (-.65,0)--(-.65,-.5)--(3.65,-.5)--(3.65,0);
			\end{tikzpicture} \ .
		\end{equation*}	
		If in addition $A$ is in CF, we can apply \cref{thm:GlobalSymmVirtTrans} to obtain transformation relations for $A$. 
		To show the rest of the claim (if $A$ in addition has the block structure of $B$) we write the symmetry condition and again use the transformation rules for $B$:
		
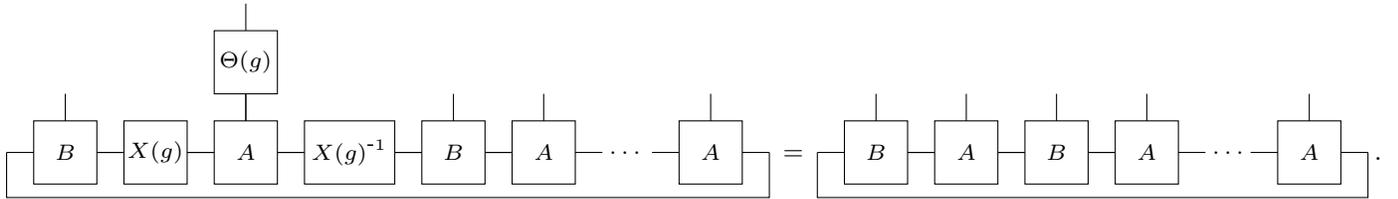
		\begin{equation*}
			\begin{tikzpicture}[baseline=-1mm]
			\foreach \i in {0,4.3} \pic[pic text = $B$] at (\i,0) {mpstensorBIG};
			\foreach \i in {2,5.3,7.15} \pic[pic text = $A$] at (\i,0) {mpstensorBIG};
			\foreach \i in {1 } \pic[pic text = $X(g)$] at (\i,0) {horizontalmatrixBIG};		
			\foreach \i in {3.15 } \pic[pic text = $\myinv{X(g)}$] at (\i,0) {horizontalmatrixWIDE};		
			\pic[pic text = {$\CC(g)$}] at (2,1) {verticalmatrixBIG};
			\node at (6.2,0) {$\ldots$};
			\draw (-.65,0)--(-.65,-.5)--(7.8,-.5)--(7.8,0);
			\end{tikzpicture}  \ = \ 
			\begin{tikzpicture}[baseline=-1mm]
			\foreach \i in {0,2 } \pic[pic text = $B$] at (\i,0) {mpstensorBIG};
			\foreach \i in {1,3,4.8} \pic[pic text = $A$] at (\i,0) {mpstensorBIG};
			\node at (3.9,0) {$\ldots$};	
			\draw (-.65,0)--(-.65,-.5)--(5.45,-.5)--(5.45,0);
			\end{tikzpicture} \ .
		\end{equation*} 
		We eliminate all $B$s as before and are left with:
		
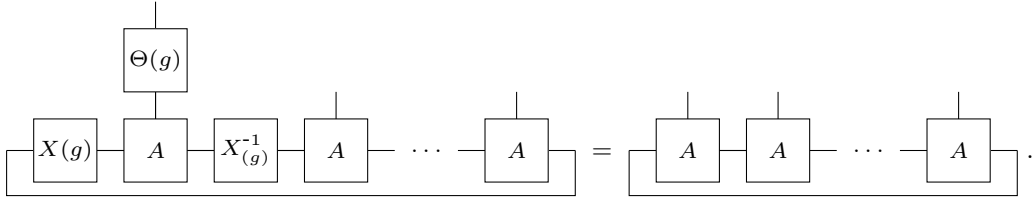
\begin{equation*}
			\begin{tikzpicture}[baseline=-1mm]
			\foreach \i in {1,3,5} \pic[pic text = $A$] at (\i,0) {mpstensorBIG};
			\foreach \i in {0,2 } \pic at (\i,0) {horizontalmatrixBIG};		
			\pic[pic text = {$\CC(g)$}] at (1,1) {verticalmatrixBIG};
			\node at (0,0) {$X(g)$};
			\node at (2,0) {$\myinv{X}_{(g)}$};
			\node at (4,0) {$\ldots$};
			\draw (-.65,0)--(-.65,-.5)--(5.65,-.5)--(5.65,0);
			\end{tikzpicture}  \ = \ 
			\begin{tikzpicture}[baseline=-1mm]
			\foreach \i in {0,1,3} \pic[pic text = $A$] at (\i,0) {mpstensorBIG};
			\node at (2,0) {$\ldots$};	
			\draw (-.65,0)--(-.65,-.5)--(3.65,-.5)--(3.65,0);
			\end{tikzpicture} \ .
		\end{equation*}	
		We can now use \cref{lem:BAAAA_CFequ} with $S^i=A^i$ and $T^i= X(g) \sum_{i^\prime } \CC(g)_{ii^\prime}A^{i^\prime} X(g)^{-1}$ to finish the proof (this is where we use the assumption about the block structure of $A$, the crucial thing is that $X(g)$ is compatible with $A$'s blocks as in \cref{lem:BAAAA_CFequ}).
	\end{proof}

\subsection{Matter and gauge field MPV}

	\VirtualGroupTrans*
	\begin{proof}
		Apply \cref{thm:OnlyFieldGaugeGroup} on the tensor $AB$ and the representations $\tilde{\RR}(g):=\II\otimes\RR(g)$ and $\tilde{\LL}(g):=\CC(g)\otimes\LL(g)$ to obtain:
		\begin{equation}  \label{eq:RRRRXXXX}
			\begin{tikzpicture}[baseline=-1mm]
			\pic[pic text = $A$] at (0,0) {mpstensorBIG};
			\pic[pic text = $B$] at (1,0) {mpstensorBIG};
			\pic[pic text = $\RR(g)$] at (1,1) {verticalmatrixBIG};
			\end{tikzpicture} \  = \ 
			\begin{tikzpicture}[baseline=-1mm]
			\pic[pic text = $A$] at (0,0) {mpstensorBIG};
			\pic[pic text = $B$] at (1,0) {mpstensorBIG};
			\pic[pic text =  $X(g)$] at (2,0) {horizontalmatrixBIG};
			\end{tikzpicture} \ , 
		\end{equation} 
		and
		\begin{equation}  
			\begin{tikzpicture}[baseline=-1mm]
			\pic[pic text = $A$] at (0,0) {mpstensorBIG};
			\pic[pic text = $B $] at (1,0) {mpstensorBIG};
			\pic[pic text = $\CC(g)$] at (0,1) {verticalmatrixBIG};
			\pic[pic text = $\LL(g)$] at (1,1) {verticalmatrixBIG};
			\end{tikzpicture} \   = \ 
			\begin{tikzpicture}[baseline=-1mm]
			\pic[pic text =  $\myinv{X{(g)}}$] at (0,0) {horizontalmatrixWIDE};
			\pic[pic text = $A$] at (1.15,0) {mpstensorBIG};
			\pic[pic text = $B$] at (2.15,0) {mpstensorBIG};
			\end{tikzpicture} \ ,
		\end{equation} 
	where $X(g)$ is a projective representation with the same multiplier as ${\RR}(g)$.	Apply \cref{thm:OnlyFieldGaugeGroup} once more, this time on the tensor $BA$ and the representations $\tilde{\RR}(g):=\RR(g)\otimes\CC(g)$ and $\tilde{\LL}(g):=\LL(g)\otimes\II$ to obtain:

		 \begin{equation}  \label{eq:AAAAPRE}
			\begin{tikzpicture}[baseline=-1mm]
			\pic[pic text = $B$] at (0,0) {mpstensorBIG};
			\pic[pic text = $A $] at (1,0) {mpstensorBIG};
			\pic[pic text = $\RR(g)$] at (0,1) {verticalmatrixBIG};
			\pic[pic text = $\CC(g)$] at (1,1) {verticalmatrixBIG};
			\end{tikzpicture} \   = \ 
			\begin{tikzpicture}[baseline=-1mm]
			\pic[pic text = $B$] at (0,0) {mpstensorBIG};
			\pic[pic text = $A$] at (1,0) {mpstensorBIG};
			\pic[pic text =  $Y(g)$] at (2,0) {horizontalmatrixBIG};
			\end{tikzpicture} \ ,
		\end{equation} 
		 and
		\begin{equation}     \label{eq:LLLLYYYY}
			\begin{tikzpicture}[baseline=-1mm]
			\pic[pic text = $B$] at (0,0) {mpstensorBIG};
			\pic[pic text = $A$] at (1,0) {mpstensorBIG};
			\pic[pic text = $\LL(g)$] at (0,1) {verticalmatrixBIG};
			\end{tikzpicture} \  = \ 
			\begin{tikzpicture}[baseline=-1mm]
			\pic[pic text =  $\myinv{Y(g)}$] at (0,0) {horizontalmatrixWIDE};
			\pic[pic text = $B$] at (1.15,0) {mpstensorBIG};
			\pic[pic text = $A$] at (2.15,0) {mpstensorBIG};
			\end{tikzpicture}  \ ,
		\end{equation} 
		where $Y(g)$ is a projective representation with inverse multiplier to $\LL(g)$.
		By contracting  \cref{eq:RRRRXXXX}  from the left with the tensor $BA\ldots B$, and taking the appropriate linear combination to obtain the identity matrix out of the tensor $BA\ldots BA$ (using the normality of $BA$), we eliminate the the $A$ in \cref{eq:RRRRXXXX}). By contracting  \cref{eq:LLLLYYYY} with $BA\ldots B$ from the right - we eliminate the $A$ in  \cref{eq:LLLLYYYY} (using the normality of $AB$). This proves the transformation rule for $B$ - \cref{eq:BRightGroupTrans}.
	 	Next plug in the transformation rules of $B$ under $\RR(g)$  into \cref{eq:AAAAPRE} to obtain:
		\begin{equation}  
		 	 \begin{tikzpicture}[baseline=-1mm]
		 	 \pic[pic text = $B$] at (0,0) {mpstensorBIG};
		 	 \pic[pic text = $X(g)$] at (1,0) {mpstensorBIG}; 
		 	 \pic[pic text = $A $] at (2,0) {mpstensorBIG};
		 	 \pic[pic text = $\CC(g)$] at (2,1) {verticalmatrixBIG};
		 	 \end{tikzpicture} \   = \ 
		 	 \begin{tikzpicture}[baseline=-1mm]
		 	 \pic[pic text = $B$] at (0,0) {mpstensorBIG};
		 	 \pic[pic text = $A$] at (1,0) {mpstensorBIG};
		 	 \pic[pic text =  $Y(g)$] at (2,0) {horizontalmatrixBIG};
		 	 \end{tikzpicture} \ .
		\end{equation} 
	 	 Finally, eliminate the $B$ from the equation as in the previous steps to obtain the transformation rule for $A$ and finish the proof.
	\end{proof}
	
	\RedToNormal*
	
	\begin{proof}
		We argue similarly to  \cite{Cirac2017} where it is described how to obtain, from an arbitrary tensor, a tensor in CF generating the same MPV. Begin by finding all of $AB$'s minimal invariant subspaces $S_\alpha$, such that $A^iB^jP_\alpha = P_\alpha A^iB^jP_\alpha$ for all $i$ and $j$, where $P_\alpha$ is the orthogonal projection to $S_\alpha$. Let $\hat{P_\alpha}$ be the partial isometry $\hat{P}_\alpha:\mathbb{C}^{D_1} \rightarrow S_\alpha$ such that $\hat{P}_\alpha^\dagger \hat{P_\alpha} = P_\alpha$ and $\hat{P}_\alpha \hat{P}_\alpha^\dagger = \II|_{S_\alpha}$. Define $A^i_\alpha :=\hat{P}_\alpha A^i$ and $ B^j_\alpha :=  B^j \hat{P}_\alpha^\dagger$. Then
		\begin{equation*}
			\begin{split}
			\ket{\psi^N_{AB}} = 
			&\sum_{\{i\},\{j\}} \tr \left( A^{i_1}B^{j_1} \ldots A^{i_N}B^{j_N} \right) \ket{i_1 j_1 \ldots i_N j_N}  \\
			= &\sum_{\{i\},\{j\},\alpha} \tr \left( P_\alpha A^{i_1}B^{j_1} \ldots  A^{i_N}B^{j_N} P_\alpha \right) \ket{i_1 j_1 \ldots i_N j_N}  \\
			= &\sum_{\{i\},\{j\},\alpha} \tr \left( P_\alpha A^{i_1}B^{j_1} P_\alpha  \ldots P_\alpha  A^{i_N}B^{j_N} P_\alpha \right) \ket{i_1 j_1 \ldots i_N j_N}  \\
			= &\sum_{\{i\},\{j\},\alpha} \tr \left( \hat{P}_\alpha A^{i_1}B^{j_1} \hat{P}_\alpha^\dagger \hat{P_\alpha} \ldots  \hat{P}_\alpha^\dagger \hat{P_\alpha} A^{i_N}B^{j_N} \hat{P}_\alpha^\dagger \right) \ket{i_1 j_1 \ldots i_N j_N}  \\
			= & \sum_{\alpha} \ket{\psi^N_{A_\alpha B_\alpha}}	\ .	
			\end{split}
		\end{equation*}
		Note that the bond dimension of the tensor $A_\alpha B_\alpha$ is $dim(S_\alpha)$ which is smaller than the original bond dimension $D_2$.
		Now $A_\alpha B_\alpha$ has no invariant subspaces but $B_\alpha A_\alpha $ might, therefore, perform the same for   $B_\alpha A_\alpha$ - for each $\alpha$  find all minimal invariant subspaces $T_{\alpha\beta}$ of $B_\alpha A_\alpha$. Let $Q_{\alpha\beta}$ be the orthogonal projections to the invariant subspaces and $\hat{Q}_{\alpha\beta}$   the partial isometries.
		Define  
		$A^i_{\alpha\beta} := A^i_\alpha \hat{Q}_{\alpha\beta}^\dagger = \hat{P}_\alpha A^i \hat{Q}_{\alpha\beta}^\dagger$,
		and 
		$ B^j_{\alpha\beta} := \hat{Q}_{\alpha\beta} B^j_\alpha = \hat{Q}_{\alpha\beta} B^j \hat{P}_\alpha^\dagger$.
		For each $\alpha$ we have 
		\begin{equation*}
			\ket{\psi^N_{A_\alpha B_\alpha}} =\sum_{\beta} \ket{\psi^N_{A_{\alpha\beta} B_{\alpha\beta} }}		\ ,
		\end{equation*}
		and thus 
		\begin{equation*}
			\ket{\psi^N_{AB}} =  \sum_{\alpha} \ket{\psi^N_{A_\alpha B_\alpha}}	= 	\sum_{\alpha\beta} \ket{\psi^N_{A_{\alpha\beta} B_{\alpha\beta}}} \ .
		\end{equation*}
		Now each $A_{\alpha \beta}B_{\alpha \beta}$ might be reducible. 		
		Continue iterating this decomposition, once for $AB$ and once for $BA$. Since the bond dimension of the tensors obtained at each step decreases, this procedure is bound to end after a finite number of steps. In the final step, we obtain the tensors 
		$A^i_\chi = \hat{P}_\chi A^i \hat{Q}_{\chi}^\dagger $
		and $B^j_\chi = \hat{Q}_{\chi} B^j \hat{P}_\chi^\dagger$, where $\chi$ incorporates all the previous indices, such that both $A_\chi B_\chi$ and $B_\chi A_\chi$ have no non trivial invariant subspaces.
		We can then perform the second step (as in \cite{Cirac2017}) which involves blocking the tensors in order to eliminate the periodicity of the associated CP maps. The blocking scheme is the following: $\tilde{A}^{ijk}:= A^iB^jA^k$ and $\tilde{B}^{lmn} := B^lA^mB^n$.
		%Under this scheme $\tilde{B}\tilde{A}$ is $BABABA$ and $\tilde{A}\tilde{B}$ is $ABABAB$. 
		We can find the least common multiple of the length needed to eliminate the periodicity of  all CP maps, and perform step 1 again if needed (after blocking the CP maps again become reducible \cite{Wolf2012a}). We can repeat these steps as many times as needed. The process terminates at some point because the bond dimension decreases at each step. Finally, rescale the matrices $A_\chi B_\chi$ by a constant $\mu_\chi$ to make the spectral radius  of $E_{A_\chi B_\chi}$ and $E_{B_\chi A_\chi}$ equal to $1$. The following lemma is required:
		 \begin{lemma}
		 	$E_{A_\chi B_\chi}$ and $E_{B_\chi A_\chi}$ have the same spectral radius.
		 \end{lemma}
		 \begin{proof}
		 	Let $X$ be an eigenvector of $E_{A_\chi B_\chi}$ with eigenvalue $\lambda$: 
		 	$E_{A_\chi B_\chi}(X)=E_{A_\chi} E_{B_\chi}(X) = \lambda X$. Apply $E_{B_\chi}$ to both sides to obtain $E_{B_\chi A_\chi} E_{B_\chi}(X)= \lambda E_{B_\chi}(X)$, i.e., $E_{B_\chi}(X)$ is an eigenvector of $E_{B_\chi A_\chi}$ with eigenvalue $\lambda$. Interchanging $A$ and $B$ we obtain that $E_{A_\chi B_\chi}$ and $E_{B_\chi A_\chi}$ have the same spectrum, and therefore the same spectral radius.
		 \end{proof}
	\end{proof}
	
	\begin{remark}[Blocking of the symmetry operators]
		In the blocking scheme described in \cref{prop:RedToNormal}, if we start out with a MPV with a local symmetry under the operators $\RR(g)\otimes\CC(g)\otimes\LL(g)$, after blocking we need to redefine the operators to act on the blocked degrees of freedom as follows:
		$\tilde{\RR}(g):= \RR(g)\otimes \CC(g) \otimes ( \LL(g) \RR(g))$,
		$\tilde{\CC}(g):= \CC(g)\otimes(\LL(g) \RR(g))\otimes\CC(g)$ and
		$\tilde{\LL}(g):= (\LL(g) \RR(g)) \otimes \CC(g)\ \otimes  \LL(g)$.
	\end{remark}
	
	\EveryNormCompInv*

	\begin{proof}
		Pick a BNT  $\{ A_j B_j\}$ out of the normal tensors  $\{ A_\chi B_\chi\}$ and construct a new tensor $C$ by blocking the tensors $\{ A_\chi B_\chi\}$ diagonally (possibly changing the order of the blocks):
		\begin{equation*}
			C^{ii^\prime} = \oplus_\chi \mu_\chi A^i_\chi B^{i^\prime}_\chi = \oplus_j \oplus_q \mu_{j,q}  V_{j,q}^{-1}  {A^i_j B^{i^\prime}_j} V_{j,q} \ ,
		\end{equation*}
		where for every $\chi$ there is a $j$ and a $q$ such that  $\mu_\chi A_\chi B_\chi= \mu_{j,q}  V_{j,q}^{-1}  {A^i_j B^{i^\prime}_j} V_{j,q}$. Now $C$  is in CF and generates the same   MPV as $AB$. We have
		\begin{equation*}
			O\ket{\psi^N_C}=	O\ket{\psi^N_{AB}} = \ket{\psi^N_{AB}}  = \ket{\psi^N_C} \ .
		\end{equation*}
		We can now use \cref{lem:BAAAA_CFequ} 
		(use \cref{eq:BAAAA_CFequ} from the proof of the lemma) for the tensor $C=AB$ to obtain 
		\begin{equation*}
			\begin{tikzpicture}[baseline=-1mm]
			\foreach \i in {0,2} \pic[pic text = $A_j$] at (\i,0) {mpstensor};
			\foreach \i in {1,3,5} \pic[pic text = $B_j$] at (\i,0) {mpstensor};
			\foreach \i in {0,1,2,3,5} \draw (\i,1)--(\i,1.25);
			\draw  (-0.25,1)--(5.25,1)--(5.25,0.5)--(-0.25,0.5)--(-0.25,1) ;
			\node at (4,0) {$\ldots$};
			\node at (2,.75) {$ ( \II \otimes)\, O \, (\otimes \II)$};
			\end{tikzpicture} \ = \ 
			\begin{tikzpicture}[baseline=-1mm]
			\foreach \i in {0,2} \pic[pic text = $A_j$] at (\i,0) {mpstensor};
			\foreach \i in {1,3,5} \pic[pic text = $B_j$] at (\i,0) {mpstensor};
			\node at (4,0) {$\ldots$};
			\end{tikzpicture} \ ,
		\end{equation*} 
		where the operator in the box contains $O$ (we need to extend it by at most one $\otimes \II$ from the right and from the left in order  to occupy a full $AB\ldots AB$ block). Finally, we have
		\begin{equation*}
			O\ket{\psi^N_{A_\chi B_\chi}}=
			O\ket{\psi^N_{V_{j,q}^{-1}  {A_j B_j} V_{j,q}}} 
			=	\ket{\psi^N_{   {A_j B_j} }} = \ket{\psi^N_{A_\chi B_\chi}}
		\end{equation*}
	\end{proof}
	
	Recall the definition of an elementary $A$ block:
	\ElementaryAblock*
	
	\WigEckClassA*
	
	\begin{proof}
		Write out \cref{eq:AGroupTrans}:
		\begin{equation*}
			\sum_{i^\prime} \Theta(g)_{ii^\prime } A^{i^\prime} = X(g)^{-1} A^i Y(g) \ . 
			\end{equation*}
		Taking the complex conjugate of both sides
		\begin{equation*} 
			\sum_{i^\prime} \Theta{(g^{-1})}_{i^\prime i}  \overline{A^{i^\prime}} = 
			\overline{X(g)^{-1}}  \overline{A^i}  \overline{Y(g)}  \, 
		\end{equation*}
		we see that $\vec{\overline{A}}$ satisfies \cref{eq:VecOpDef} for $\vec{v}=\vec{e^i}$ and the group element $g^{-1}$, with $\kappa=\CC(g), \pi = \overline{X(g)}$ and $\eta = \overline{Y(g)}$.
		Therefore $\vec{\overline{A}}$ is a vector operator with respect to the above representations. In the case when $\CC(g)={D^{J_0}(g)},X(g)={D_\gamma^j(g)}$ and $Y(g)={D_{\myinv{\gamma}}^l(g)}$ are irreducible representations, according to \cref{thm:GenWigEck} $\overline{A}$ is of the form: 
		\begin{equation*}
			\overline{A^M} = \sum_{J:{D^j(g)}={D^{J_0}(g)}} \alpha_J \sum_{m,n}\innerCG{\overline{j},m;l,n}{J,M} \ket{m}\bra{n} \ ,
		\end{equation*}
		taking the complex conjugate, we find the desired form of $A$.
	\end{proof}

	\begin{example} \label{ex:Aconstr1}
		A direct calculation using the Clebsch-Gordan series \cite{Klimyk}:
		\begin{equation*}
			{D^j(g)}_{m,m'}{D^l(g)}_{n,n'} = \sum_{L,N,N'}\innerCG{j,m;l,n}{L,N}\innerCG{L,N'}{j,m';l,n'}{D^l(g)}_{N,N'} 
		\end{equation*}
		shows that the tensor composed of the matrices
		\begin{equation*}
			{A^{J,M}} =   \sum_{m,n}\innerCG{J,M}{\overline{j},m;l,n} \ket{m}\bra{n} \ , 
		\end{equation*} 
		for a fixed value of $J$, satisfies 
		\begin{equation*} 
			\begin{tikzpicture}[baseline=-1mm]
			\pic[pic text = $A$] at (0,0) {mpstensorBIG};
			\pic[pic text = ${D^J(g)}$] at (0,1) {verticalmatrixWIDE};
			%\pic[pic text = $\oplus_J {D^J{(g)}}$] at (0,1) {verticalmatrixSuperWIDE};
			\end{tikzpicture} \ = \ 
			\begin{tikzpicture}[baseline=-1mm]
			\pic[pic text =  $\myinv{{D^j}(g)}$] at (0,0) {horizontalmatrixWIDE};
			\pic[pic text = $A$] at (1.15,0) {mpstensorBIG};
			\pic[pic text = ${D^l(g)}$] at (2.3,0) {horizontalmatrixWIDE};
			\end{tikzpicture} \ .
		\end{equation*} 
		Consequently, the tensor composed out of all matrices $\{A^{J,M}\}_{J\in\mathfrak{J},M}$  (all $J$ appearing in the decomposition  $\overline{{D^j(g)}}\otimes {D^l(g)} = \oplus_{J\in\mathfrak{J}} {D^J(g)}$)  satisfies:  
		\begin{equation*} 
			\begin{tikzpicture}[baseline=-1mm]
			\pic[pic text = $A$] at (0,0) {mpstensorBIG};
			%\pic[pic text = ${D^J(g)}$] at (0,1) {verticalmatrixWIDE};
			\pic[pic text =  $\oplus_{J\in\mathfrak{J}} {D^J{(g)}}$] at (0,1) {verticalmatrixSuperWIDE};
			\end{tikzpicture} \ = \ 
			\begin{tikzpicture}[baseline=-1mm]
			\pic[pic text =  $\myinv{{D^j}(g)}$] at (0,0) {horizontalmatrixWIDE};
			\pic[pic text = $A$] at (1.15,0) {mpstensorBIG};
			\pic[pic text = ${D^l(g)}$] at (2.3,0) {horizontalmatrixWIDE};	
			\end{tikzpicture} \ .
		\end{equation*} 
		
		In addition to being a symmetric tensor, this tensor is always injective:  
		let $D:=dim(j)=dim(l)$. Due to the fact that the C-G coefficients are the entries of a unitary matrix, the matrices $A^{J,M}$ satisfy $\tr\left({A^{J,M}}^\dagger A^{J^\prime,M^\prime} \right) = \delta_{J,J^\prime} \delta_{M,M^\prime}$. Since there are $D\times D$ of them, they form an ONB of the space of $D\times D$ matrices. 
	\end{example}
	We can now prove the following proposition, the proof of which we postponed in the previous section.
	\NonTrivAcoupling*
	
	\begin{proof}
		For each ${D_\gamma^{j_k}(g)}$ appearing in  $X(g)=\oplus^s_{k=1} {D_\gamma^{j_k}(g)}$, let  $J(k)$ be an irreducible representation index appearing in the decomposition of  $\overline{D_\gamma^{j_k}(g)}\otimes {{D_\gamma^{j_k}(g)}}$. Let $A^{(k)}$ be the tensor presented in \cref{ex:Aconstr1}, satisfying 
		\begin{equation*} 
			\begin{tikzpicture}[baseline=-1mm]
			\pic[pic text = $A^{(k)}$] at (0,0) {mpstensorBIG};
			\pic[pic text = ${D^{J(k)}(g)}$] at (0,1) {verticalmatrixSuperWIDE};
			\end{tikzpicture} \ = \ 
			\begin{tikzpicture}[baseline=-1mm]
			\pic[pic text =  $\myinv{{D^{j_k}}(g)}$] at (0,0) {horizontalmatrixWIDE1};
			\pic[pic text = $A^{(k)}$] at (1.35,0) {mpstensorBIG};
			\pic[pic text = ${D^{j_k}(g)}$] at (2.5,0) {horizontalmatrixWIDE};
			\end{tikzpicture} \ .
		\end{equation*} 	
		Let the matter Hilbert space be $\mathcal{H}_A := \oplus_k \mathcal{H}_{J(k)}$. Let the tensor $A$ in each sector $J(k)$ of the physical space  be zero except for in the $k,k$ virtual block, such that:
		\begin{equation*}
			\left[X^{-1}(g)A^{J_k,M}X(g) \right]_{l,l^\prime}  =\delta(l,k) \delta(l^\prime,k) D^{J_k}_{M,M^\prime}(g) A^{(k)J_k,M^\prime} \ .
			\end{equation*}
	\end{proof}
	
	\RXLYconnectionNormal*
	
	\begin{proof}
		\begin{enumerate}
		\item
			Assume the contrary is true, then according to \cref{prop:ElemBBlock}, $B^{k,m,n}$ are all zero and this value of $k$ does not contribute to the MPV.
		\item \label{item:Y}
			If there is a $Y^a(g)$ for which there is  not an appropriate $k$ then according to \cref{prop:ElemBBlock}, $B^{k,m,n}$ all have a zero row which is a contradiction to the normality of $AB$.  
		\item 
			As in \cref{item:Y}, $B^{k,m,n}$ now would have a zero column and would contradict normality of $BA$.
		\end{enumerate} 
	\end{proof}

	\LocalSymmWOGlobal*

	The proof is given by the following example:
 
	\begin{example}
		Let $G=D_{10}$ the dihedral group of order 10. It is  the group  generated by two elements: $r$  and $s$ satisfying $ r^5 = s^2 = (sr)^2 =e$. $D_{10}$ has two inequivalent two dimensional irreducible representations $\rho_1$ and $\rho_2$  generated by:
		\begin{equation*}
			\begin{split}
			\rho_1:\; & r\mapsto R_1  := 
			\left(
			\begin{array}{cc}
			e^{i\theta} &0\\
			0 & e^{-i\theta}
			\end{array}
			\right)\\
			& s\mapsto S  := 			
			\left(
			\begin{array}{cc}
			0 & 1\\
			1 & 0
			\end{array}
			\right)\\
			\rho_2:\; & r\mapsto R_2 := 
			\left(
			\begin{array}{cc}
			e^{i2\theta} &0\\
			0 & e^{-i2\theta}
			\end{array}
			\right)\\
			& s\mapsto S  := 			
			\left(
			\begin{array}{cc}
			0 & 1\\
			1 & 0
			\end{array}
			\right)	\ ,				
			\end{split}
		\end{equation*}
		where $\theta = 2\pi/5$.
		The tensor product  $\overline{\rho_1} \otimes \rho_2 $ decomposes into   
		$  \rho_1 \oplus \rho_2$:
		\begin{equation*}
			\begin{split}
			\overline{\rho_1} \otimes \rho_2 :\;  & r\mapsto R_1\otimes R_2  = 
			\left(
			\begin{array}{cccc}
			e^{i\theta} &0&0&0\\
			0 & e^{-i3\theta}&0&0\\
			0&0&e^{i3\theta}&0\\
			0&0&0&e^{-i\theta}
			\end{array}
			\right)\\
			 & s\mapsto S\otimes S  = 			
			\left(
			\begin{array}{cccc}
			0 &0&0& 1\\
			0&0&1 & 0\\
			0&1&0&0\\
			1&0&0&0
			\end{array}  
			\right)	\ .				
			\end{split}
		\end{equation*} 
		It is clear from inspection of the above $4\times 4$ matrices that the unitary transformation realizing the direct sum decomposition is a permutation of the basis elements, the non zero Clebsch-Gordan coefficients are:
		\begin{equation*}
			\begin{split}
			\innerCG{\rho_1, 1 }{\overline{\rho_1}, 1 ; \rho_2, 1} = &1 \\
			\innerCG{\rho_1, 2}{\overline{\rho_1}, 2; \rho_2, 2} = &1 \\
			\innerCG{\rho_2, 1}{\overline{\rho_1}, 1; \rho_2, 2} = &1 \\
			\innerCG{\rho_2, 2}{\overline{\rho_1}, 2; \rho_2, 1} = &1  \ .
			\end{split}
		\end{equation*}	
		Following \cref{ex:Aconstr1}, and using these coefficients, define the tensor $A$:
		
		\begin{equation*}
			A^1=\left(
			\begin{array}{cc}
			1 & 0\\
			0 & 0
			\end{array}
			\right) \;\;\;
			A^2=\left(
			\begin{array}{cc}
			0 & 0\\
			0 & 1
			\end{array}
			\right) \ .
		\end{equation*}
		$A$ satisfies:
		\begin{equation}  \label{eq:ex2A}
			\begin{tikzpicture}[baseline=-1mm]
			\pic[pic text = $A$] at (0,0) {mpstensorBIG};
			\pic[pic text = $\rho_1(g)$] at (0,1) {verticalmatrixBIG};
			\end{tikzpicture} \ = \ 
			\begin{tikzpicture}[baseline=-1mm]
			\pic[pic text = $\myinv{{\rho_1}(g)}$] at (0,0) {horizontalmatrixWIDE};
			\pic[pic text = $A$] at (1.15,0) {mpstensorBIG};
			\pic[pic text = $\rho_2(g)$] at (2.3,0) {horizontalmatrixWIDE};
			\end{tikzpicture} \ .
		\end{equation}
		According to \cref{prop:ElemBBlock} the following tensor $B$:
		
		\begin{equation*}
			\begin{split}
			B^{11}&=\left(
			\begin{array}{cc}
			1 & 0\\
			0 & 0
			\end{array}
			\right) \;\;\;
			B^{12}=\left(
			\begin{array}{cc}
			0 & 1\\
			0 & 0
			\end{array}
			\right)\\
			B^{21}&=\left(
			\begin{array}{cc}
			0 & 0\\
			1 & 0
			\end{array}
			\right) \;\;\;
			B^{22}=\left(
			\begin{array}{cc}
			0 & 0\\
			0 & 1
			\end{array}
			\right)  \ ,
			\end{split}
		\end{equation*}
		satisfies:
		\begin{equation} \label{eq:ex2B}
			\begin{tikzpicture}[baseline=-1mm]
			\pic[pic text = $B$] at (0,0) {mpstensorBIG};
			\pic[pic text = $\rho_1(g)$] at (0,1) {verticalmatrixBIG};
			\end{tikzpicture} \ = \ 
			\begin{tikzpicture}[baseline=-1mm]
			\pic[pic text = $B$] at (1,0) {mpstensorBIG};
			\pic[pic text = $\rho_1(g)$] at (2,0) {horizontalmatrixBIG};
			\end{tikzpicture} \;\;\;\;\;\; ; \;\;\;\;\;  
			\begin{tikzpicture}[baseline=-1mm]
			\pic[pic text = $B$] at (0,0) {mpstensorBIG};
			\pic[pic text = $\overline{\rho_2(g)}$] at (0,1) {verticalmatrixBIG};
			\end{tikzpicture} \ = \ 
			\begin{tikzpicture}[baseline=-1mm]
			\pic[pic text = $\myinv{{\rho_2}(g)}$] at (0,0) {horizontalmatrixWIDE};
			\pic[pic text = $B$] at (1.15,0) {mpstensorBIG};
			%	\pic[pic text = $\rho_1$] at (2,0) {horizontalmatrixBIG};
			\end{tikzpicture} \ .
		\end{equation}  
		
		\cref{eq:ex2A}) and \cref{eq:ex2B} are easily verified for the generators of the group, $r$ and $s$, and therefore hold for any group element.  From these equations it follows that $\ket{\psi^N_{AB}} $ has a local symmetry (\cref{def:BABSymm} with 
		$\RR(g) = \rho_1(g)$, $\CC(g) = \rho_1(g)$ and $\LL(g) = \overline{\rho_2(g)}$);  however, $\rho_1$ is not a global symmetry for $\ket{\psi^N_A}$, 
		%If it were, according to \cref{prop:traceless} the matrices $A^i$ would have to be traceless.
		as is easily verified for a MPV of length $1$.
		Similarly, a direct computation shows $\RR(g)\otimes\LL(g) \ket{\psi^2_{B}} \neq \ket{\psi^2_{B}}$. 
	\end{example}

	\GaugeGlobSymm*
	\begin{proof}
		As $X(g)$ appears in  \cref{eq:AGlobalSymm} together with its inverse, it is defined only up to a phase. As we assumed all $X_j(g)$ are from the same cohomology class, we can lift each one of them to be projective representations with the same multiplier $\gamma$.  We can assume without loss of generality (same argument as in \cref{rem:Xfreedom}) that each  $X_j(g)$  is block diagonal:
		$X(g)=\oplus_j \oplus_q \oplus_{a_j} {D_\gamma^{a_j}(g)}$. Set $\RR(g) = X(g)$, $\LL(g) = \overline{X(g)}$ and let $B$ be completely block diagonal: 
		\begin{equation*}
		B^{j,q,a_j;m,n} = \ket{j,q,a_j;m}\bra{j,q,a_j;n} \ ,
		\end{equation*}
		i.e., for each irreducible block of $X(g)$ there is a corresponding sector in $\mathcal{H}_B$:
		\begin{equation*}
		\mathcal{H}_B = \oplus_j \oplus_q \oplus_{a_j}
		\mathcal{H}_{\overline{a_j}}\otimes \mathcal{H}_{a_j}  \ ,
		\end{equation*}
		where $\overline{a_j}$ is the complex conjugate representation to $a_j$.
	\end{proof}

	\begin{example}[An $SU(2)$ gauge invariant MPV]
		
		For $G=SU(2)$ we demonstrate the construction of a general locally invariant MPV emphasizing the constituents of   physical theories and relating our setting and notation to  \cite{Zohar:2014qma,Zohar:2015jnb}. 
		Write the irreducible representations $D^j(g)$ in terms of their generators:
		\begin{equation*}
			D^j(g) =  \exp \left( i \sum_a  \tau _a^j{\varphi_a(g)} \right) , \; \forall g\in SU(2),
		\end{equation*}
		where $\{\varphi_a(g) \}_{a=1}^3$ are real parameters and  
		$\{ \tau^j_a\}_{a=1}^3$ are Hermitian $(2j+1)\times (2j+1)$ matrices satisfying the  $\mathfrak{su}(2)$ Lie algebra relations:
		\begin{equation*}
			\left[ \tau^j_a , \tau^j_b \right] =  i \varepsilon_{abc}\tau^j_c \ , 
		\end{equation*}
		where $ \varepsilon_{abc}$ is the totally antisymmetric tensor. 
		Let $D^r$ and $D^l$ be two irreducible representations of $SU(2)$ and 
		let $\mathfrak{J}_0$ be the set of   irreducible representation indices appearing in the decomposition of the tensor product: $\overline{D^r(g)}\otimes D^l(g) \cong \oplus_{J\in \mathfrak{J}_0 } D^J(g)$. Let $\mathfrak{J}\subseteq \mathfrak{J}_0$.  
		Define  the representation  $\CC(g)$ as generated by 
		$\{Q_a:= \bigoplus_{J\in\mathfrak{J}} \tau_a^J \}_{a=1}^3$: 
		\begin{equation*}
			\CC(g) =  \bigoplus_{J\in\mathfrak{J}} D^J(g) = 
			\bigoplus_{J\in\mathfrak{J}} \exp \left( i \sum_a  \tau _a^J{\varphi_a(g)} \right) =  
			\exp \left( i \sum_a Q_a{\varphi_a(g)} \right) \ .
		\end{equation*}
		As in  \cref{ex:Aconstr1}, the  tensor $A$, defined by the matrices:
		\begin{equation} \label{eq:A_SUN_ex}
			A^{J,M} =  \sum_{m,n} \alpha_J
			\innerCG{J,M}{\overline{r},m;l,n} \ket{m}\bra{n} \ , J\in \mathfrak{J}, M =1,\ldots,dim(J) \ 
		\end{equation}
		satisfies:
		\begin{equation*} 
			\begin{tikzpicture}[baseline=-1mm]
			\pic[pic text = $A$] at (0,0) {mpstensorBIG};
			\pic[pic text = $\CC(g)$] at (0,1) {verticalmatrixBIG};
			\end{tikzpicture} \ = \ 
			\begin{tikzpicture}[baseline=-1mm]
			\pic[pic text = $\myinv{{D^r}(g)}$] at (0,0) {horizontalmatrixWIDE};
			\pic[pic text = $A$] at (1.15,0) {mpstensorBIG};
			\pic[pic text = ${D^l(g)}$] at (2.3,0) {horizontalmatrixWIDE};
			\end{tikzpicture} \ .
		\end{equation*} 
		This relation, written in terms of the generators, reads:
		\begin{equation*}
			\sum_{M^\prime} \left[\exp \left( i \sum_a  \tau _a^J{\varphi_a(g)} \right)\right]_{M,M^\prime} A^{J,M^\prime}  = 
			\exp \left(- i \sum_a  \tau _a^r{\varphi_a(g)} \right)
			A^{J,M}
			\exp \left( i \sum_a  \tau _a^l{\varphi_a(g)} \right)  	\ .	 
		\end{equation*}
		Differentiating this equation with respect to any one of the  group parameters $\varphi_a $ we obtain the ``virtual Gauss law'' satisfied by $A$:
		\begin{equation*}
			Q_a:  A^{J,M} \mapsto
			\sum_{M^\prime} \left[\tau^J_a \right]_{M,M^\prime} A^{J,M^\prime} = 
			-\tau^r_a A^{J,M} + A^{J,M} \tau^l_a \ . 
		\end{equation*}
		Next,   add a gauge field degree of freedom to the matter MPV, described by a tensor:
		$ B^{m,n}=\ket{m}\bra{n} $,
		and define the transformations:
		\begin{equation*}
			\RR(g) = \II \otimes D^r(g) \;\;\; ; \;\;\; \LL(g) = \overline{D^l(g)} \otimes \II \ .
		\end{equation*}
		The action of $\LL(g)$ on the gauge field Hilbert space is given by:
		\begin{equation*}
			\LL(g) \ket{m,n} = 
			(\overline{D^l(g)} \otimes \II )\ket{m,n} = 
			\sum_{m^\prime} \overline{D^l(g)}_{m^\prime,m} \ket{m^\prime,n}  = 
			\sum_{m^\prime}  {D^l(g^{-1})}_{m,m^\prime} \ket{m^\prime,n} \ ; 
		\end{equation*}
		whereas $\RR(g)$ acts as:
		\begin{equation*}
			\RR(g) \ket{m,n} = 
			\sum_{n^\prime}  {D^r({g})}_{n^\prime,n} \ket{m,n^\prime} \ . 
		\end{equation*}
		
		$\RR(g)$ and $\LL(g)$ can be defined in terms of  right and left  generators  $\{R_a\}_{a=1}^3$ and $\{L_a\}_{a=1}^3$, as described in \cref{sub:GenAndGaussLaw}:
		\begin{equation*}
			\begin{split}
			\RR(g) =& \exp\left( i \sum_a  R_a{\varphi_a(g)}\right) \\
			\LL(g) =& \exp\left( i \sum_a  L_a{\varphi_a(g)}\right) \ . 
			\end{split}
		\end{equation*}
		In our case $R_a$ is simply given by  $\II \otimes \tau^r_a$  but in general $R_a$ and $L_a$ can have a block diagonal structure. 
		Define the generators of the local gauge transformation around lattice site $2K+1$:
		\begin{equation*}
			G_a^{[2K+1]}:=	\left(R_a^{[2K]} + Q_a^{[2K+1]} +L_a^{[2K+2]} \right)  \ .
		\end{equation*}	
		From our construction it follows that   for all $ g\in G$ and for all lattice sites $K$:
		\begin{equation*}
			\RR^{[2K]}(g)\otimes\CC^{[2K+1]}(g)\otimes \LL^{[2K+2]}(g)\ket{\psi^N_{AB}} = \ket{\psi^N_{AB}} \ . 
		\end{equation*}
		Once again, differentiating with respect to the group parameters $\varphi_a$ we obtain:
		\begin{equation} \label{eq:GaussLaw2}  
			\left(R_a^{[2K]} + Q_a^{[2K+1]} +L_a^{[2K+2]} \right) \ket{\psi^N_{AB}} = G_a^{[2K+1]} \ket{\psi^N_{AB}}
			=0 \ .
		\end{equation} 
		This is the lattice version of Gauss' law. 
		In  physical theories   $D^l = \overline{D^r}$ and thus states $\ket{\psi_A}$ have a global symmetry generated by $\{Q_a\}$ - the $SU(2)$ charge operators.    $R_a$ and $L_a$ are identified with right and left electric fields respectively
		\cite{Zohar:2014qma}.
	
		One could generalize the above construction for 
		\begin{equation*}
			\RR(g) = \oplus_k \left( \II \otimes D^{r_k}(g) \right) 
			\;\;\; ; \;\;\; 
			\LL(g) = \oplus_k \left( \overline{D^{l_k}(g)} \otimes \II  \right) \ 
		\end{equation*}
		by constructing $A$ and $B$ as above for each $k$ sector and combining  them together block diagonally (in both physical and virtual dimensions).  Duplicating the virtual representations while keeping the physical ones fixed can be achieved by $B^{m,n} \mapsto (B^{m,n}\oplus B^{m,n})  $, $A^{J,M} \mapsto (A_1^{J,M}\oplus A_2^{J,M})  $. This can  be used to enlarge the number of variational parameters. The tensors $A_1$ and $A_2$ must both have the same structure (\cref{eq:A_SUN_ex}) but can have different parameters $\alpha_J$. The generalization to of the above to $G=SU(N)$ is straightforward.
	\end{example}

\section{Summary}
	In this work, we studied and classified translationally invariant MPVs with a local (gauge) symmetry under arbitrary groups. The states we classified may involve two types of building blocks, $A$ and $B$ tensors, which represent matter and gauge fields respectively.
	
	We showed that matter-only MPVs may have   a local symmetry, when one transforms a single site, only if they are trivial (composed of products of invariant states at each site). We also classified pure gauge states, which involve only $B$ tensors and have local invariance when one transforms two neighboring sites, including the well-known structure of physical states involving only gauge fields. These two building blocks can be combined in a way that allows coupling matter fields (with global symmetry) to gauge fields (with local symmetry) in a locally symmetric manner, as in conventional gauge theory scenarios. Furthermore, we expanded the class of gauge invariant states to include ones that involve matter and gauge fields which do not possess the known symmetry properties when decoupled. We classified the structure and properties of such MPVs as well. 
	
	Further work shall include a generalization to further dimensions, i.e.\ using PEPS. In our work we were able to connect some of the results to the symmetry properties and structure of previous gauge invariant PEPS constructions \cite{Zohar:2015jnb,Tagliacozzo2014,Haegeman:2014maa} when the space dimension is reduced to one, and therefore higher dimensional generalizations in the spirit of the current work should be possible. In particular, the tensor describing the gauge field, as it resides on the links of a lattice, is a one dimensional object for any spatial dimension, and has shown, in some particular cases, properties known from previous PEPS studies. Another important generalization one should consider is a fermionic representation of the matter, combining the spirit of this work with previous works on fermionic PEPS with gauge symmetry \cite{Zohar:2015eda,Zohar:2016wcf} or with global symmetry \cite{BultinckVerstraete,Bultinck:2017orh}. From the physical point of view, a physical study aiming at understanding the new classes of gauge invariant states introduced in this paper, in which the matter and gauge field do not posses separate symmetries, may also potentially unfold new physical phenomena and phases.

\section{Acknowledgments}
	IK acknowledges the support of the DAAD.

\newpage

\bibliography{library} 

\begin{thebibliography}{10}
\expandafter\ifx\csname url\endcsname\relax
  \def\url#1{\texttt{#1}}\fi
\expandafter\ifx\csname urlprefix\endcsname\relax\def\urlprefix{URL }\fi
\expandafter\ifx\csname href\endcsname\relax
  \def\href#1#2{#2} \def\path#1{#1}\fi

\bibitem{Peskin:257493}
M.~E. Peskin, D.~V. Schroeder, \href{https://cds.cern.ch/record/257493}{{An
  Introduction to Quantum Field Theory; 1995 ed.}}, Westview, Boulder, CO,
  1995.
\newline\urlprefix\url{https://cds.cern.ch/record/257493}

\bibitem{Fradkin2013}
E.~Fradkin, Field Theories of Condensed Matter Physics, 2nd Edition, Cambridge
  University Press, 2013.

\bibitem{Wilson:1974sk}
K.~G. Wilson, {Confinement of Quarks}, Phys. Rev. D10 (1974) 2445--2459,
  [,45(1974)].
\newblock \href {http://dx.doi.org/10.1103/PhysRevD.10.2445}
  {\path{doi:10.1103/PhysRevD.10.2445}}.

\bibitem{Troyer:2004ge}
M.~Troyer, U.-J. Wiese, {Computational complexity and fundamental limitations
  to fermionic quantum Monte Carlo simulations}, Phys. Rev. Lett. 94 (2005)
  170201.
\newblock \href {http://arxiv.org/abs/cond-mat/0408370}
  {\path{arXiv:cond-mat/0408370}}, \href
  {http://dx.doi.org/10.1103/PhysRevLett.94.170201}
  {\path{doi:10.1103/PhysRevLett.94.170201}}.

\bibitem{Kogut1975}
J.~Kogut, L.~Susskind,
  \href{https://link.aps.org/doi/10.1103/PhysRevD.11.395}{Hamiltonian
  formulation of wilson's lattice gauge theories}, Phys. Rev. D 11 (1975)
  395--408.
\newblock \href {http://dx.doi.org/10.1103/PhysRevD.11.395}
  {\path{doi:10.1103/PhysRevD.11.395}}.
\newline\urlprefix\url{https://link.aps.org/doi/10.1103/PhysRevD.11.395}

\bibitem{Zohar:2015hwa}
E.~Zohar, J.~I. Cirac, B.~Reznik, {Quantum Simulations of Lattice Gauge
  Theories using Ultracold Atoms in Optical Lattices}, Rept. Prog. Phys. 79~(1)
  (2016) 014401.
\newblock \href {http://arxiv.org/abs/1503.02312} {\path{arXiv:1503.02312}},
  \href {http://dx.doi.org/10.1088/0034-4885/79/1/014401}
  {\path{doi:10.1088/0034-4885/79/1/014401}}.

\bibitem{Wiese:2013uua}
U.-J. Wiese, {Ultracold Quantum Gases and Lattice Systems: Quantum Simulation
  of Lattice Gauge Theories}, Annalen Phys. 525 (2013) 777--796.
\newblock \href {http://arxiv.org/abs/1305.1602} {\path{arXiv:1305.1602}},
  \href {http://dx.doi.org/10.1002/andp.201300104}
  {\path{doi:10.1002/andp.201300104}}.

\bibitem{White1992}
S.~R. White, \href{http://link.aps.org/doi/10.1103/PhysRevLett.69.2863}{Density
  matrix formulation for quantum renormalization groups}, Physical Review
  Letters 69~(19) (1992) 2863--2866.
\newblock \href {http://dx.doi.org/10.1103/PhysRevLett.69.2863}
  {\path{doi:10.1103/PhysRevLett.69.2863}}.
\newline\urlprefix\url{http://link.aps.org/doi/10.1103/PhysRevLett.69.2863}

\bibitem{Verstraete2008}
F.~Verstraete, V.~Murg, J.~I. Cirac,
  \href{http://arxiv.org/abs/0907.2796}{{Matrix product states, projected
  entangled pair states, and variational renormalization group methods for
  quantum spin systems}} 57~(2)  143--224.
\newblock \href {http://arxiv.org/abs/0907.2796} {\path{arXiv:0907.2796}},
  \href {http://dx.doi.org/10.1080/14789940801912366}
  {\path{doi:10.1080/14789940801912366}}.
\newline\urlprefix\url{http://arxiv.org/abs/0907.2796}

\bibitem{Perez-Garcia2010}
D.~P\'{e}rez-Garc\'{\i}a, M.~Sanz, C.~E. Gonzalez-Guillen, M.~M. Wolf, J.~I.
  Cirac, \href{http://arxiv.org/abs/0908.1674v1;
  http://arxiv.org/pdf/0908.1674v1}{{A canonical form for Projected Entangled
  Pair States and applications}} 12  025010.
\newblock \href {http://dx.doi.org/10.1088/1367-2630/12/2/025010}
  {\path{doi:10.1088/1367-2630/12/2/025010}}.
\newline\urlprefix\url{http://arxiv.org/abs/0908.1674v1;
  http://arxiv.org/pdf/0908.1674v1}

\bibitem{Banuls2013}
M.~Ba{\~{n}}uls, K.~Cichy, J.~Cirac, K.~Jansen,
  \href{https://doi.org/10.1007/JHEP11(2013)158}{The mass spectrum of the
  schwinger model with matrix product states}, Journal of High Energy Physics
  2013~(11) (2013) 158.
\newblock \href {http://dx.doi.org/10.1007/JHEP11(2013)158}
  {\path{doi:10.1007/JHEP11(2013)158}}.
\newline\urlprefix\url{https://doi.org/10.1007/JHEP11(2013)158}

\bibitem{Shimizu2014}
Y.~Shimizu, Y.~Kuramashi,
  \href{https://link.aps.org/doi/10.1103/PhysRevD.90.014508}{Grassmann tensor
  renormalization group approach to one-flavor lattice schwinger model}, Phys.
  Rev. D 90 (2014) 014508.
\newblock \href {http://dx.doi.org/10.1103/PhysRevD.90.014508}
  {\path{doi:10.1103/PhysRevD.90.014508}}.
\newline\urlprefix\url{https://link.aps.org/doi/10.1103/PhysRevD.90.014508}

\bibitem{Buyens2014}
B.~Buyens, J.~Haegeman, K.~Van~Acoleyen, H.~Verschelde, F.~Verstraete,
  \href{https://link.aps.org/doi/10.1103/PhysRevLett.113.091601}{Matrix product
  states for gauge field theories}, Phys. Rev. Lett. 113 (2014) 091601.
\newblock \href {http://dx.doi.org/10.1103/PhysRevLett.113.091601}
  {\path{doi:10.1103/PhysRevLett.113.091601}}.
\newline\urlprefix\url{https://link.aps.org/doi/10.1103/PhysRevLett.113.091601}

\bibitem{Silvi2014}
P.~Silvi, E.~Rico, T.~Calarco, S.~Montangero,
  \href{http://stacks.iop.org/1367-2630/16/i=10/a=103015}{Lattice gauge tensor
  networks}, New Journal of Physics 16~(10) (2014) 103015.
\newline\urlprefix\url{http://stacks.iop.org/1367-2630/16/i=10/a=103015}

\bibitem{2014arXiv1411.0020B}
B.~{Buyens}, K.~{Van Acoleyen}, J.~{Haegeman}, F.~{Verstraete}, {Matrix product
  states for Hamiltonian lattice gauge theories}, ArXiv e-prints\href
  {http://arxiv.org/abs/1411.0020} {\path{arXiv:1411.0020}}.

\bibitem{PhysRevLett.112.201601}
E.~Rico, T.~Pichler, M.~Dalmonte, P.~Zoller, S.~Montangero,
  \href{https://link.aps.org/doi/10.1103/PhysRevLett.112.201601}{Tensor
  networks for lattice gauge theories and atomic quantum simulation}, Phys.
  Rev. Lett. 112 (2014) 201601.
\newblock \href {http://dx.doi.org/10.1103/PhysRevLett.112.201601}
  {\path{doi:10.1103/PhysRevLett.112.201601}}.
\newline\urlprefix\url{https://link.aps.org/doi/10.1103/PhysRevLett.112.201601}

\bibitem{Saito:2014bda}
H.~Saito, M.~C. Bañuls, K.~Cichy, J.~I. Cirac, K.~Jansen, {The temperature
  dependence of the chiral condensate in the Schwinger model with Matrix
  Product States}, PoS LATTICE2014 (2014) 302.
\newblock \href {http://arxiv.org/abs/1412.0596} {\path{arXiv:1412.0596}}.

\bibitem{Kuhn2015}
S.~K{\"u}hn, E.~Zohar, J.~I. Cirac, M.~C. Ba{\~{n}}uls,
  \href{https://doi.org/10.1007/JHEP07(2015)130}{Non-abelian string breaking
  phenomena with matrix product states}, Journal of High Energy Physics
  2015~(7) (2015) 130.
\newblock \href {http://dx.doi.org/10.1007/JHEP07(2015)130}
  {\path{doi:10.1007/JHEP07(2015)130}}.
\newline\urlprefix\url{https://doi.org/10.1007/JHEP07(2015)130}

\bibitem{PhysRevD.92.034519}
M.~C. Ba\~nuls, K.~Cichy, J.~I. Cirac, K.~Jansen, H.~Saito,
  \href{https://link.aps.org/doi/10.1103/PhysRevD.92.034519}{Thermal evolution
  of the schwinger model with matrix product operators}, Phys. Rev. D 92 (2015)
  034519.
\newblock \href {http://dx.doi.org/10.1103/PhysRevD.92.034519}
  {\path{doi:10.1103/PhysRevD.92.034519}}.
\newline\urlprefix\url{https://link.aps.org/doi/10.1103/PhysRevD.92.034519}

\bibitem{PhysRevD.93.094512}
M.~C. Ba\~nuls, K.~Cichy, K.~Jansen, H.~Saito,
  \href{https://link.aps.org/doi/10.1103/PhysRevD.93.094512}{Chiral condensate
  in the schwinger model with matrix product operators}, Phys. Rev. D 93 (2016)
  094512.
\newblock \href {http://dx.doi.org/10.1103/PhysRevD.93.094512}
  {\path{doi:10.1103/PhysRevD.93.094512}}.
\newline\urlprefix\url{https://link.aps.org/doi/10.1103/PhysRevD.93.094512}

\bibitem{PhysRevX.6.011023}
T.~Pichler, M.~Dalmonte, E.~Rico, P.~Zoller, S.~Montangero,
  \href{https://link.aps.org/doi/10.1103/PhysRevX.6.011023}{Real-time dynamics
  in u(1) lattice gauge theories with tensor networks}, Phys. Rev. X 6 (2016)
  011023.
\newblock \href {http://dx.doi.org/10.1103/PhysRevX.6.011023}
  {\path{doi:10.1103/PhysRevX.6.011023}}.
\newline\urlprefix\url{https://link.aps.org/doi/10.1103/PhysRevX.6.011023}

\bibitem{Buyens:2016hhu}
B.~Buyens, J.~Haegeman, F.~Hebenstreit, F.~Verstraete, K.~Van~Acoleyen,
  {Real-time simulation of the Schwinger effect with Matrix Product
  States}\href {http://arxiv.org/abs/1612.00739} {\path{arXiv:1612.00739}}.

\bibitem{Silvi:2016cas}
P.~Silvi, E.~Rico, M.~Dalmonte, F.~Tschirsich, S.~Montangero, {Finite-density
  phase diagram of a (1+1)-d non-abelian lattice gauge theory with tensor
  networks}\href {http://arxiv.org/abs/1606.05510} {\path{arXiv:1606.05510}},
  \href {http://dx.doi.org/10.22331/q-2017-04-25-9}
  {\path{doi:10.22331/q-2017-04-25-9}}.

\bibitem{PhysRevD.93.085012}
A.~Milsted, \href{https://link.aps.org/doi/10.1103/PhysRevD.93.085012}{Matrix
  product states and the non-abelian rotor model}, Phys. Rev. D 93 (2016)
  085012.
\newblock \href {http://dx.doi.org/10.1103/PhysRevD.93.085012}
  {\path{doi:10.1103/PhysRevD.93.085012}}.
\newline\urlprefix\url{https://link.aps.org/doi/10.1103/PhysRevD.93.085012}

\bibitem{Banuls:2016hhv}
M.~C. Bañuls, K.~Cichy, J.~I. Cirac, K.~Jansen, S.~Kühn, H.~Saito, {The
  multi-flavor Schwinger model with chemical potential - Overcoming the sign
  problem with Matrix Product States}, PoS LATTICE2016 (2016) 316.
\newblock \href {http://arxiv.org/abs/1611.01458} {\path{arXiv:1611.01458}}.

\bibitem{PhysRevD.94.085018}
B.~Buyens, F.~Verstraete, K.~Van~Acoleyen,
  \href{https://link.aps.org/doi/10.1103/PhysRevD.94.085018}{Hamiltonian
  simulation of the schwinger model at finite temperature}, Phys. Rev. D 94
  (2016) 085018.
\newblock \href {http://dx.doi.org/10.1103/PhysRevD.94.085018}
  {\path{doi:10.1103/PhysRevD.94.085018}}.
\newline\urlprefix\url{https://link.aps.org/doi/10.1103/PhysRevD.94.085018}

\bibitem{Saito:2015ryj}
H.~Saito, M.~C. Bañuls, K.~Cichy, J.~I. Cirac, K.~Jansen, {Thermal evolution
  of the one-flavour Schwinger model with using Matrix Product States}, PoS
  LATTICE2015 (2016) 283.
\newblock \href {http://arxiv.org/abs/1511.00794} {\path{arXiv:1511.00794}}.

\bibitem{PhysRevLett.118.071601}
M.~C. Ba\~nuls, K.~Cichy, J.~I. Cirac, K.~Jansen, S.~K\"uhn,
  \href{https://link.aps.org/doi/10.1103/PhysRevLett.118.071601}{Density
  induced phase transitions in the schwinger model: A study with matrix product
  states}, Phys. Rev. Lett. 118 (2017) 071601.
\newblock \href {http://dx.doi.org/10.1103/PhysRevLett.118.071601}
  {\path{doi:10.1103/PhysRevLett.118.071601}}.
\newline\urlprefix\url{https://link.aps.org/doi/10.1103/PhysRevLett.118.071601}

\bibitem{2017EPJWC13704001B}
M.~C. {Ba{\~n}uls}, K.~{Cichy}, J.~{Ignacio Cirac}, K.~{Jansen}, S.~{K{\"u}hn},
  H.~{Saito}, Towards overcoming the {Monte Carlo} sign problem with tensor
  networks, in: European Physical Journal Web of Conferences, Vol. 137 of
  European Physical Journal Web of Conferences, 2017, p. 04001.
\newblock \href {http://arxiv.org/abs/1611.04791} {\path{arXiv:1611.04791}},
  \href {http://dx.doi.org/10.1051/epjconf/201713704001}
  {\path{doi:10.1051/epjconf/201713704001}}.

\bibitem{Banuls:2017ena}
M.~C. Bañuls, K.~Cichy, J.~I. Cirac, K.~Jansen, S.~Kühn, {Efficient basis
  formulation for 1+1 dimensional SU(2) lattice gauge theory: Spectral
  calculations with matrix product states}\href
  {http://arxiv.org/abs/1707.06434} {\path{arXiv:1707.06434}}.

\bibitem{Tagliacozzo2014}
L.~Tagliacozzo, A.~Celi, M.~Lewenstein,
  \href{https://link.aps.org/doi/10.1103/PhysRevX.4.041024}{Tensor networks for
  lattice gauge theories with continuous groups}, Phys. Rev. X 4 (2014) 041024.
\newblock \href {http://dx.doi.org/10.1103/PhysRevX.4.041024}
  {\path{doi:10.1103/PhysRevX.4.041024}}.
\newline\urlprefix\url{https://link.aps.org/doi/10.1103/PhysRevX.4.041024}

\bibitem{Haegeman:2014maa}
J.~Haegeman, K.~Van~Acoleyen, N.~Schuch, J.~I. Cirac, F.~Verstraete, {Gauging
  quantum states: from global to local symmetries in many-body systems}, Phys.
  Rev. X5~(1) (2015) 011024.
\newblock \href {http://arxiv.org/abs/1407.1025} {\path{arXiv:1407.1025}},
  \href {http://dx.doi.org/10.1103/PhysRevX.5.011024}
  {\path{doi:10.1103/PhysRevX.5.011024}}.

\bibitem{Zohar:2015eda}
E.~Zohar, M.~Burrello, T.~Wahl, J.~I. Cirac, {Fermionic Projected Entangled
  Pair States and Local U(1) Gauge Theories}, Annals Phys. 363 (2015) 385--439.
\newblock \href {http://arxiv.org/abs/1507.08837} {\path{arXiv:1507.08837}},
  \href {http://dx.doi.org/10.1016/j.aop.2015.10.009}
  {\path{doi:10.1016/j.aop.2015.10.009}}.

\bibitem{Zohar:2015jnb}
E.~Zohar, M.~Burrello, {Building Projected Entangled Pair States with a Local
  Gauge Symmetry}, New J. Phys. 18~(4) (2016) 043008.
\newblock \href {http://arxiv.org/abs/1511.08426} {\path{arXiv:1511.08426}},
  \href {http://dx.doi.org/10.1088/1367-2630/18/4/043008}
  {\path{doi:10.1088/1367-2630/18/4/043008}}.

\bibitem{Zohar:2016wcf}
E.~Zohar, T.~B. Wahl, M.~Burrello, J.~I. Cirac, {Projected Entangled Pair
  States with non-Abelian gauge symmetries: an SU(2) study}, Annals Phys. 374
  (2016) 84--137.
\newblock \href {http://arxiv.org/abs/1607.08115} {\path{arXiv:1607.08115}},
  \href {http://dx.doi.org/10.1016/j.aop.2016.08.008}
  {\path{doi:10.1016/j.aop.2016.08.008}}.

\bibitem{Zohar:2014qma}
E.~Zohar, M.~Burrello, {Formulation of lattice gauge theories for quantum
  simulations}, Phys. Rev. D91~(5) (2015) 054506.
\newblock \href {http://arxiv.org/abs/1409.3085} {\path{arXiv:1409.3085}},
  \href {http://dx.doi.org/10.1103/PhysRevD.91.054506}
  {\path{doi:10.1103/PhysRevD.91.054506}}.

\bibitem{Perez-Garcia2007}
D.~P\'{e}rez-Garc\'{\i}a, F.~Verstraete, M.~M. Wolf, J.~I. Cirac,
  \href{http://arxiv.org/abs/quant-ph/0608197v2;
  http://arxiv.org/pdf/quant-ph/0608197v2}{{Matrix Product State
  Representations}} 7  401.
\newline\urlprefix\url{http://arxiv.org/abs/quant-ph/0608197v2;
  http://arxiv.org/pdf/quant-ph/0608197v2}

\bibitem{Cirac2017}
J.~I. Cirac, D.~Perez-Garcia, N.~Schuch, F.~Verstraete,
  \href{http://arxiv.org/abs/1606.00608}{Matrix product density operators:
  Renormalization fixed points and boundary theories} 378  100--149.
\newblock \href {http://arxiv.org/abs/1606.00608} {\path{arXiv:1606.00608}},
  \href {http://dx.doi.org/10.1016/j.aop.2016.12.030}
  {\path{doi:10.1016/j.aop.2016.12.030}}.
\newline\urlprefix\url{http://arxiv.org/abs/1606.00608}

\bibitem{Wolf2012a}
M.~Wolf, Quantum channels and operations guided tour,
  https://www-m5.ma.tum.de/foswiki/pub/M5/Allgemeines/MichaelWolf/QChannelLecture.pdf.

\bibitem{Sanz2010}
M.~Sanz, D.~Perez-Garcia, M.~M. Wolf, J.~I. Cirac,
  \href{http://arxiv.org/abs/0909.5347}{A quantum version of wielandt's
  inequality} 56~(9)  4668--4673.
\newblock \href {http://arxiv.org/abs/0909.5347} {\path{arXiv:0909.5347}},
  \href {http://dx.doi.org/10.1109/TIT.2010.2054552}
  {\path{doi:10.1109/TIT.2010.2054552}}.
\newline\urlprefix\url{http://arxiv.org/abs/0909.5347}

\bibitem{HallQMath}
B.~C. Hall, Quantum Theory for Mathematicians., Vol. 267 of Graduate Texts in
  Mathematics, Springer New York, 2013.

\bibitem{Chuangxun}
C.~CHENG, \href{http://maths.nju.edu.cn/~ccheng/Writings/preps-notes.pdf}{A
  character theory for projective representations of finite group}.
\newline\urlprefix\url{http://maths.nju.edu.cn/~ccheng/Writings/preps-notes.pdf}

\bibitem{Klimyk}
N.~Vilenkin, A.~Klimyk, Representation of Lie Groups and Special Functions:
  Recent Advances, Springer, 1995, section 4.3.3.

\bibitem{Sanz2009}
M.~Sanz, M.~M. Wolf, D.~P\'erez-Garc\'{\i}a, J.~I. Cirac,
  \href{https://link.aps.org/doi/10.1103/PhysRevA.79.042308}{Matrix product
  states: Symmetries and two-body hamiltonians}, Phys. Rev. A 79 (2009) 042308.
\newblock \href {http://dx.doi.org/10.1103/PhysRevA.79.042308}
  {\path{doi:10.1103/PhysRevA.79.042308}}.
\newline\urlprefix\url{https://link.aps.org/doi/10.1103/PhysRevA.79.042308}

\bibitem{PollmannBergOshikawa}
F.~{Pollmann}, A.~M. {Turner}, E.~{Berg}, M.~{Oshikawa}, {Entanglement spectrum
  of a topological phase in one dimension}, prb 81~(6) (2010) 064439.
\newblock \href {http://arxiv.org/abs/0910.1811} {\path{arXiv:0910.1811}},
  \href {http://dx.doi.org/10.1103/PhysRevB.81.064439}
  {\path{doi:10.1103/PhysRevB.81.064439}}.

\bibitem{Agrawala1980}
V.~K. Agrawala, \href{http://dx.doi.org/10.1063/1.524639}{Wigner–{E}ckart
  theorem for an arbitrary group or lie algebra}, Journal of Mathematical
  Physics 21~(7) (1980) 1562--1565.
\newblock \href {http://arxiv.org/abs/http://dx.doi.org/10.1063/1.524639}
  {\path{arXiv:http://dx.doi.org/10.1063/1.524639}}, \href
  {http://dx.doi.org/10.1063/1.524639} {\path{doi:10.1063/1.524639}}.
\newline\urlprefix\url{http://dx.doi.org/10.1063/1.524639}

\bibitem{Perez-Garcia2008PhRvL}
D.~{P{\'e}rez-Garc{\'{\i}}a}, M.~M. {Wolf}, M.~{Sanz}, F.~{Verstraete}, J.~I.
  {Cirac}, {String Order and Symmetries in Quantum Spin Lattices}, Physical
  Review Letters 100~(16) (2008) 167202.
\newblock \href {http://arxiv.org/abs/0802.0447} {\path{arXiv:0802.0447}},
  \href {http://dx.doi.org/10.1103/PhysRevLett.100.167202}
  {\path{doi:10.1103/PhysRevLett.100.167202}}.

\bibitem{HallLieGroup}
B.~C. Hall, Lie groups, Lie algebras and representations , an elementary
  introduction, 2nd Edition, Springer, 2015.

\bibitem{FollandHarmAna}
G.~B. Folland, A Course in Abstract Harmonic Analysis, CRC Press, 1994, theorem
  7.25.

\bibitem{BultinckVerstraete}
N.~{Bultinck}, D.~J. {Williamson}, J.~{Haegeman}, F.~{Verstraete}, {Fermionic
  projected entangled-pair states and topological phases}, ArXiv e-prints\href
  {http://arxiv.org/abs/1707.00470} {\path{arXiv:1707.00470}}.

\bibitem{Bultinck:2017orh}
N.~Bultinck, D.~J. Williamson, J.~Haegeman, F.~Verstraete, {Fermionic matrix
  product states and one-dimensional topological phases}, Phys. Rev. B95~(7)
  (2017) 075108.
\newblock \href {http://arxiv.org/abs/1610.07849} {\path{arXiv:1610.07849}},
  \href {http://dx.doi.org/10.1103/PhysRevB.95.075108}
  {\path{doi:10.1103/PhysRevB.95.075108}}.

\end{thebibliography}
\bibliographystyle{elsarticle-num}

\end{document}